\documentclass[12pt]{article}

\usepackage{graphicx}
\usepackage{amsthm}
\usepackage{amsmath}
\usepackage{amssymb}
\usepackage{appendix}
\usepackage{color}
\usepackage{sgame}
\usepackage{verbatim}
\usepackage{caption}
\usepackage{subcaption}
\usepackage{mathtools}
\usepackage{dsfont}
\usepackage{float}
\usepackage[margin=1in]{geometry}
\usepackage[round]{natbib}
\usepackage[utf8]{inputenc}
\usepackage[english]{babel}
\usepackage{lscape}
\usepackage[colorlinks=true,allcolors=blue]{hyperref}
\usepackage{booktabs}
\usepackage{longtable}
\usepackage{multirow}
\usepackage{threeparttable}
\usepackage{dcolumn}
\usepackage{adjustbox}
\usepackage{tikz}
\usepackage{xcolor}
\usepackage{istgame}

\newcommand{\verteq}{\rotatebox{90}{$\;\;=\;\;$}}

\newcommand{\equalto}[2]{\underset{\normalsize\overset{\mkern4mu\verteq}{#2}}{#1}}

\usetikzlibrary{calc}
\newtheorem{theorem}{Theorem}

\newtheorem{assumption}{Assumption}

\newtheorem{corollary}{Corollary}

\newtheorem{definition}{Definition}

\newtheorem{lemma}{Lemma}

\newtheorem{proposition}{Proposition}

\newtheorem{hypothesis}{Hypothesis}
\newtheorem{remark}[theorem]{Remark}

\begin{document}

\title{Cognitive Hierarchies in Multi-Stage Games of Incomplete Information: Theory and 
Experiment\thanks{{\scriptsize This paper was previously circulated under the title ``Cognitive Hierarchies in 
Multi-Stage Games of Incomplete Information.''
I am indebted to my advisor, Thomas Palfrey, for 
his constant support and encouragement. I thank Marina Agranov, Pierpaolo Battigalli,
Colin Camerer, John Duffy, Michael Gibilisco, Charles Holt,
Paul J. Healy,
Rosemarie Nagel, Kirby Nielsen,  Jean-Laurent Rosenthal, Charles Sprenger,
Omer Tamuz, Joseph Tao-yi Wang, Yi Xin
and audiences at the University of Houston, California Institute of Technology, 
Stanford University, the 2022 ESA North American Conference,
the 2023 ESA North American Conference and the 
2023 Los Angeles Experiments (LAX) Workshop
for their valuable comments. I also thank John Duffy and Michael McBride
for their generosity and support of my use of the Experimental Social Science Laboratory
at UC Irvine. The experiment is approved by Caltech IRB \#22-1256.
This work is supported by NSF Dissertation Improvement Grant \#2243268. 
All errors are my own.}
}}
\author{Po-Hsuan Lin\thanks{{\scriptsize Division of the Humanities and Social Sciences, California Institute of
Technology, Pasadena, CA 91125 USA. Email: plin@caltech.edu}}
}

\date{November 2, 2023}

\maketitle


\begin{abstract}

Sequential equilibrium is the conventional approach for analyzing multi-stage games of incomplete information. It 
relies on mutual consistency of beliefs. 
To relax mutual consistency, I theoretically and 
experimentally explore the 
dynamic cognitive hierarchy (DCH) solution. 
One property of DCH is that the solution can 
vary between two different games sharing the same reduced normal form, i.e., 
violation of invariance under strategic equivalence. 
I test this prediction in a laboratory experiment using two strategically equivalent 
versions of the dirty-faces game. The game parameters are calibrated to maximize the expected 
difference in behavior between the two versions, as predicted by DCH.
The experimental results indicate significant differences in 
behavior between the two versions, and 
more importantly, the observed differences align with DCH.
This suggests that implementing a dynamic game experiment in reduced 
normal form (using the ``strategy method'') could lead to distortions in behavior.


\end{abstract}

{\footnotesize

}

\thispagestyle{empty}

\newpage \setcounter{page}{1} \renewcommand{\baselinestretch}{1.2}



\section{Introduction}\label{sec:intro}

Multi-stage games of incomplete information are important 
workhorse models in economics, political science, finance, social networks, and even biology.
These games tend to be more intricate than games of perfect information 
due to the potentially large number of information sets, regardless of how simple the game rules 
are.\footnote{For example, \cite{johanson2013measuring} estimates
that in a two-person Texas Hold’em game, the number of information sets is around $10^{162}$,
which is $10^{82}$ times larger than the number of atoms in the universe.}
The standard approach for analyzing these games is to solve for the sequential 
equilibrium, wherein players are assumed to 
form \emph{mutually consistent} beliefs at each information set. 
In other words, each player's conjecture about the behavioral strategies of 
others aligns with the actual strategies of those other players.

The assumption of mutual consistency of the belief system is crucial in standard equilibrium theory, 
as it, along with the best response requirement, pins down precise predictions of equilibrium outcomes.
However, this requirement may be implausibly strong from an empirical standpoint, especially for 
complicated multi-stage games of incomplete information, as indicated 
by behavior observed in many laboratory experiments (e.g., \citealp{camerer2011behavioral}).

In response to these findings, I develop a new tool for 
analyzing multi-stage games of incomplete information without relying 
on mutual consistency of beliefs: the 
``Dynamic Cognitive Hierarchy (DCH) Solution.'' 
The contribution of this paper encompasses both theoretical and experimental aspects.
Theoretical contributions involve 
extending the DCH solution from games of perfect information, as 
characterized by \cite{lin2022cognitive}, to multi-stage games of 
incomplete information. On the experimental front, I design and 
conduct a laboratory experiment to test a key implication of 
DCH---\emph{the violation of invariance under strategic 
equivalence}\footnote{The violation of invariance under strategic 
equivalence is sometimes referred to as ``the representation 
effect'' in the literature.}---in 
the context of the dirty-faces game,
a classic game for studying iterative rationality.



Two extensive games are strategically equivalent if they share the same reduced normal form.
It has been argued that any good equilibrium should exhibit invariance in strategically equivalent games
(e.g., \citealp{kohlberg1986strategic}). This invariance property is also appealing from the standpoint of experimental 
design as it suggests that implementing a dynamic game with its reduced normal form does 
not distort behavior. By doing so, one can gather more experimental data, 
particularly at information sets that are only occasionally reached. 
This approach to experimental design is commonly referred to as the ``strategy method'' \citep{selten1967strategiemethode}.

In experimental methodology, an ongoing debate surrounds whether the 
use of the strategy method distorts behavior 
\citep{brandts2011strategy}. DCH sheds light on this debate by 
indicating a potential violation of invariance under strategic 
equivalence, suggesting that the strategy method could theoretically 
create distortions in behavior. To empirically test this prediction, 
I implement a dirty-faces game experiment, which consists of two 
treatments: the sequential and the simultaneous treatments.
In the sequential treatment, the game is played period-by-period. Players can observe the history
and are asked to make decisions at realized information sets. In contrast, in the simultaneous treatment,
the game is played in reduced normal form, where players choose their contingent strategies.


Furthermore, the game parameters used in the experiment 
are selected using an ``optimal design approach,'' where  I first 
calibrate DCH using data from previous dirty-faces game experiments reported by other studies, and 
then select the game parameters to maximize the expected treatment effect. 
The utilization of an optimal design approach offers a systematic method for 
experimenters to choose game parameters, which is not commonly applied in economic experiments. 
To some extent, the experimental design in this paper serves as a 
proof-of-concept illustrating how this approach can help experimenters in designing future
theory-testing experiments.

By employing this fine-tuned
experimental design, significant treatment effects involving the 
violations of invariance
are detected in the data. Both the direction and magnitudes of
the observed differences align with the predictions of DCH.
Furthermore, when comparing DCH with alternative behavioral solution concepts that relax 
the best response requirement and the ability to make Bayesian inferences, 
we find that DCH significantly outperforms other models in both treatments. 
While there is evidence of the failure of best responses and Bayesian inferences, 
the observed violation of invariance in the data is primarily attributed to the relaxation of mutual consistency.



To offer readers a better intuitive understanding of the paper, 
I will next provide an overview of the DCH solution, illustrate its application
in the dirty-faces game, and discuss the experimental design and findings.

\bigskip

\noindent\textbf{Overview of the DCH Solution}

\bigskip

The DCH solution is akin to the ``level-$k$ model,'' a non-equilibrium framework introduced 
by \cite{nagel1995unraveling}. That model relaxes the mutual consistency requirement
in simultaneous-move games by assuming a hierarchical structure of strategic sophistication among the players.
In the level-$k$ model, each player is endowed with a specific level of sophistication. 
Level 0 players are non-strategic and choose their actions randomly. 
Level $k$ players, on the other hand, incorrectly believe that all other players are level $k-1$ and best respond to this 
belief.\footnote{In this paper, 
level 0 players may be interchangeably referred to as ``non-strategic players,'' 
and level $k\geq 1$ players as ``strategic players'' since they best respond to their beliefs.}

The level-$k$ model has been widely applied to organize experimental data in simultaneous-move games 
like the beauty contest game, coordination games, sender-receiver games, auction games, and more. Nevertheless, 
when applying the standard level-$k$ model to dynamic games, a logical conundrum arises: level $k$ players are 
assumed to choose actions that maximize the continuation value of the game, while believing that all other 
players are level $k-1$ in the continuation game. Consequently, each player's belief about the 
levels of others remains fixed from the beginning, potentially leading to situations where level $k$ 
players are ``surprised'' by an opponent's move that contradicts the strategy of a level $k-1$ player.

To illustrate this conundrum, consider the extensive game shown in Figure \ref{fig:illustation_game}. 
Suppose the level-$k$ model predicts that level 1 player 1 will choose $A$, while level 2 player 1 will choose $B$.
Since level 3 player 2 thinks player 1 is level 2, he believes that player 1 will certainly choose $B$. 
From the perspective of level 3 player 2, $A$ is an ``off-path event''
of player 1. If $A$ is chosen, level 3 player 2's belief about player 1's level is incompatible 
with the history.

\begin{figure}[htbp!]
\centering
\begin{tikzpicture}[scale=1.3,font=\footnotesize]
\tikzstyle{solid node}=[circle,draw,inner sep=1.5,fill=black]
\tikzstyle{hollow node}=[circle,draw,inner sep=1.5]
\tikzstyle{level 1}=[level distance=15mm,sibling distance=3.5cm]
\tikzstyle{level 2}=[level distance=15mm,sibling distance=1.5cm]
\tikzstyle{level 3}=[level distance=15mm,sibling distance=1cm]
\node(0)[solid node,label=above:{$1$}]{}
child{node[solid node,label=above left:{$2$}]{}
child{node[hollow node]{} edge from parent node[left]{$L$}}
child{node[hollow node]{} edge from parent node[right]{$R$}}
edge from parent node[left,xshift=-5]{$A$: {\color{blue}Level 1}}
}
child{node[solid node,label=above right:{$2$}]{}
child{node[hollow node]{} edge from parent node[left]{$L$}}
child{node[hollow node]{} edge from parent node[right]{$R$}}
edge from parent node[right,xshift=5]{$B$: {\color{red}Level 2}}
};
\end{tikzpicture}

\caption{An illustration game for the level-$k$ model in extensive games.}
\label{fig:illustation_game}
\end{figure}

To avoid this issue, 
the DCH solution assumes that level $k$ players believe all other players have lower 
levels distributed anywhere from level 0 to $k-1$, and update their beliefs about others’ levels as the history unfolds. 
Specifically, suppose each player's level is drawn from the distribution $p = \left(p_k \right)_{k=0}^{\infty}$.
Level 0 players uniformly randomize at every information set.
Level $k$ players' prior belief about any other player being level $j<k$ is $p_j / \sum_{l=0}^{k-1} p_l$,
which follows the Cognitive Hierarchy (CH) specification proposed by \cite{camerer2004cognitive}. 
Level $k$ players correctly perceive lower-level players' behavioral strategies 
and update their beliefs about the other players' levels using Bayes' rule as the game progresses.

The previous illustrative example clearly demonstrates how DCH solves the conundrum. Suppose DCH level 1 
player 1 chooses $A$, while DCH level 2 player 1 selects $B$. In contrast to the level-$k$ model, DCH 
assumes that level 3 player 2's prior belief about player 1 being level $j=0, 1, 2$ is $p_j / \sum_{l=0}^2 p_l$.
Moreover, after observing player 1's choice of $A$, level 3 player 2 eliminates the possibility of player 1 being
level 2 (otherwise, $B$ would have been chosen), and the beliefs about player 1
being level 0 and level 1 are $\frac{0.5p_0}{0.5p_0 + p_1}$ and $\frac{p_1}{0.5p_0 + p_1}$, respectively.
Because level 0 players uniformly randomize at every information set, strategic players' 
beliefs are always well-defined, effectively resolving the conundrum.

When extending DCH from games of perfect information 
to multi-stage games of incomplete information, 
players will update their beliefs about others' levels 
\emph{and} payoff-relevant private types at the same time. That is, 
the DCH belief system is a joint measure about the types and levels of other players.


Proposition \ref{prop:independence} establishes that if the private types are independently drawn across players, 
every level of every player's posterior belief remains independent across players at every information set.
It's important to note that this property of DCH holds only in multi-stage games with 
observed actions\footnote{Games of perfect information belong to the framework of 
multi-stage games with observed actions.} but not in general extensive games. 
Furthermore, when private types are correlated across players, Proposition
\ref{prop:corr_type} demonstrates that the original game (with correlated types)
can be transformed into another game with independent types, 
and the DCH behavioral strategy profiles remain invariant in both games.
These two propositions provide the recipe for solving 
DCH in multi-stage games of incomplete information.
Lastly, because level 0 players uniformly 
randomize at every information set, when the 
sizes of action sets differ, level 0 players' behavioral 
strategies might not be outcome-equivalent. This, in turn, affects
all higher-level players, as DCH is solved recursively from the 
bottom of the hierarchy, causing DCH solutions
to differ across strategically equivalent games.
This violation of invariance under strategic equivalence
is then illustrated in the dirty-faces game.

\bigskip
\noindent\textbf{The Dirty-Faces Game}
\bigskip


The dirty-faces game is a diagnostic game to study iterative rationality. 
It was originally introduced as a mathematical puzzle by 
(\citealt{littlewood1953littlewood}, pp. 3-4):
\begin{quote}
Three ladies, A,B,C, in a railway carriage all have dirty faces and are all 
laughing. It suddenly flashes on A: why doesn't B realize C is laughing
at her? Heavens! I must be laughable.
\end{quote}
This game is frequently discussed in understanding iterated reasoning 
and plays a central role in many studies of common knowledge (see, for example, 
\cite{binmore1988common,fudenberg1993game}, and \citealt{geanakoplos1994common}). 
To illustrate the violation of invariance under strategic 
equivalence, I consider two versions of the game: the 
\emph{sequential} and  \emph{simultaneous} versions.  


The \emph{sequential} (two-person) dirty-faces game was previously studied experimentally 
by \cite{weber2001behavior} and \cite{bayer2007dirty}.\footnote{\cite{weber2001behavior} and \cite{bayer2007dirty}
studied both two-person and three-person games. Since the experimental focus of this paper is on the two-person game, 
I will primarily focus on the discussion of two-person games in this paper 
and provide the discussion of three-person games in Appendix \ref{sec:dirty_exp_appendix}.} 
In this game, each player is randomly assigned a face type: either ``dirty'' or ``clean.'' Once the face
types are determined, players can see the other player's face, but not their own. Additionally, if at least one player has
a dirty face, a public announcement is made, ensuring that this information becomes common knowledge. After
observing the other's face and the announcement, players take actions in a series of periods. In each period, players
simultaneously choose between ``wait'' and ``claim'' (to have a dirty face). The actions are revealed to both players at
the end of each period. If both players decide to wait, the game proceeds to the next period. Otherwise, the game
ends after the period in which at least one player chooses to claim. Players receive rewards for correctly claiming to
have a dirty face but are penalized for making false claims.

The second strategically equivalent version is the \emph{simultaneous} (two-person) dirty-faces game, 
which, to the best of my knowledge, has never been studied experimentally before. 
In this version, the information structure remains identical to the sequential version.  
The only difference is that 
after observing the other's face and announcement, both players simultaneously 
decide the earliest period to claim as if the game were played in the sequential version. 
The payoffs are then determined accordingly.

Since the two versions are strategically equivalent, the sequential equilibrium is outcome-equivalent. 
Consider two players, Ann and Bob. When a player, let's say Ann, sees a clean face (along with the
announcement), she should recognize that her own face is dirty and claim in period 1. Yet, if Ann sees a
dirty face and is rational, she will wait in period 1 since she has no knowledge of her own face type. In
equilibrium, Ann believes that Bob would have claimed in the first period if her own face were clean. If the game does
indeed reach the second period, Ann will realize that her face is dirty and claim.


In contrast, due to the presence of non-strategic level 0 players, 
the DCH solution lacks common knowledge of rationality, 
causing it to differ 
dramatically from the standard equilibrium solution. Furthermore, DCH 
makes distinct predictions between the two versions, 
demonstrating the violation of invariance under strategic equivalence.

To understand the intuition of DCH, we first focus 
on the sequential version. In period 1, all strategic players behave 
like rational players, claiming immediately upon seeing a clean face 
and the announcement, while waiting when they see a dirty face. In 
contrast, level 0 players choose randomly regardless of their 
observations. Unlike strategic players, the actions of level 0 
players convey no information about the true face types. Therefore, 
when observing a dirty face and the game proceeds beyond period 1, 
strategic players will believe that there are two possible situations:
\begin{enumerate}
    \item If the other player is level 0, then their action 
    is randomly determined and provides
    no information about my own face.
    \item If the other player is level $k\geq 1$, then my own face is 
    certainly dirty because the other player would have claimed in period 1 if my own face were clean. 
\end{enumerate}
Consequently, after period 1, a strategic player faces a dynamic tradeoff. If she waits
and the game proceeds to the next period, she will become more certain about having a 
dirty face because the other player is less likely to be a level 0 player. However, 
the risk of waiting is that the game might be randomly terminated (by a level 0 opponent)
and the payoff is further discounted 
due to impatience. As a result, the DCH solution is characterized by level-dependent 
stopping periods, which depend on the prior distribution of levels and the payoffs. 
Because lower-level players are more likely to believe the other is level 0,
they need to wait longer to become certain enough about having a dirty face.

The intuition of DCH in the simultaneous version is the same as in 
the sequential version. 
However, in the simultaneous version, the number 
of available strategies changes, causing level 0 players'
behavioral strategies to differ between the two 
versions. Consequently, all higher-level players 
behave differently between the two versions, 
demonstrating the violation of invariance under 
strategic equivalence in DCH.\footnote{Refer to Section 
\ref{sec:representation_effect} for a detailed discussion of the 
violation of invariance in dirty-faces games.}
Furthermore, the magnitude of the 
difference in behavior predicted by DCH depends on 
the game parameters. Specifically, there exists two disjoint sets of game
parameters, where strategic players tend to claim earlier when 
observing a dirty face in one set in the sequential version and later 
in the other set.



\bigskip
\noindent\textbf{Experimental Design and Findings}
\bigskip

To assess the violation of invariance unders 
strategic equivalence in the dirty-faces game, 
I design and conduct a laboratory experiment that manipulates the timing structures
(sequential vs. simultaneous) using a between-subject design. The main challenge in
designing the experiment, as suggested by DCH, is
the selection of game parameters. Specifically, the magnitude of the treatment effect
predicted by DCH depends on both the game parameters and 
the true distribution of levels, which is unknown before the experiment is conducted.

To address this, I develop an ``optimal design approach'' where I first estimate 
the distribution of levels using data from an experimental dirty-faces game reported 
by \cite{bayer2007dirty}. Then, I select game
parameters to maximize diagnosticity 
by considering a mix of parameters expected to yield various magnitudes of 
the treatment effect.\footnote{It is worth remarking that the 
optimal design approach is not referring to the one in the statistical literature, 
whose goal is to maximize the determinant of the information matrix. For more 
information on optimal design in the statistical literature
and its applications in risky lottery experiments, see \cite{moffatt2020experimetrics} Chapter 14, 
and \cite{bland2023optimizing}.}

This optimal design approach has two advantages. First, it provides experimenters with a
systematic method for selecting game parameters when designing experiments. This is 
important to experimenters since, as
noted by \cite{moffatt2020experimetrics}, when choosing the parameters, ``most 
experimenters have followed an informal approach'' (\citealt{moffatt2020experimetrics}, pp. 335).
Second, it serves as a stress test for DCH. After calibrating the 
prior distribution of levels, DCH offers precise predictions regarding the magnitudes 
of violations of invariance. Instead of fitting the model ex post, 
this approach provides a benchmark prediction before the experiment is conducted, 
enabling us to assess the predictive power of DCH.

A significant violation of invariance under strategic
equivalence is detected in the data, and more importantly,
\emph{how} it is violated aligns with DCH. 
Additionally, 
to analyze whether the difference found in the data can be 
attributed to the relaxation of other equilibrium requirements, 
I compare DCH with two alternative models: the Agent Quantal Response Equilibrium by \cite{mckelvey1998quantal} 
and the Cursed Sequential Equilibrium by \cite{fong2023cursed}. These alternatives relax 
the best response requirement and Bayesian inferences, respectively.
The estimation results indicate that for both treatments,
DCH fits the data significantly better than the other two solutions.
Although relaxing other requirements could improve the fitness, 
the observed difference is primarily attributed to the relaxation of mutual consistency.

This experimental result highlights how implementing a dynamic game 
experiment in reduced normal form (using the 
``strategy method'') can 
lead to significant distortions in behavior. 
In addition, DCH provides a better explanation for how these 
distortions arise compared to other behavioral solution concepts. 
This suggests that if using the strategy method is necessary, DCH can 
offer a more reasonable assessment of behavioral distortions compared 
to the natural approach, which allows subjects to make decisions when 
it's their turn, as if the game tree were 
being fully implemented.\footnote{This approach 
is commonly referred to as the ``direct-response method.''}

The paper is organized as follows. Section \ref{sec:lit_review} discusses 
the related literature. Section \ref{sec:model} sets up the model, and general properties of the 
DCH solution are established in Section \ref{sec:theo_prop}. In Section \ref{sec:dirty_face}, 
I demonstrate the violation of invariance under
strategic equivalence of DCH in a class of two-person dirty-faces 
games.  Section \ref{section:exp_design} describes the experimental design, and Section \ref{section:exp_result} 
reports the experimental results. Finally, Section \ref{sec:conclusion} concludes the paper.

\bigskip



\section{Related Literature}
\label{sec:lit_review}

The DCH solution is closely related to a number of behavioral models of games. Over the past thirty years,
the idea of limited depth of reasoning has been theoretically studied by various 
researchers, including \cite{selten1991anticipatory, selten1998features}, \cite{aumann1992irrationality}, 
\cite{stahl1993evolution}, \cite{alaoui2016endogenous, alaoui2018cost} and 
\cite{lin2022cognitive}. In 
addition to theoretical work, \cite{nagel1995unraveling} conducted the first experiment on the ``beauty 
contest game'' to study people's iterative reasoning process. 
In this game, each player simultaneously chooses an integer between 0 and 100. 
The winner is the player whose choice is closest to the average of all numbers multiplied by  
$p\in(0,1)$. The unique equilibrium predicts that all players should choose 0. However, empirical 
observations show that almost no player chooses the equilibrium action. Instead, players seem to behave as 
if they are performing some finite number of iterative best responses.\footnote{This empirical 
pattern has been robustly replicated in different environments. 
For instance, \cite{ho1998iterated} and \cite{bosch2002one} have found similar 
results in both laboratory and field experiments.}

To explain the data, \cite{nagel1995unraveling} proposed the ``level-$k$ model,'' 
which assumes that each player is endowed with a specific ``level'' of reasoning. 
Level 0 players randomly select actions from their action sets. 
For every $k\geq 1$, level $k$ players believe that they are one level of reasoning 
higher than the rest and best respond accordingly. 
The level-$k$ model has been applied to various environments, 
including simultaneous-move games \citep{costa2001cognition, crawford2007fatal}, 
two-person guessing games \citep{costa2006cognition}, auctions \citep{crawford2007level}, 
and sender-receiver games \citep{cai2006overcommunication, wang2010pinocchio}.

The standard level-$k$ model has been successful in explaining the data and has been extended in
various ways. One such approach is the CH framework proposed by
\cite{camerer2004cognitive}, which assumes that players best respond to a mixture of lower-level
players.\footnote{In a similar vein, \cite{stahl1995players} and \cite{levin2022bridging}
allow each level of players to best respond to a mixture of
lower-level players and equilibrium players.}
In this framework, level $k$ players best respond to a mixture of lower levels, ranging from level 0 to $k-1$.
Furthermore, players hold accurate beliefs about the relative proportions of the lower levels. 
However, this approach is primarily developed for simultaneous-move games, and DCH 
extends it to dynamic games.

Another direction is to endogenize the levels of players. 
\cite{alaoui2016endogenous} considered a 
cost-benefit analysis approach, where players decide their levels of sophistication 
by weighing the benefits of additional levels against the costs of doing so. 
The implications of this model were further explored in \cite{alaoui2020reasoning}.
In the same spirit as \cite{stahl1996boundedly}, \cite{ho2013dynamic} and 
\cite{ho2021bayesian} considered a canonical laboratory environment where 
players repeatedly play the same game and endogenously choose a new 
level of sophistication for the next iteration 
of the game. This approach is different from DCH, where players update their beliefs about other players' 
levels after each move within a single game. Furthermore, DCH players are strategic learners as they can 
correctly anticipate the evolution of posterior beliefs in later information sets. This leads to a much 
different learning dynamic compared to naive adaptive learning models.

At a more conceptual level, the DCH solution is related to other solution concepts
for dynamic games that relax the requirements of sequential equilibrium.
DCH is a non-equilibrium model that allows players at different levels to best respond to
different conjectures about other players' strategies, while Agent Quantal
Response Equilibrium (AQRE) by \cite{mckelvey1998quantal} is an equilibrium model
in which players make stochastic choices. Both DCH and AQRE assume that players follow Bayes'
rule to make inferences. In contrast, Cursed Sequential Equilibrium (CSE) by
\cite{fong2023cursed} and Sequential Cursed Equilibrium (SCE) by \cite{cohen2022sequential}
are two different equilibrium models in which players are
able to make best responses but are unable to make correct Bayesian inferences.\footnote{The CSE
proposed by \cite{fong2023cursed} captures the situation where 
players fail to understand how other players' actions depend on their own private information. 
On the other hand, the SCE introduced by \cite{cohen2022sequential} 
depicts the bias where people fail to realize how others' action depend on their 
information set partitions. See \cite{fong2023note} for a detailed comparison of the two
solution concepts.}

One common theoretical property of these behavioral solution concepts is the violation 
of invariance under strategic equivalence, albeit in different ways. 
However, whether implementing the same dynamic game with different 
methods creates any behavioral distortion is an empirical 
question. \cite{brandts2011strategy} surveyed 29 experiments that compared
the behavior under different elicitation methods and found that the
invariance may be violated under certain conditions.\footnote{In 
their survey, they observed no difference in 16 studies, systematic 
differences in four studies, and mixed evidence in nine of them. 
In particular, they found suggestive evidence that the frequency of
violation of invariance is related to the number of
available actions.} More recently, \cite{li2017obviously} compared 
the second-price auction and the ascending clock auction and found that 
players are more likely to follow dominant 
strategies in the ascending clock auction.
Additionally, \cite{garcia2020hot}
experimentally studied the invariance in four centipede games and reported
that in three of the four games, players tend to terminate the game earlier 
in the sequential version of the game, which is consistent with DCH 
\citep{lin2022cognitive}.\footnote{In the three centipede games 
where termination occurs earlier in the sequential version, DCH predicts that 
the distribution of terminal nodes from the simultaneous version (strategy method) 
will first-order stochastically dominate the distribution from the sequential 
version (direct response method). In the fourth centipede game where FOSD 
is not predicted, the empirical distributions of terminal nodes 
from the two versions are almost identical.}
Finally, it is worth noting that \cite{chen2023invariance, chen2023invariance2}
point out that the violation of invariance is linked to 
the emotional salience induced during the experiment.


This paper also contributes to the literature on dirty-faces games. 
The concept of dirty-faces game was originally introduced by \cite{littlewood1953littlewood} 
as a means to illustrate the transmission of common knowledge. 
\cite{binmore1988common}, \cite{fudenberg1993game}, and \cite{geanakoplos1994common}
were the first to theoretically study the dirty faces games 
with the knowledge operator.
Furthermore, \cite{liu2008dirty} demonstrated that if players are unaware 
of other players' face types, they might incorrectly claim their face types, and hence 
influence the transmission of knowledge among the players.

\cite{weber2001behavior} and \cite{bayer2007dirty}
conducted the first two experiments on dirty-faces games and found that many
subjects fail to perform such iterative reasoning.
More recent experiments have further demonstrated the persistence of failure in
iterative reasoning even when playing against fully rational robot players
\citep{grehl2015experimental, chen2023measure}. This failure has also been found to be correlated with cognitive abilities
\citep{devetag2003games, bayer2016logical, bayer2016logical2}, while the deviations from the 
equilibrium significantly decrease when the participants are selected 
through a market mechanism \citep{choo2022can}.
Overall, these experimental findings provide support for the existence of non-strategic types of 
players who are not sequentially rational,
highlighting the heterogeneity in strategic sophistication within the population.


\section{The Model}\label{sec:model}

Section \ref{subsection:game_def} introduces the multi-stage games with observed actions,
as proposed by \cite{FUDENBERG1983251} and \cite{fudenberg1991perfect}.
This framework provides a tractable approach to studying how players 
learn about the types and levels of others.
Next, the DCH solution for this family of games is defined in section \ref{subsection:dch_def}.

\subsection{Multi-Stage Games with Observed Actions}
\label{subsection:game_def}

Let $N=\{1,\ldots ,n\}$ be a finite set of players.
Each player $i\in N$ has a \textit{type} $\theta_i$
drawn from a finite set $\Theta_i$.
Let $\theta \in \Theta \equiv \times_{i=1}^{n}\Theta_i$ 
be the type profile and $\theta_{-i}$ be the type profile
without player $i$.
All players have the common (full support) prior distribution 
$\mathcal{F}: \Theta \rightarrow (0,1)$.
At the beginning of the game, each player is told his own type,
but is not informed anything about the types of others.
Therefore, each player $i$'s initial belief about the types of others 
when his type is $\theta_i$ is: 
$$\mathcal{F}(\theta_{-i} | \theta_i) = \frac{\mathcal{F}(\theta_{-i},\theta_i)}{\sum_{\theta'_{-i}\in \Theta_{-i}} \mathcal{F}(\theta'_{-i}, \theta_i)}.$$
If the types are independent across players, 
each player $i$'s initial belief about the types of others is 
$\mathcal{F}_{-i}(\theta_{-i}) = \Pi_{j\neq i}\mathcal{F}_{j}(\theta_{j})$ where
$\mathcal{F}_{j}(\theta_j)$ is the marginal distribution of player 
$j$'s type.

The game is played in ``periods'' $t=1,2,\ldots, T$ where $T<\infty$.
In each period, 
players simultaneously choose their actions, which will be 
revealed at the end of the period. 
The feasible set of actions can vary with histories, 
so games with alternating moves are also included.
Let $\mathcal{H}^{t-1}$ be the set 
of all available histories at period $t$, 
where $\mathcal{H}^0 = \{h_{\emptyset} \}$ and 
$\mathcal{H}^T$ is the set of terminal histories.
Let $\mathcal{H} = \cup_{t=0}^{T} \mathcal{H}^t$ be the 
set of all available histories of the game, 
and let $\mathcal{H}\backslash\mathcal{H}^T$ be the set of 
non-terminal histories.

For every player $i$, the available information 
at period $t$ is in $\Theta_{i} \times \mathcal{H}^{t-1}$.
Therefore, each player $i$'s information sets can be 
specified as $\mathcal{I}_i \in \Pi_{i} = \{(\theta_i, h): \theta_i \in \Theta_i, 
h\in \mathcal{H}\backslash\mathcal{H}^{T}\}$.
For the sake of simplicity, the feasible set of actions for every player at 
every history is assumed to be type-independent. 
Let $A_i(h^{t-1})$ be the feasible set of actions for player $i$ at history $h^{t-1}$ and let
$A_i = \cup_{h\in \mathcal{H}\backslash\mathcal{H}^T}A_i(h)$ be the set of player $i$'s 
all feasible actions in the game. 
For each player $i$, $A_i$ is assumed to be finite and 
$|A_i(h)|\geq 1$ for any $h\in \mathcal{H}\backslash\mathcal{H}^T$.
Let $a_i^t\in A_i(h^{t-1})$ be player $i$'s action at history $h^{t-1}$, and 
let $a^t=\left(a_1^t, \ldots, a_{n}^t \right)\in \times_{i=1}^n A_i(h^{t-1})$
denote the action profile at period $t$. 
If $a^t$ is the action profile chosen at period $t$, 
then $h^t = (h^{t-1}, a^t)$.

A behavioral strategy for player $i$ is a function 
$\sigma_i: \Pi_i \rightarrow \Delta(A_i)$ satisfying 
$\sigma_i(\theta_i, h^{t-1}) \in \Delta(A_i(h^{t-1}))$.
Let $\sigma_i(a_i^t \mid \theta_i, h^{t-1})$ 
denote the probability for player $i$ to choose $a_i^t \in A_i(h^{t-1})$.
A strategy profile $\sigma=\left(\sigma_i\right)_{i\in N}$ specifies a behavioral strategy 
for each player $i$. Lastly, each player $i$ has a payoff function 
(in von Neumann-Morgenstern utilities)
$u_i: \mathcal{H}^{T} \times \Theta \rightarrow \mathbb{R},$ and let $u=(u_1, \ldots, u_n)$ be 
the profile of utility functions.
A multi-stage game with observed actions, $\Gamma$, 
is defined by the tuple $\Gamma = \langle N, \mathcal{H}, \Theta, \mathcal{F}, u\rangle$.

\subsection{Dynamic Cognitive Hierarchy Solution}\label{subsection:dch_def}

Each player $i$ is endowed with a level of sophistication 
$\tau_i \in \mathbb{N}_0$ which is independently drawn from the distribution $P_i(\tau_i)$.
Without loss of generality, I assume $P_i(\tau_i)>0$ for all $i \in N$ and 
$\tau_i \in \mathbb{N}_0$. Let $\tau = (\tau_1,\ldots, \tau_n)$ be the level
profile and $\tau_{-i}$ be the level profile without player $i$.
Due to the independence, the level profile is drawn from a 
distribution $P: \mathbb{N}_0^{|N|} \rightarrow (0,1)$ such that $P(\tau) = \Pi_{i=1}^{n} P_i(\tau_i)$.
Following \cite{lin2022cognitive}, I assume that players' ``types'' and ``levels'' are drawn independently.

\begin{assumption}\label{assumption:level_type_independence}
$\mathcal{F}$ and $P$ are independent distributions. 
\end{assumption}

Each player $i$ has a prior belief about the opponents' levels which satisfies the property of 
\emph{truncated rational expectations.} That is, for each $i$, $j\neq i$, and $k$, 
let $\hat{P}_{ij}^k(\tau_j)$ be level $k$ player $i$'s
prior belief about player $j$'s level, and $\hat{P}_{ij}^k(\tau_j)$ satisfies:
\begin{equation}
\hat{P}_{ij}^{k}(\tau_{j})=%
\begin{cases}
\frac{P_j(\tau_j)}{\sum_{m=0}^{k-1}P_j(m)} & \mbox{if }\tau_j <k \\
0 & \mbox{if }\tau_j \geq k.%
\end{cases}
\label{truncated}
\end{equation}
The intuition of (\ref{truncated}) is that despite mistakenly believing all other players 
are at most level (\emph{k}-1),\footnote{The cognitive hierarchy specification is 
in line with behavioral and psychological evidence of overconfidence across 
various domains (see, for instance, \cite{camerer1999overconfidence, moore2008trouble} and \citealt{enke2023confidence}).
While recent findings by \cite{halevy2021magic} suggest the possibility of players 
believing others to be more sophisticated than themselves, this behavior 
falls beyond the scope of this paper, as DCH is developed within the 
confines of the standard CH framework.}
each level of players have a correct belief about the 
relative proportions of players who are less 
sophisticated than they are.

In the DCH solution, a strategy profile is a level-dependent 
profile of behavioral strategy of each level of each player. 
Let $\sigma_{i}^k$ be level $k$ player $i$'s behavioral strategy, 
where level 0 players uniformly randomize at every information set.\footnote{
Uniform randomization is not the only way to model level 0 players' behavior; however, one 
compelling justification for its use is its universal applicability to all games in the same manner. 
In fact, the DCH solution is well-defined as long as level 0 players' 
behavioral strategy is fully mixed at every information set.} 
That is, for every $i\in N$, $\theta_i \in \Theta_i$, $h\in \mathcal{H}\backslash\mathcal{H}^T$, and for all $a\in A_i(h)$,
$$\sigma_i^0(a \mid \theta_i, h) = \frac{1}{|A_i(h)|}.$$

At every history $h^t$, every strategic level $k$ player $i$ forms a \emph{joint belief}
about all other players' types and levels.\footnote{Level 1 players always believe other 
players are level 0 whose actions are uninformative about their types. Therefore, they don't 
update their beliefs about the levels and types of others. }
Their posterior beliefs at history $h^{t}$ depend on the 
level-dependent strategy profile and the prior beliefs. 
To formalize the belief updating process, let 
$\sigma_j^{-k} = (\sigma_j^{0},\ldots, \sigma_j^{k-1})$ be 
the profile of strategies adopted by the levels below $k$ of player $j$.
Furthermore, let $\sigma_{-i}^{-k} = (\sigma_{1}^{-k},\ldots, \sigma_{i-1}^{-k}, \sigma_{i+1}^{-k}, \ldots, \sigma_{n}^{-k})$
be the profile of behavioral strategies of the levels below $k$ 
of all players other than player $i$.
It is worth noticing that all strategic players believe every history is possible 
because $\hat{P}_{ij}^{k}(0)>0$ for all $i,j\in N$ and $k>0$, and 
$\sigma_j^0(a \mid \theta_j, h)>0$ for all $j$, $\theta_j$, $h$ and $a\in A_j(h)$.
Consequently, Bayes' rule can be applied to derive every level of players'
posterior belief about other players' types and levels. 
Specifically, for any $i\in N$, $k\geq 1$ and $\theta_i\in \Theta_i$, 
a level-dependent strategy profile will induce the posterior belief
$\mu_i^k(\theta_{-i}, \tau_{-i}\mid \theta_i, h)$ at every $h\in \mathcal{H}\backslash\mathcal{H}^T$
with $\mu_i^{k}(\theta_{-i} \mid \theta_i, h^{t-1})$ and
$\mu_{i}^{k}(\tau_{-i} \mid \theta_i, h^{t-1})$ 
being level $k$ player $i$'s marginal beliefs 
of other players' types and levels at history $h^{t-1}$, respectively.
Lastly, for any $j\neq i$, let $\mu_{ij}^k(\theta_j, \tau_j \mid \theta_i, h^{t-1})$ denote
level $k$ player $i$'s belief about player $j$'s type and level at history $h^{t-1}$.

In the DCH solution, players correctly anticipate how they will 
update their posterior beliefs at all future histories of the game, i.e., 
players are strategic learners. 
Therefore, for any $i$, $k$, $\theta_i$ and any  
level-dependent strategy-profile of others $\sigma_{-i}^{-k}$,
type $\theta_i$ level $k$ player $i$ 
believes the probability of $a^t_{-i}\in A_{-i}(h^{t-1})$ being chosen is 
$$\tilde{\sigma}_{-i}^{-k}(a_{-i}^t \mid \theta_i, h^{t-1}) \equiv \sum_{\theta_{-i}\in \Theta_{-i}}
\sum_{\{\tau_{-i}: \tau_j <k \; \forall j\neq i\}}
\mu_{i}^k(\theta_{-i}, \tau_{-i} \mid \theta_i, h^{t-1}) 
\prod_{j\neq i}\sigma_j^{\tau_j}(a_j^t \mid \theta_j, h^{t-1}).$$
Furthermore, for every level of players, given lower-level players' strategies, 
they can compute the probability of any outcome 
being realized at any non-terminal history. 
In particular, for any $i\in N$, $\tau_i>0$, $\theta\in\Theta$, 
$\sigma$, and $\tau_{-i}$ such that
$\tau_j<\tau_i$ for any $j\neq i$, 
let $P_i^{\tau_i}(h^T | \theta, h^{t-1}, \tau_{-i},
\sigma_{-i}^{-\tau_i},\sigma_i^{\tau_i})$ be level $\tau_i$ player $i$'s belief about the conditional realization probability of $h^T\in\mathcal{H}^T$
at history $h^{t-1}\in \mathcal{H}\backslash
\mathcal{H}^T$ if the type profile is $\theta$,
the level profile is $\tau$, and player $i$ uses 
$\sigma_i^{\tau_i}$. Finally, level $\tau_i$ player $i$'s expected
payoff at any $h^t\in \mathcal{H}\backslash
\mathcal{H}^T$ is:
\begin{align}\label{eq:dch_maximization}
\mathbb{E}&u_i^{\tau_i}(\sigma \mid \theta_i,h^t)=  \nonumber \\   
&\; \sum_{h^T\in\mathcal{H}^T} \sum_{\theta_{-i}\in \Theta_{-i}}
\sum_{\{\tau_{-i}: \tau_j <k \; \forall j\neq i\}}\mu_i^{\tau_i}(\theta_{-i},\tau_{-i} \mid
\theta_i, h^t) P_i^{\tau_i}(h^T|\theta, h^{t}, \tau_{-i},
\sigma_{-i}^{-\tau_i},\sigma_i^{\tau_i}) 
u_i(h^T,\theta_i,\theta_{-i}).
\end{align}

The \emph{DCH solution} of the game is defined as the 
level-dependent assessment $(\sigma^*, \mu^*)$,
such that $\sigma_i^{k*}(\cdot| \theta_i, h^t)$ maximizes (\ref{eq:dch_maximization}) for 
all $i$, $k$, $\theta_i$ and $h^t\in \mathcal{H}\backslash\mathcal{H}^T$
and the \emph{DCH belief system} $\mu^*$ is induced by $\sigma^*$. Moreover, 
players are assumed to 
uniformly randomize over optimal actions when they are indifferent. 
This is a typical assumption in level-$k$ models, and it is convenient because it ensures a unique 
DCH solution.

\begin{lemma}\label{lemma_uniqueness}
    The DCH solution is unique. 
\end{lemma}

\begin{proof}
See Appendix \ref{appendix:proof_general}.
\end{proof}

\begin{remark}
For one-stage games, the DCH solution reduces to the standard CH solution because one-stage games are 
essentially static games.
\end{remark}

\section{General Properties of the DCH Solution}
\label{sec:theo_prop}

In this section, I first characterize some general properties of the belief 
updating process of DCH.
Assume for now that players' types are independently drawn, i.e., 
$\mathcal{F}(\theta)=\prod_{i\in N} \mathcal{F}_i(\theta_i)$.
With this assumption, Proposition \ref{prop:independence} shows that 
at every information set, the posterior beliefs are independent 
across players. In other words, the DCH belief system is a product measure.

\begin{proposition}\label{prop:independence}
For any multi-stage game with observed actions $\Gamma$, 
any $h\in \mathcal{H}\backslash\mathcal{H}^T$, any $i\in N$, $\theta_i\in \Theta_i$,
and for any $k\in \mathbb{N}$, if the prior distribution of types 
is independent across players, i.e., $\mathcal{F}(\theta) = \prod_{i=1}
^n \mathcal{F}_i(\theta_i)$, then 
level $k$ player $i$'s posterior 
belief about other players' types and levels at $h$ is 
independent across players. That is,
$$\mu_i^k(\theta_{-i}, \tau_{-i}|\theta_i, h) = 
\prod_{j\neq i}\mu_{ij}^k(\theta_j,\tau_j|\theta_i, h). $$
\end{proposition}

\begin{proof}
See Appendix \ref{appendix:proof_general}.
\end{proof}

Proposition \ref{prop:independence} extends the independence property 
shown by \cite{lin2022cognitive} from games of perfect information to 
multi-stage games with observed actions. This generalization relies on that 
(1) the actions are perfectly observed and (2) players are able to 
perform Bayesian inferences.

In multi-stage games with observed actions, as the prior distribution of 
types and levels is independent across players,\footnote{The independence
property does not rely on Assumption 
\ref{assumption:level_type_independence}. It 
holds as long as the priors distributions 
of types and levels are both independent across players.} because
every player's action is perfectly monitored, strategic players 
understand that each player's action does not convey any 
information about other players'
private information. In this case, Proposition \ref{prop:independence}
shows the belief system remains to be a product measure in 
any information set. However, as pointed out by \cite{lin2022cognitive},
this is not true for general dynamic
games of imperfect information. When the actions are not perfectly observed, 
the marginal beliefs about others' levels could 
be correlated across players (see section 7.2 of \citealt{lin2022cognitive}).

Besides, the ability to perform Bayesian inferences plays a crucial role in 
maintaining the independence property. In other behavioral solution concepts, 
such as the Cursed Sequential Equilibrium proposed by \cite{fong2023cursed}, 
where players are unable to perform Bayesian inferences, players may mistakenly 
believe that others' actions are informative about another player's private 
information, even though the actions are perfectly observed and the prior 
distribution of types is independent across players.

It is worth noticing that the independence property is useful for 
solving the DCH solution. When there are more players or when the game 
structure becomes more complex, 
computing the posterior belief can become challenging, as it involves 
level-dependent probability measures. Yet, 
Proposition \ref{prop:independence} guarantees that the DCH belief system 
can be computed player-wise rather than information-set-wise, 
which simplifies the computation process.


Next, I consider the case where the prior distribution of types is not 
independent across players. When the types are correlated across players, 
their actions are informative about not only their own private 
information but also the private information of players 
whose types are correlated with them. 
Similar to the observations of \cite{myerson1985bayesian} 
and \cite{fudenberg1991perfect},
to deal with correlated types, the original game (with correlated types)
can be simply transformed into one game with independent types with a 
specific transformation.

For any multi-stage game with observed actions $\Gamma$, 
consider a corresponding transformed game $\hat{\Gamma}$ where
the prior distribution of types is the product of 
independent uniform marginal distributions. Namely,
$$\hat{\mathcal{F}}(\theta) = \frac{1}{\prod_{i=1}^n |\Theta_i|} \;\;\;\; \forall \theta\in \Theta.$$
In addition, the utility functions are transformed
to be
$$\hat{u}_i(h^T,\theta_i,\theta_{-i}) = 
\mathcal{F}(\theta_{-i}|\theta_i)u_i(h^T, \theta_i, \theta_{-i}).$$

Proposition \ref{prop:corr_type} shows that the DCH 
level-dependent behavioral strategy profile is invariant under 
the transformation between the transformed and original game, 
suggesting that the independence assumption of 
the types is without loss of generality. 
Moreover, it is surprising that the transformation is \emph{level-independent},
given that the best response of each level is determined iteratively. The intuition behind this is
that, due to the independence of types and levels (Assumption \ref{assumption:level_type_independence}), 
players cannot make inferences about others' types based on their knowledge of others' levels.

\begin{proposition}\label{prop:corr_type}
The level-dependent assessment $(\hat{\sigma}, 
\hat{\mu})$ is the DCH solution of the 
transformed game (with independent types) if and only if
the level-dependent assessment $(\sigma, 
\mu)$ is the DCH solution of the original game (with correlated types)
where $\sigma = \hat{\sigma}$ and for any $i\in N$, $\theta_i\in\Theta_i$, $k>0$, and $h^t\in \mathcal{H}\backslash\mathcal{H}^T$,
\begin{align*}
\mu_i^k(\theta_{-i},\tau_{-i}| \theta_i, h^t)
=\frac{\mathcal{F}(\theta_{-i}|\theta_i)
\hat{\mu}_{i}^k(\theta_{-i},\tau_{-i}|\theta_i, h^t)}{\sum_{\theta'_{-i}}\sum_{\{\tau'_{-i}:\tau'_j<k \; \forall j\neq i\}}
\mathcal{F}(\theta'_{-i}|\theta_i)
\hat{\mu}_{i}^k(\theta'_{-i},\tau'_{-i}|\theta_i, h^t)}.
\end{align*}
\end{proposition}

\begin{proof}
See Appendix \ref{appendix:proof_general}.
\end{proof}

Proposition \ref{prop:corr_type} shows that at any information set $\mathcal{I}_i 
= (\theta_i, h^t)$, player $i$'s belief of a specific type-level profile $(\theta_{-i}, \tau_{-i})$ is 
proportional to prior belief of $\theta_{-i}$ conditional on $\theta_i$.
In addition, if $\mathcal{F}(\theta_{-i}|\theta_i)\rightarrow 1$, i.e., $\theta_i$ is almost perfectly 
correlated with $\theta_{-i}$, then 
\begin{align*}
\mu_i^k(\theta_{-i},\tau_{-i}| \theta_i, h^t)
=&\; \frac{\mathcal{F}(\theta_{-i}|\theta_i)
\hat{\mu}_{i}^k(\theta_{-i},\tau_{-i}|\theta_i, h^t)}{\sum_{\theta'_{-i}}\sum_{\{\tau'_{-i}:\tau'_j<k \; \forall j\neq i\}}
\mathcal{F}(\theta'_{-i}|\theta_i)
\hat{\mu}_{i}^k(\theta'_{-i},\tau'_{-i}|\theta_i, h^t)}\\
\rightarrow &\; \frac{\hat{\mu}_{i}^k(\theta_{-i},\tau_{-i}|\theta_i, h^t)}{\sum_{\{\tau'_{-i}:\tau'_j<k \; \forall j\neq i\}}
\hat{\mu}_{i}^k(\theta_{-i},\tau'_{-i}|\theta_i, h^t)}
= \hat{\mu}_{i}^k(\theta_{-i},\tau_{-i}|\theta_i, h^t),
\end{align*}
implying that the belief in the transformed game aligns with the belief in the original game. 
Intuitively speaking, if the types are almost perfectly correlated, 
then the remaining information to be learned is solely others' levels. 
Since the DCH behavioral strategy profile is invariant under the transformation, 
the belief about others' levels will also be invariant under the transformation.

The next property is about the evolution of the support of the beliefs.
\cite{lin2022cognitive} have shown that in games of perfect information, 
the support of beliefs about the levels weakly shrinks as the history unfolds.
Proposition \ref{prop:support_evo} extends this result to multi-stage games with observed actions, 
indicating that the support of the marginal beliefs about the levels weakly shrinks in later periods.
Additionally, Proposition \ref{prop:support_evo} demonstrates that the marginal beliefs 
about other players' types always maintain full support.
In other words, in the DCH solution, players will become more certain about the levels of others 
in later periods while never completely ruling out any possibility of a type profile. 
To formally state the proposition, I first define the support of the marginal beliefs.

\begin{definition}[Support]
For any multi-stage game with observed actions $\Gamma$, any $i\in N$, any $\tau_i\in \mathbb{N}$, any $\theta_i\in \Theta_i$, and any history $h\in \mathcal{H}\backslash
\mathcal{H}^{T}$, let $supp_i(\theta_{-i}|\tau_i, \theta_i, h)$ and 
$supp_i(\tau_{-i}|\tau_i, \theta_i, h)$
be the support of level $\tau_i$ player $i$'s marginal belief
about other players' types and levels at information set $(\theta_i, h)$, 
respectively. In other words, for any $\theta'_{-i}$ and $\tau'_{-i}$,
\begin{align*}
    \theta'_{-i}\in supp_i(\theta_{-i}|\tau_i, \theta_i, h) &\iff 
    \mu_i^{\tau_i}(\theta'_{-i} \mid \theta_i, h) > 0, \\
    \tau'_{-i} \in supp_i(\tau_{-i}|\tau_i, \theta_i, h) &\iff 
    \mu_i^{\tau_i}(\tau'_{-i} \mid \theta_i, h) > 0.
\end{align*}
\end{definition}

\begin{proposition}\label{prop:support_evo}
Consider any multi-stage game with observed actions $\Gamma$, any $i\in N$, 
any $\tau_i\in \mathbb{N}$, and any $\theta_i \in \Theta_i$. The following 
two statements hold.
\begin{itemize}
    \item[1.] For any $h^t = (h^{t-1}, a^t) \in \mathcal{H}^t \backslash \mathcal{H}^T$, 
    $supp_i(\tau_{-i}|\tau_i, \theta_i, h^{t}) \subseteq 
    supp_i(\tau_{-i}|\tau_i, \theta_i, h^{t-1})$.
    \item[2.] For any $h\in \mathcal{H}\backslash\mathcal{H}^T$,
    $supp_i(\theta_{-i}|\tau_i, \theta_i, h) = \Theta_{-i}$.
\end{itemize}
\end{proposition}

\begin{proof}
See Appendix \ref{appendix:proof_general}.
\end{proof}

The intuition behind Proposition \ref{prop:support_evo} is that since it is always possible 
for other players to be level 0, players can always rationalize any type profile by assuming 
all other players are level 0.\footnote{If the action sets vary not only with histories but 
also with types, players may rule out the possibility of certain type profiles when specific 
actions are chosen.} This implies that no matter how sophisticated the players are, common knowledge 
of rationality is never reached in DCH, which suggests that DCH and the 
equilibrium theory are fundamentally different solution concepts.

Furthermore, another feature of DCH that sharply contrasts with the equilibrium theory is
the violation of invariance under strategic equivalence. In DCH, since level 0 players uniformly
randomize at every information set, their behavioral strategies might not be 
outcome-equivalent if the cardinality of action sets changes. All higher-level players are subsequently affected, 
as DCH is solved recursively from the bottom of the hierarchy.\footnote{According to \cite{thompson1952equivalence}
and \cite{elmes1994strategic}, two extensive games share the same reduced normal form if and 
only if they can be transformed into each other using a 
small set of elementary transformations. Specifically, \cite{elmes1994strategic} propose 
three such transformations: INT, COA, and ADD, which preserve perfect recall. Because DCH is sensitive 
to the cardinality of action sets, it varies under COA while remaining invariant under INT and ADD.}
In the remaining part of the paper, I will demonstrate 
theoretically and experimentally how the invariance property is
violated in a family of dirty-faces games.


Finally, it is worth remarking that the DCH posterior beliefs about the types and levels are 
generally \emph{correlated} despite of that the types and levels are determined independently 
(Assumption \ref{assumption:level_type_independence}). 
This will also be illustrated in the next section.

\section{DCH Analysis of the Dirty-Faces Games}\label{sec:dirty_face}

The dirty-faces game was originally developed by \cite{littlewood1953littlewood} to study the 
role of common knowledge.\footnote{It has also been referred to as the ``cheating wives puzzle'' 
\citep{gamow1958forty}, the ``cheating husbands puzzle'' \citep{moses1986cheating},
the ``muddy children puzzle'' \citep{barwise1981scenes, halpern1990knowledge}, 
and the ``red hat puzzle'' \citep{hardin2008introduction}.}
In this section, I focus on a particular specification of the game that has been theoretically 
and experimentally studied in the literature (see e.g., \cite{fudenberg1993game},
\cite{weber2001behavior} and \citealt{bayer2007dirty}).

There are two players $N=\{1,2 \}$ and there are up to $2\leq T < \infty$ periods.
At the beginning of the game, each player $i$ is randomly assigned a face type, denoted as $x_i$, 
which can be either $x_i=O$ (representing a clean face) or $x_i=X$ (representing a dirty face).
The face types are i.i.d. drawn 
from the distribution $p = \Pr(x_i=X) = 1- \Pr(x_i = O)$ where $p>0$ represents the probability of 
having a dirty face. 
After the face types are determined, each player $i$ can 
observe the other player's face type $x_{-i}$ but not their own face. Hence, 
player $i$'s private information is the other player's face type $x_{-i}$. 
Furthermore, if at least one player has a dirty face, a public announcement is made, 
informing both players of this fact.

If there is no announcement, it is common knowledge to both players that both faces are clean. To 
avoid triviality, I will focus on the case where an announcement is made.

After seeing the other player's face type and the announcement, in each period, 
every player $i$ simultaneously chooses to ``Wait'' ($W$) or ``Claim'' (to have a dirty face, $C$)
and their actions are revealed at 
the end of each period. The game will end after any period where some player chooses $C$ or after 
period $T$. The last period of the game is called the ``terminal period,'' and both players' 
payoffs are determined by their own face types and their actions in the terminal period.

Suppose period $t$ is the terminal period. If player $i$ chooses $W$ in the terminal period, 
his payoff for this game is 0 regardless of his face type. On the other hand, if player $i$
chooses $C$ to terminate the game, he will receive $\alpha>0$ if his face is dirty and 
$-1$ if his face is clean. Besides, payoffs are discounted with a common discount factor $\delta\in (0,1)$
per period. That is, if player $i$ claims in period $t$ and $x_i = X$, player $i$
will receive $\delta^{t-1}\alpha$, but player $i$ will receive $-\delta^{t-1}$ if $x_i = O$.
To make players unattractive to gamble if 
they do not have additional information except for the prior, following \cite{weber2001behavior}
and \cite{bayer2007dirty}, I assume that 
\begin{align}\label{payoff_assumption}
p\alpha - (1-p) <0 \iff 0<\bar{\alpha} \equiv \frac{p\alpha}{1-p} < 1,  
\end{align}
which guarantees it is strictly dominated to choose $C$ in period 1 when seeing a dirty face. 
Thus, a two-person dirty-faces game is defined by a tuple 
$\langle T, \delta, \alpha, p \rangle$ where $(\delta, \bar{\alpha}) \in (0,1)^2$.

With common knowledge of rationality, the unique equilibrium can be solved through the
following iterative reasoning: When player $i$ sees a clean face,
the public announcement will lead him to realize that his own face is dirty and 
claim in period 1. On the other hand, when player $i$ sees a dirty face,
he will wait in period 1 because of the uncertainty about his own face.
However, if player $-i$ also waits in period 1,
player $i$ will then recognize that his own face is dirty and claim in period 2,
as player $i$ knows that if his own face were clean,
player $-i$ would have claimed in period 1.

When implementing this game in a laboratory experiment, the natural approach is to specify this game 
as a \emph{sequential} dirty-faces game and allow subjects to make decisions period-by-period, 
following the rules described above, using the direct-response method.
Alternatively, the other approach is the strategy method which specifies this game as a 
\emph{simultaneous} dirty-faces game---after seeing the other's face and the announcement, players
simultaneously decide a ``plan'' which specifies the period to claim or always wait.
From the standard game-theoretic perspective, the sequential and simultaneous dirty-faces
game are strategically equivalent as they share the same reduced normal form. 
In the following, I will demonstrate that the DCH solution varies in these two versions of the game,
illustrating the violation of invariance under strategic equivalence of DCH.

\subsection{DCH Solution for the Sequential Dirty-Faces Games}\label{sec:ext_dirty_face}

In the sequential dirty-faces game, since there are (at most) $T$ periods, a behavioral strategy for 
player $i$ is a mapping from the period and the observed face type ($x_{-i} \in \{O,X \}$)
to the probability of choosing $C$. The behavioral strategy is denoted by 
$$\sigma_i: \{1,\ldots, T\} \times \{O,X\} \rightarrow [0,1].$$

For the sake of simplicity, I assume that each player $i$'s level is i.i.d. drawn from 
the distribution $p = (p_k)_{k=0}^{\infty}$ where $p_k>0$ for all $k$. In the DCH solution, 
each player's optimal behavioral strategy is level-dependent. Let the behavioral strategy of 
level $k$ player $i$ be $\sigma_i^k$. Following previous notations, let 
$\mu_i^k(x_i, \tau_{-i} |t, x_{-i})$ be level $k$ player $i$'s belief about their own face and 
the level of the other player, conditional on observing $x_{-i}$ and being at period $t$.
Level 0 players will uniformly randomize everywhere, so $\sigma_{i}^0(t,x_{-i})=1/2$ 
for all $t$ and $x_{-i}$.

Proposition \ref{prop_extensive_dirty} fully characterizes the DCH solution
for the sequential dirty-faces games. 
When observing a clean face, a player can immediately figure out that his face is dirty. 
Therefore, DCH coincides with the equilibrium prediction when $x_{-i}=O$.
However, if a player sees a dirty face and the other player waits in period 1, 
he cannot tell his face type for sure, no matter how sophisticated he is.
Instead, he will believe that he is more likely to have a dirty face as the game continues. 
As a result, conditional on observing a dirty face, level $k\geq 2$ players will claim 
as long as the reward $\bar{\alpha}$ is high enough or the discount rate $\delta$ is sufficiently low.
Otherwise, they will wait for more evidence.

\begin{proposition}\label{prop_extensive_dirty}
For any sequential two-person dirty-faces game, the level-dependent strategy profile of the 
DCH solution satisfies that for any $i\in N$,
\begin{itemize}
    \item[1.] $\sigma_i^k(t,O)=1$ for any $k\geq 1$ and $1\leq t \leq T$.
    
    \item[2.] $\sigma_i^1(t,X)=0$ for any $1\leq t \leq T$. Moreover, for any $k\geq 2$, 
    \begin{itemize}
        \item[(1)] $\sigma_i^k(1,X)=0$,
        \item[(2)] for any $2\leq t \leq T-1$, $\sigma_{i}^{k}(t,X) = 1$ if and only if
        \begin{align*}
            \bar{\alpha} \; \geq \; \frac{\left[ \left(\frac{1}{2}\right)^{t-1} - \left(\frac{1}{2}\right)^{t}\delta \right]p_0}{\left[ \left(\frac{1}{2}\right)^{t-1} - \left(\frac{1}{2}\right)^{t}\delta \right]p_0 + (1-\delta)\sum_{j=1}^{k-1}p_j} ,
        \end{align*}
        \item[(3)] $\sigma_{i}^{k}(T,X) = 1 $ if and only if
        \begin{align*}
        \bar{\alpha} \; \geq \; \frac{\left(\frac{1}{2}\right)^{T-1}p_0}{\left(\frac{1}{2}\right)^{T-1}p_0 + \sum_{j=1}^{k-1}p_j}.   
    \end{align*}
    \end{itemize}
\end{itemize}
\end{proposition}

\begin{proof}
See Appendix \ref{appendix:proof_dirty}.
\end{proof}

To gain insights into the mechanics of the model, I analyze the behavior of level 1 and 2 players. 
Level 1 players believe the other player is non-strategic, implying that the other player's 
actions do not provide any information about their face type. Consequently, for level 1 players,
the announcement and their own observations are the only relevant sources of information.
As a result, level 1 players in each period 
will behave exactly the same as in period 1---they will claim 
when seeing a clean face, and wait when seeing a dirty face. 
This is because they cannot gather any additional information from the other's actions, 
and their belief about their own face type remains unchanged in every period.

For level 2 players, they will also claim to have a dirty face immediately upon seeing a clean face. 
Furthermore, level 2 players are aware that level 1 players will claim in period 
1 if they observe a clean face. In contrast, when observing a dirty face, 
level 2 players will wait in period 1 (which is a strictly dominant strategy) and form a joint
belief about the other player's level and their own face type if the game proceeds to period 2.
Due to the presence of level 0 players, even if the game proceeds to period 2, level 2 players are still 
uncertain about their face types. However, they will know it is \emph{impossible} 
that the other player is level 1 and their own face is clean. Specifically, level 2 players' posterior belief
about the their own face $x_i$ and the other player's level $\tau_{-i}$ is 
$\mu^2_i(x_i,\tau_{-i}\mid 2,X)$ where
\begin{align*}
    &\mu_i^2(X,0|2,X) = \frac{\left(\frac{1}{2}\right)p_0 p}{\left(\frac{1}{2}\right) p_0 + p p_1},
    &\mu_i^2(O,0|2,X) = \frac{\left(\frac{1}{2}\right)p_0 (1-p)}{\left(\frac{1}{2}\right) p_0 + p p_1},\\
    &\mu_i^2(X,1|2,X) = \frac{pp_1}{\left(\frac{1}{2}\right) p_0 + p p_1},
    &\mu_i^2(O,1|2,X) = 0. \;\;\;\;\;\;\;\;\;\;\;\;\;\;\;\;\;\;\;
\end{align*}

As the game proceeds beyond period 2, level 2 players will make the inference that 
if the other player is level 1, then their own face is dirty; otherwise, the other player's 
actions are uninformative about their face types. Moreover, at any period $2\leq t \leq T$,
level 2 players' marginal belief about having a dirty face is
\begin{align*}
\mu_i^2(X|t,X) 
=  \underbrace{\frac{\left(\frac{1}{2}\right)^{t-1} p_0 p}{\left(\frac{1}{2}\right)^{t-1} p_0 + p p_1}}_{
= \; \mu_i^2(X,0 \mid t,X)}
+ \underbrace{\frac{pp_1}{\left(\frac{1}{2}\right)^{t-1} p_0 + p p_1}}_{=\; \mu_i^2(X,1\mid t,X)}
= \frac{p\left[\left(\frac{1}{2}\right)^{t-1}p_0 + p_1 \right]}{\left(\frac{1}{2}\right)^{t-1} p_0 + p p_1},
\end{align*}
which is increasing in $t$, suggesting that level 2 players are more certain about having 
a dirty face in later periods. This is level 2 players' benefit of waiting
when seeing a dirty face. However, the 
cost of waiting is that the other player may randomly end the game (if the other is level 0) 
and the payoff is discounted. Therefore, level 2 players' tradeoff is analogous to the sequential 
sampling problem of \cite{wald1947sequential}---they decide the \emph{optimal 
stopping period} to claim. 
The optimal stopping period depends on the parameters $\bar{\alpha}$ and $\delta$, as well as the 
distribution of levels. This is in sharp contrast with the equilibrium prediction
that the equilibrium prediction is independent of 
the parameters.

In particular, for any period $2\leq t \leq T$, level 2 player $i$'s expected payoff of 
claiming to have a dirty face is
$$\mathbb{E}u_i^2(C|t) := \delta^{t-1}\left[\alpha \mu^2_i(X|t,X) - \mu^2_i(O|t,X)\right]. $$
At period $T$, the last period of the game, it is optimal to claim if and only if
$$\mathbb{E}u_i^2(C|T)\geq 0 \iff \alpha \geq \frac{\mu^2_i(O|T,X)}{\mu^2_i(X|T,X)}
\iff \bar{\alpha} \geq  \frac{\left(\frac{1}{2}\right)^{T-1}p_0}{\left(\frac{1}{2}\right)^{T-1}p_0 + p_1}.$$
For any other period $2\leq t' \leq T-1$, it is optimal to claim at period $t'$ only if 
$$\mathbb{E}u_i^2(C|t') \geq \Pr(t'+1 | t',X)\mathbb{E}u_i^2(C|t'+1),$$
where $\Pr(t'+1 | t',X)$ is level 2 player $i$'s belief about the probability that player $-i$ would wait 
in period $t'$.\footnote{$\Pr(t'+1 | t',X)$ is the probability that the other player chooses to wait in 
period $t'$, which is
\begin{align*}
    \Pr(t'+1 | t',X) = 0.5 \cdot \mu^2_i(0|t',X) + 1\cdot \mu^2_i(1|t',X)
            = \frac{\left(\frac{1}{2}\right)^{t'}p_0 + p p_1}{\left(\frac{1}{2}\right)^{t'-1}p_0 + p p_1 }.
\end{align*}
}
Rearranging the inequality yields the condition stated in Proposition \ref{prop_extensive_dirty}. 
Furthermore, the proof in Appendix \ref{appendix:proof_dirty} shows that these conditions are 
not only necessary but also sufficient to pin down level 2 players' optimal 
stopping periods. By induction on the levels, it can be shown that no matter how sophisticated the 
players are, the behavior is characterized by the solution of a sequential sampling problem.

Proposition \ref{prop_extensive_dirty} characterizes the level-dependent behavioral strategies. 
Alternatively, the DCH solution can be characterized by the level-dependent stopping period
(given observing $x_{-i}$), which is formally defined in Definition \ref{def_optimal_stopping}.

\begin{definition}[Stopping Period]\label{def_optimal_stopping}
For any sequential two-person dirty-faces game and its DCH level-dependent strategy profile $\sigma$, 
let $\hat{\sigma}^{k}_i (x_{-i})$ be level $k$ player $i$'s earliest period to claim 
to have a dirty face conditional on observing $x_{-i}$ for any $k\geq 1$ and $i\in N$. Specifically, 
\begin{align*}
    \hat{\sigma}_i^k(x_{-i}) = 
    \begin{cases}
    \min\left\{ t' :  \sigma_i^k(t' ,x_{-i}) = 1 \right\},  \;\;\;\; \mbox{ if } \exists \; t \mbox{ s.t. }  \sigma_i^k(t,x_{-i}) = 1 \\
    T+1, \;\;\;\;\;\;\;\;\;\;\;\;\;\;\;\;\;\;\;\;\;\;\;\;\;\;\;\;\;\;\;\;\;\;
    \mbox{ otherwise. }
    \end{cases}
\end{align*}
\end{definition}

With Definition \ref{def_optimal_stopping}, Corollary \ref{coro_extensive_dirty}
is a direct consequence of Proposition \ref{prop_extensive_dirty}. 
If $x_{-i} = O$, every strategic level of players will know their face is dirty and
claim to have a dirty face in period 1, viz. $\hat{\sigma}_i^k(O)=1$ for every $k\geq 1$.
In contrast, if $x_{-i} = X$, Corollary \ref{coro_extensive_dirty} shows that 
the optimal stopping period is monotonically decreasing in $k$, implying that  
higher-level players tend to claim in fewer periods.

\begin{corollary}\label{coro_extensive_dirty}
For any sequential two-person dirty-faces game, the DCH level-dependent strategy profile $\sigma$
can be equivalently characterized by level-dependent stopping periods. Moreover, for any $i\in N$,
we know 
\begin{enumerate}
    \item $\hat{\sigma}_i^k(O)=1$ for any $k\geq 1$, 
    \item $\hat{\sigma}_i^1(X)=T+1$, and $\hat{\sigma}_i^k(X)\geq 2$ for all $k\geq 2$.
    \item $\hat{\sigma}_i^k(X)$ is weakly decreasing in $k$. 
\end{enumerate}
\end{corollary}

\begin{proof}
See Appendix \ref{appendix:proof_dirty}.
\end{proof}

To summarize, I illustrate the DCH optimal stopping periods of level 2 and level infinity 
players when seeing a dirty face, i.e., $\hat{\sigma}_i^2(X)$ and 
$\hat{\sigma}_i^\infty(X)$. Because the set of dirty-faces games is described by 
$(\delta, \bar{\alpha})$, it is simply the unit square on the 
$(\delta, \bar{\alpha})$-plane. For the illustrative purpose, I consider $T=5$ and the 
distribution of levels follows Poisson(1.5), which is an empirically 
regular prior according to \cite{camerer2004cognitive}.

The DCH stopping periods can be solved according to 
Proposition \ref{prop_extensive_dirty} and are plotted in Figure \ref{fig:dynamic_CH_solution}.
From the figure, we can find that DCH predicts 
it is possible for strategic players to choose any stopping period in $\{2,3,4,5,6 \}$,
depending on the parameters $\bar{\alpha}$ and $\delta$.
For instance, level 2 players will claim in period 2 (red area) if and only if 
\begin{align*}
    \bar{\alpha} \geq \frac{\left(\frac{1}{2} - \frac{1}{4}\delta \right)p_0}{\left(\frac{1}{2} 
    - \frac{1}{4}\delta \right)p_0 + (1-\delta)p_1} = 
    \frac{\left(\frac{1}{2} - \frac{1}{4}\delta \right)e^{-1.5}}{\left(\frac{1}{2} 
    - \frac{1}{4}\delta \right)e^{-1.5} + (1-\delta)1.5e^{-1.5}} = \frac{2-\delta}{8-7\delta}.
\end{align*}

In addition, DCH predicts the comparative statics that 
the optimal stopping period is weakly decreasing in $\bar{\alpha}$ and 
weakly increasing in $\delta$ for any level $k\geq 2$. The intuition is that 
when $\bar{\alpha}$ is larger or $\delta$ is smaller, waiting becomes more costly, which 
causes the players to claim earlier with a less certain belief about their own face type.

\begin{figure}[htbp!]
\centering
\noindent%
\makebox[\textwidth]{
\includegraphics[width=\textwidth]{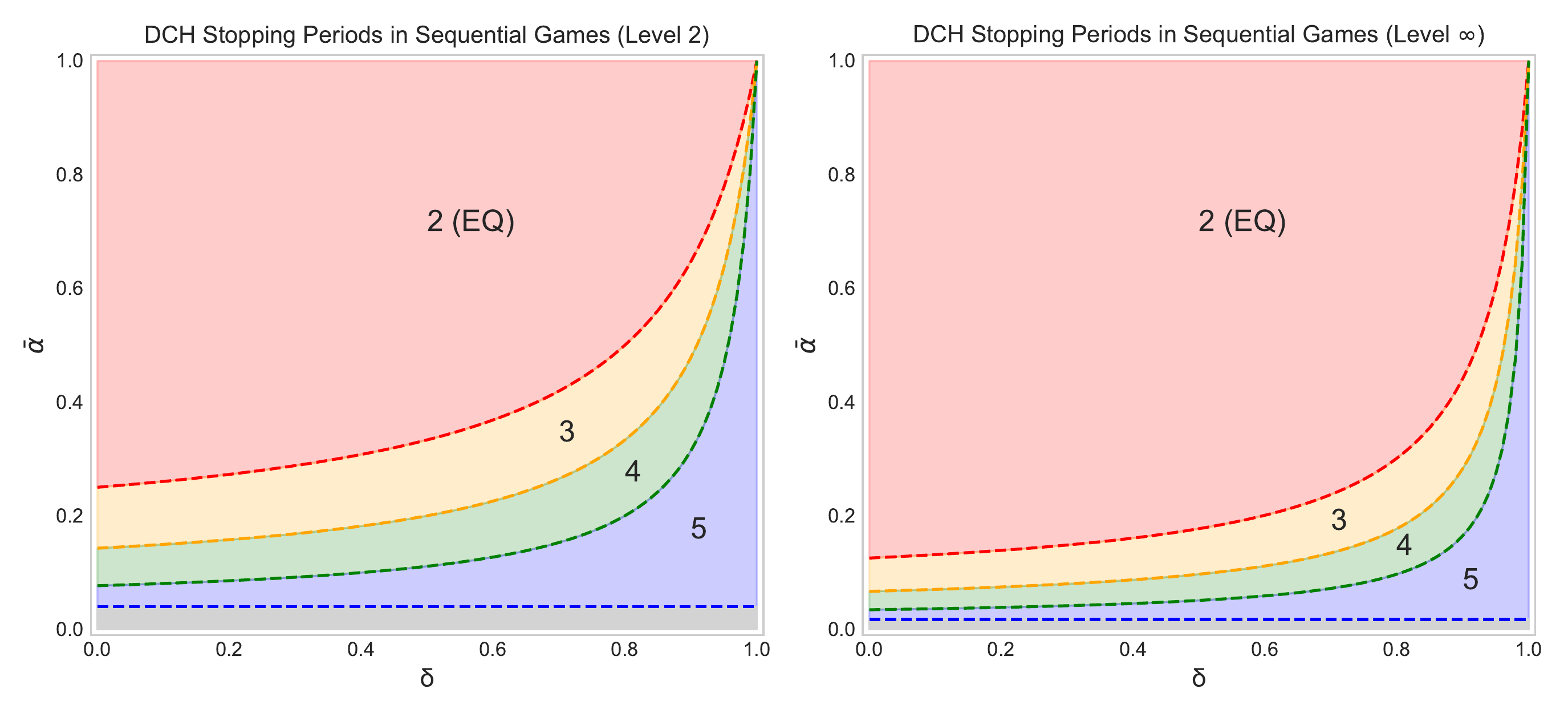}
}
\caption{DCH stopping periods in sequential dirty-faces games 
    for level 2 (left) and level $\infty$ players (right) as $x_{-i} = X$ where $T=5$ and the distribution of levels follows Poisson(1.5).}
    \label{fig:dynamic_CH_solution}
\end{figure}

\subsection{DCH Solution for the Simultaneous Dirty-Faces Games}\label{sec:normal_dirty_face}

In contrast, the strategically equivalent simultaneous dirty-faces game is essentially a 
one-period game where players simultaneously choose an action from the set 
$S=\{1, \ldots, T+1 \}$. Action $t\leq T$ represents the plan to wait from
period 1 to $t-1$ and claim in period $t$. Action $T+1$ is the plan to always wait.
In the simultaneous dirty-faces game, a mixed strategy for player $i$ is a mapping from the 
observed face type ($x_{-i}\in \{O,X \}$) to a probability distribution over the action set.
The mixed strategy is denoted by 
$$\tilde{\sigma}_i: \{O,X \} \rightarrow \Delta(S). $$  
Suppose $(s_i, s_{-i})$ is the action profile. If $s_i \leq s_{-i}$, then 
the payoff for player $i$ is computed as the case where player $i$ claims in
period $s_i$; if $s_i > s_{-i}$, then player $i$'s payoff is 0.

The DCH solution for one-stage games coincides with the standard CH solution. For the sake of 
simplicity, I again assume that each player's level is i.i.d. drawn from the distribution 
$p= \left(p_k \right)_{k=0}^{\infty}$ where $p_k >0$ for all $k$. 
Level 0 players will uniformly randomize, regardless of what they observe, so 
$\tilde{\sigma}_i^0(x_{-i}) = \frac{1}{T+1}$ for all $i, x_{-i}$.
Since level $k\geq 1$ players will generically choose pure strategies, I slightly abuse the notation to 
use $\tilde{\sigma}_i^k(x_{-i})$ to denote the pure strategies.\footnote{Specifically, for any $t\in\{1,2,\ldots, T, T+1\}$, we use $\tilde{\sigma}_i^k(x_{-i})=t$ 
to denote the degenerated distribution:
$$\tilde{\sigma}_i^k(x_{-i})(t)=1, \;\;\mbox{ and }\;\;
\tilde{\sigma}_i^k(x_{-i})(t')=0 \;\; \forall \; t'\neq t. $$}

Proposition \ref{prop_strategic_dirty} is parallel to Proposition \ref{prop_extensive_dirty} that 
characterizes the DCH solution for the simultaneous dirty-faces games. 
The intuition is similar to the analysis of sequential dirty-faces games. 
When observing a clean face, players can figure out their face types immediately. 
Hence, they will choose the strictly dominant strategy $\tilde{\sigma}_i^k(O)=1$ for all 
$k\geq 1$. 
On the other hand, when observing a dirty face, players have to make hypothetical 
inferences about their face types and the other player's
level of sophistication.


\begin{proposition}\label{prop_strategic_dirty}
For any simultaneous two-person dirty-faces game, the level-dependent strategy profile of 
the DCH solution satisfies that for any $i\in N$,
\begin{itemize}
    \item[1.] $\tilde{\sigma}_i^k(O)=1$ for any $k\geq 1$.
    
    \item[2.] $\tilde{\sigma}_i^1(X)=T+1$. Moreover, for any $k\geq 2$, 
    \begin{itemize}
        \item[(1)] $\tilde{\sigma}_i^k(X)\geq 2$,
        \item[(2)] for any $2\leq t \leq T-1$, $\tilde{\sigma}_i^{k}(X) \leq t$ if and only if
        \begin{align*}
            \bar{\alpha} \; \geq \; \frac{\left[ \frac{T+2-t}{T+1} - \frac{T+1-t}{T+1}\delta \right]p_0}{\left[ \frac{T+2-t}{T+1} - \frac{T+1-t}{T+1}\delta \right]p_0 + (1-\delta)\sum_{j=1}^{k-1}p_j},
        \end{align*}
        \item[(3)] $\tilde{\sigma}_i^{k}(X) \leq T $ if and only if
        \begin{align*}
        \bar{\alpha} \; \geq \; \frac{\frac{2}{T+1} p_0}{ \frac{2}{T+1}p_0 + \sum_{j=1}^{k-1}p_j}.   
    \end{align*}
    \end{itemize}
\end{itemize}
\end{proposition}

\begin{proof}
See Appendix \ref{appendix:proof_dirty}.
\end{proof}

The characterization is similar to Proposition \ref{prop_extensive_dirty}.
When seeing a dirty face, strategic players will not choose 1, i.e., wait in period 1. 
Instead, they will claim in period $t$
if and only if the reward $\bar{\alpha}$ is sufficiently high or the discount rate 
$\delta$ is low enough. However, the critical level of $\bar{\alpha}$ is different,
indicating a violation of invariance under strategic equivalence.
Although DCH makes different quantitative predictions in the sequential and 
simultaneous dirty-faces games, it makes a similar qualitative prediction that 
higher-level players tend to claim 
earlier than lower-level players. This is proven in Corollary \ref{coro_normal_dirty}.

\begin{corollary}\label{coro_normal_dirty}
For any simultaneous two-person dirty-faces game, the DCH level-dependent 
strategy profile $\tilde{\sigma}$ satisfies for any $i\in N$ and 
any $k\geq 2$, $\tilde{\sigma}_i^k(X)$ is 
weakly decreasing in $k$.
\end{corollary}

\begin{proof}
See Appendix \ref{appendix:proof_dirty}.
\end{proof}

To conclude, I illustrate the DCH strategies for the simultaneous dirty-faces games of 
level 2 and level infinity players when seeing a dirty face, i.e., $\tilde{\sigma}_i^2(X)$
and $\tilde{\sigma}_i^\infty(X)$, with $T=5$ and the prior distribution of levels being 
Poisson(1.5). The DCH strategies can be solved by Proposition \ref{prop_strategic_dirty}
and plotted in the unit square on the $(\delta,\bar{\alpha})$-plane.

From Figure \ref{fig:static_CH_solution}, we can observe that 
the DCH solution for the simultaneous games is similar to the DCH 
solution for the sequential games. In both games, DCH predicts that 
the stopping periods are weakly decreasing in $\bar{\alpha}$
and weakly increasing in $\delta$ for any level $k\geq 2$, although the boundaries of 
the areas are different.

\begin{figure}[htbp!]
\centering
\noindent%
\makebox[\textwidth]{
\includegraphics[width=\textwidth]{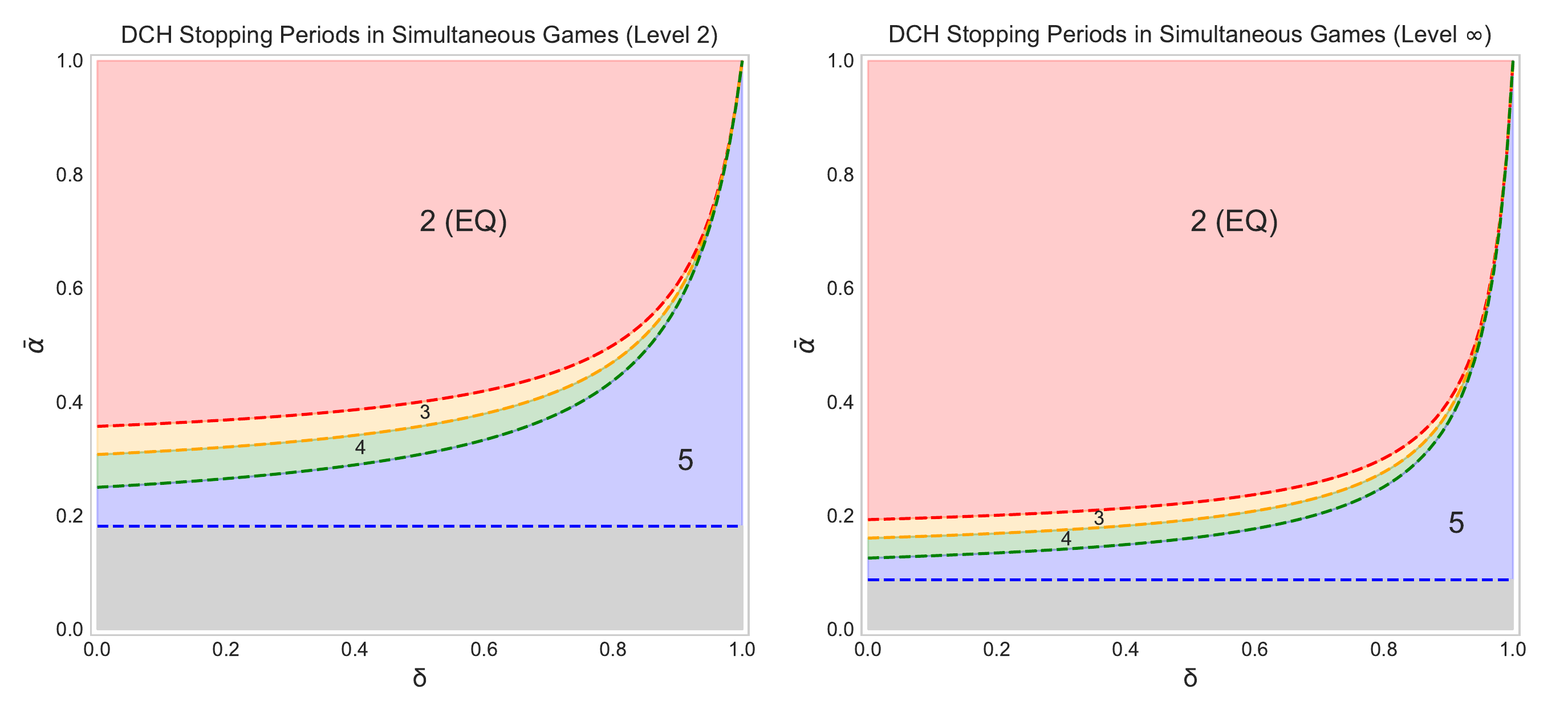}
}
\caption{DCH stopping periods in simultaneous dirty-faces games 
    for level 2 (left) and level $\infty$ players (right) 
    as $x_{-i} = X$ where $T=5$, and the distribution of levels follows Poisson(1.5).}
    \label{fig:static_CH_solution}
\end{figure}

\subsection{The Violation of Invariance under Strategic Equivalence}\label{sec:representation_effect}

As discussed in the previous sections, DCH predicts that players 
might behave differently in two strategically equivalent dirty-faces 
games. This result is driven by the fact that when the cardinalities 
of the action sets differ, the behavioral strategies of level 0 
players may not be outcome-equivalent. This initiates a chain 
reaction that affects the behavior of higher-level players because 
the DCH solution is solved recursively. In this subsection, I will examine \emph{how} changes in 
the cardinalities of the action sets influence behavior in 
dirty-faces games.



In the sequential game, the cardinality of each player $i$'s action set (conditional on each $x_{-i}$) is 
$2^T$, while in the simultaneous game, the cardinality is $T+1$. 
This difference in cardinalities 
leads to different behavior among level 0 players.
For example, in the first period, level 0 players 
will claim to have a dirty face with a probability of $1/2$ in the sequential game, while in the 
simultaneous game, they will claim with a probability of
$1/(T+1)$. 
Although this difference does not impact level 1 players, it 
significantly influences how level 2 (and more sophisticated) players 
update their beliefs regarding their own face types.



\begin{remark}
The standard CH solutions for the sequential and simultaneous dirty-faces games coincide. Therefore, 
the difference of DCH between two versions of the game can alternatively be interpreted as the 
difference between the DCH and the standard CH on sequential games.
\end{remark}

To characterize this distinction, we can begin by partitioning the 
set of dirty-faces games based on the stopping rules of each level in 
the sequential games, i.e., $\hat{\sigma}_i^k(X)$. Specifically, for 
any level $k\geq 1$, we define $\mathcal{E}_{t}^k$ as the set of games 
in which level $k$ players will claim \emph{no later than} period $t$ 
when they observe a dirty face.\footnote{Therefore, Corollary \ref{coro_extensive_dirty} implies for 
level 1 players, $\mathcal{E}_{t}^1 = \emptyset$ for all $t=1,\ldots, T$ and 
$\mathcal{E}_{T+1}^1 = (0,1)^2 $; for higher-level players, 
$\mathcal{E}_1^k = \emptyset$ for all $k\geq 1$.} 
That is, 
\begin{align*}
    (\delta, \bar{\alpha})\in \mathcal{E}_{t}^k \iff \hat{\sigma}_i^k(X)\leq t \;\;\mbox{under}\;\;
    (\delta, \bar{\alpha}).
\end{align*}
The partition is visualized in Figure \ref{fig:dynamic_CH_solution}. For instance, $\mathcal{E}_{2}^2$
corresponds to the ``2 (EQ)'' area in the left panel.\footnote{Formally,
when the distribution of levels follows Poisson(1.5), 
$\mathcal{E}_{2}^2$ is characterized by: $(\delta, \bar{\alpha})\in \mathcal{E}_{2}^2 \iff
(2-\delta)/(8-7\delta) \leq \bar{\alpha} <1 $, and $0<\delta<1$.
}
Second, the set of dirty-faces games can be alternatively partitioned by the 
stopping rules of each level in the simultaneous games, i.e., $\tilde{\sigma}_i^k(X)$.
For any $t\geq 1$ and $k\geq 1$, let $\mathcal{S}_{t}^k$ be the set of dirty-faces games 
where $\tilde{\sigma}_i^k(X)\leq t$, which is illustrated in 
Figure \ref{fig:static_CH_solution}.
Proposition \ref{prop:dirty_representation} compares the DCH solutions in different versions
of the game with the set inclusions of $\mathcal{E}_{t}^k$ and $\mathcal{S}_{t}^k$.


\begin{proposition}\label{prop:dirty_representation}
Consider any $T\geq 2$ and the set of two-person 
dirty-faces games. For any level $k\geq 2$, the following relationships hold.
\begin{itemize}
    \item[1.] $\mathcal{S}_{T}^k \subset \mathcal{E}_{T}^k$.
    \item[2.] $\mathcal{S}_{t}^k \subset \mathcal{E}_{t}^k$
    for any $ \left[\ln(T+1)/\ln{2}\right] \leq t \leq T-1$.
    \item[3.] There is no set inclusion relationship between 
    $\mathcal{S}_{t}^k$ and $\mathcal{E}_{t}^k$ for 
    $ 2 \leq t < \left[\ln(T+1)/\ln{2}\right]$. 
    Moreover, for any $i\in N$, there exists 
    $\overline{\delta}(T,t) \in (0,1)$ such that 
    $t=\hat{\sigma}_i^k(X) \leq \tilde{\sigma}_i^k(X)$
    if $\delta \leq \overline{\delta}(T,t)$ and $\hat{\sigma}_i^k(X) \geq \tilde{\sigma}_i^k(X)=t$ if $\delta > \overline{\delta}(T,t)$. Specifically,
    \begin{align*}
        \overline{\delta}(T,t) = \frac{(2^t-2)(T+1)-(t-1)2^t}{(2^t-1)(T+1)-t2^t}.
    \end{align*}
\end{itemize}
\end{proposition}

\begin{proof}
See Appendix \ref{appendix:proof_dirty}.
\end{proof}

Proposition \ref{prop:dirty_representation} formally compares the DCH solutions of the 
sequential games and the simultaneous games. First, $\mathcal{S}_{T}^k \subset \mathcal{E}_{T}^k$ for any 
$k\geq 2$ implies that when seeing a dirty face, players are more likely to claim 
before the game ends in the sequential game than the simultaneous game. 
Yet, this does not imply players will always claim earlier
in the sequential games than the simultaneous games.
The second and third results show that when the horizon is long enough and 
the players are sufficiently patient, i.e., $\delta>\overline{\delta}(T,t)$,
it is possible for players to claim \emph{later} in the sequential games. 
More surprisingly, the cutoff $\overline{\delta}(T,t)$ is independent of the 
level of sophistication and the prior distributions of levels, suggesting that 
the differences between the two versions have 
the same impact on each level of players.

To illustrate this proposition, consider the case where $T=5$ and the distribution of 
levels follows Poisson(1.5). By Proposition \ref{prop:dirty_representation},
we can find that $\mathcal{S}_{t}^k \subset \mathcal{E}_{t}^k$ 
for any $k\geq 2$ and $3\leq t \leq 5$, while there is no set inclusion relation between 
$\mathcal{S}_{2}^k$ and $\mathcal{E}_{2}^k$. These two sets for $k=2$ and $\infty$ are plotted 
in Figure \ref{fig:representation_effect}.
By Proposition \ref{prop_extensive_dirty} and Proposition \ref{prop_strategic_dirty},
$\mathcal{E}_{2}^2$ and $\mathcal{S}_{2}^2$ can be characterized by:
\begin{align*}
    (\delta, \bar{\alpha}) \in \mathcal{E}_{2}^2 &\iff \bar{\alpha} \geq \frac{ \left(\frac{1}{2} - \frac{1}{4}\delta\right)e^{-1.5}}{\left(\frac{1}{2} - \frac{1}{4}\delta\right)e^{-1.5} + (1-\delta)1.5e^{-1.5}} = \frac{2-\delta}{8-7\delta}, \\
    (\delta, \bar{\alpha}) \in \mathcal{S}_{2}^2 & \iff \bar{\alpha} \geq \frac{ \left(\frac{5}{6} - \frac{2}{3}\delta\right)e^{-1.5}}{\left(\frac{5}{6} - \frac{2}{3}\delta\right)e^{-1.5} + (1-\delta)1.5e^{-1.5}} = \frac{5-4\delta}{14-13\delta}.
\end{align*}

The boundaries of $\mathcal{E}_{2}^2$ and $\mathcal{S}_{2}^2$ 
intersect at $\delta=0.8$, suggesting that when $\delta < 0.8$,
level 2 players tend to claim earlier in the sequential games, and 
vice versa. By similar calculations, it can be shown that the 
boundaries of $\mathcal{E}_{2}^\infty$ and $\mathcal{S}_{2}^\infty$ also 
intersect at $\delta=0.8$, illustrating that the cutoff 
$\overline{\delta}(5,2)$ is the same for every level.


Lastly, as the maximum horizon $T$ increases, the cardinalities of the action sets in the 
sequential and simultaneous games become increasingly distinct. 
Consequently, the behavior of level 0 players diverges even further between the sequential 
and simultaneous games. As $T\rightarrow \infty$, Proposition \ref{prop:dirty_representation}
implies for any period $t\geq 2$ and level $k\geq 2$, $\mathcal{S}_t^k$ and 
$\mathcal{E}_t^k$ do not have set inclusion relationship, suggesting that 
higher-level players do not definitely learn their face types earlier in one game or another. Their 
behavior depends on the parameters $(\delta, \bar{\alpha})$. 
The result is formally presented in Corollary \ref{coro:limit_representation}.

\begin{figure}[htbp!]
\centering
\noindent%
\makebox[\textwidth]{
\includegraphics[width=\textwidth]{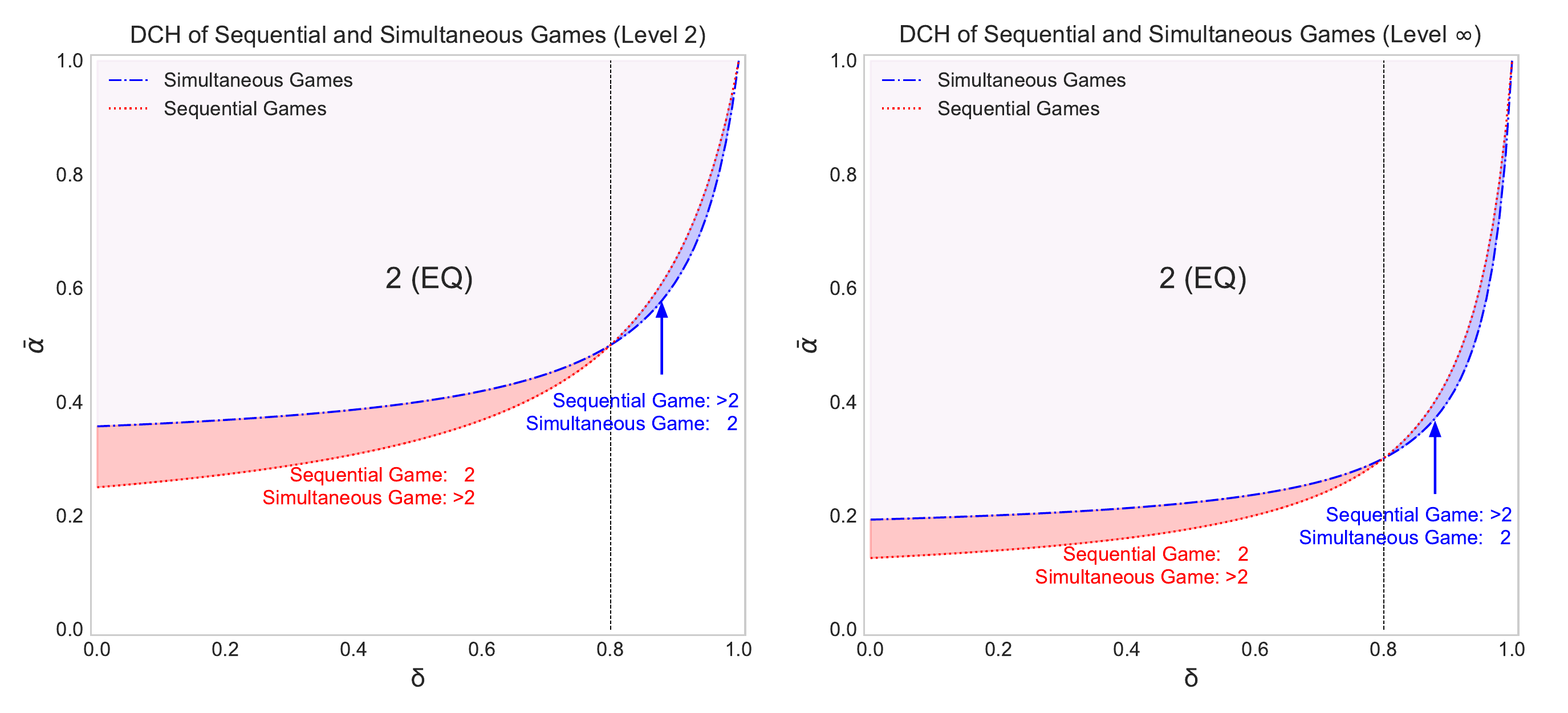}
}
\caption{The set of dirty-faces games where $\tilde{\sigma}_i^k(X)=2$ or $\hat{\sigma}_i^k(X)=2$ 
    for $k=2$ (left) and $k=\infty$ (right) where $T=5$ 
    and the distribution of levels follows Poisson(1.5).}
    \label{fig:representation_effect}
\end{figure}

\begin{corollary}\label{coro:limit_representation}
When $T\rightarrow\infty$, for any $t\geq 2$ and $k\geq 2$, there is no set 
inclusion relationship between $\mathcal{S}_t^k$ and 
$\mathcal{E}_t^k$. Specifically, if $\delta < \overline{\delta}^*(t)$, then 
$t=\hat{\sigma}_i^k(X) \leq \tilde{\sigma}_i^k(X)$; and if $\delta > \overline{\delta}^*(t)$,
then $\hat{\sigma}_i^k(X) \geq \tilde{\sigma}_i^k(X)=t$  where 
\begin{align*}
    \overline{\delta}^*(t) = [2^t-2]/[2^t-1].
\end{align*}
\end{corollary}

\begin{proof}
See Appendix \ref{appendix:proof_dirty}.
\end{proof}

When there are more than two players, DCH predicts a bigger
difference between the two versions
in the sense that the boundaries between the sequential and simultaneous games are further apart. 
For the purpose of illustration, in Appendix \ref{sec:dirty_exp_appendix}, 
I characterize the DCH solutions of three-person three-period games and find that players tend to learn 
their face types earlier in the sequential games.

\section{Experimental Design, Hypotheses and Procedures}\label{section:exp_design}



As demonstrated in the previous section, DCH makes various predictions about 
how people's behavior would vary with the timing (sequential vs. simultaneous) and the 
payoff structures of the dirty-faces games. To test these predictions, 
I conduct a laboratory experiment 
on two-person dirty-faces games tailored to evaluate the DCH solution.

Specifically, the primary goal of the experiment is to measure the violation of invariance under strategic equivalence 
and understand how it interacts with the payoff structures. Furthermore, the variation of the payoff 
structures provides the opportunity to explore the sensitivity of behavior to payoffs in 
both the sequential and simultaneous versions of the game.
Lastly, the stylized facts found in this experiment will help identify the strengths and the weaknesses 
of the DCH solution and alternative theories.

The theoretical analysis of DCH suggests that the main challenge in designing the experiment lies in the 
fact that the magnitude of the difference between the two versions depends on the payoff structure and the 
distribution of levels, which remains unknown before the experiment is run. 
To address this, I first estimate the distribution of levels using the dirty-faces game experimental
data collected by \cite{bayer2007dirty}, and then choose the 
game parameters
to maximize the diagnosticity based on the calibration results.

\subsection{Calibration}\label{subsec:calibration}

The dirty-faces game experiment by \cite{bayer2007dirty} is implemented under the 
direct response method with two treatments:
two-person two-period games and three-person three-period games. 
In both treatments, the prior probability of having a dirty face
is $2/3$, the discount factor $\delta$ is $4/5$, and the reward $\alpha$ is $1/4$.\footnote{
In \cite{bayer2007dirty}, the payoff of correctly claiming a dirty face
is 100 ECU (experimental currency unit) and the penalty of wrongly claiming 
a dirty face is $-400$ ECU. Therefore, the relative reward 
of correctly claiming a dirty face $\alpha=1/4$ can be obtained by normalizing the payoffs.}
I will focus on the data from two-person games because this environment is the closest to my 
experiment.\footnote{\cite{weber2001behavior}'s dataset consists of two experiments where 
experiment 2 is comparable with \cite{bayer2007dirty}'s design. However, there are much fewer
observations in this experiment than \cite{bayer2007dirty} and there is no discount factor, 
making this dataset less ideal for the purpose of calibration.}  
A detailed analysis of the data can be found in Appendix \ref{sec:old_dirty_appendix}.

There are 42 subjects (from two sessions) in the two-person treatment of \cite{bayer2007dirty}.
At the beginning of the experiment, the computer randomly matches two subjects into a group.
Subjects play 14 rounds of dirty-faces games against the same opponent, 
with the face types in each round being independently drawn according to the prior probabilities.
In each round, an announcement is made on 
the screen to both subjects if there is at least one person having a dirty face (type $X$).
At the end of each round, subjects are told their own payoffs from that round and 
they are paid with the sum of the earnings of all 14 rounds.

In the calibration exercise, I exclude the data from the situation where there is no 
public announcement\footnote{If there is no public announcement, it is 
common knowledge that both subjects' faces are clean.}, resulting in 690 observations at the 
information set level.  Following previous notations, I use $(t,x_{-i})$
to denote the situation where subject $i$ sees type $x_{-i}$ at period $t$.
Table \ref{tab:two_person_estimation_all} reports the empirical frequency of 
choosing \emph{claim} at each information set, revealing that the behavior is inconsistent 
with the prediction of standard equilibrium theory, particularly when observing a dirty face. 


Following the literature on the cognitive hierarchy theory, I assume the prior distribution of 
levels follows a Poisson 
distribution. In the Poisson-DCH model, each individual $i$'s level is 
identically and independently drawn from $\left(p_k \right)_{k=0}^\infty$ where 
$p_k = e^{-\tau}\tau^k /k!$ for all $k\in \mathbb{N}_0$ 
and $\tau>0$. Once the distribution of levels is specified, DCH makes a precise prediction
about the aggregate choice frequency at each information set.\footnote{The aggregate choice frequency 
can be constructed as follows. Consider any game, any player $i$, any information set $\mathcal{I}_i$, 
and any available action $c_i$ at this information set. Let $P_k(c_i | \mathcal{I}_i)$ represent the 
probability of level $k$ player $i$ choosing $c_i$ at $\mathcal{I}_i$. Additionally, 
let $f(k|\mathcal{I}_i, \tau)$ be the posterior distribution of levels at information set $\mathcal{I}_i$. 
The choice frequency predicted by DCH for action $c_i$ at information set $\mathcal{I}_i$ is the 
aggregation of choice probabilities from all levels, weighted by the proportion $f(k|\mathcal{I}_i, \tau)$:
\begin{equation*}
\mathcal{D}(c_i | \mathcal{I}i, \tau) \equiv \sum_{k=0}^{\infty} f(k|\mathcal{I}_i, \tau) P_k(c_i | \mathcal{I}_i, \tau ).
\end{equation*}} 
The rationale for estimating the
Poisson-DCH model is to find $\tau$, estimated using the maximum likelihood method, which minimizes 
the difference between the choice frequencies predicted by DCH and the empirical frequencies. 
See Appendix \ref{subsecsec:likelihood_functions} for 
the details on the construction of the likelihood function.

It is worth remarking that since $\tau$ is the mean (and variance) of the Poisson distribution, 
the economic interpretation of $\tau$ is as the average level of sophistication among 
the population. Additionally, another property of the Poisson-DCH model 
is that as $\tau\rightarrow\infty$, the aggregate choice frequencies predicted by DCH converge 
to the equilibrium predictions. This provides a second interpretation of $\tau$: the higher the 
value of $\tau$, the closer the predictions are to the equilibrium. 
See Proposition \ref{prop:covergence_poisson} in Appendix \ref{appendix:proof_dirty} for the proof.

\begin{table}[htbp!]
\centering
\begin{threeparttable}
\caption{Estimation Results for Two-Person Dirty-Faces Games}
\label{tab:two_person_estimation_all}
\begin{tabular}{ccccccc}
\hline
 & $(t,x_{-i})$ & $N$ & $\sigma_i^*(t,x_{-i})$ & $\hat{\sigma}_i(t,x_{-i})$ & DCH & \begin{tabular}[c]{@{}c@{}}Standard\\ CH\end{tabular}  \\ \hline
$\sigma_i(t,x_{-i})$ & $(1,O)$ & 123 & 1.000 & 0.943 & 0.859 & 0.791   \\
 & $(2,O)$ & \phantom{11}6 & 1.000 & 0.500 & 0.500 & 0.500   \\
 & $(1,X)$ & 391 & 0.000 & 0.210 & 0.141 & 0.104   \\
 & $(2,X)$ & 170 & 1.000 & 0.618 & 0.503 & 0.477   \\ \hline
Parameter & $\tau$ & &  &  & \textbf{1.269} & 1.161   \\
 & S.E. & & &  & \textbf{(0.090)} & (0.095)   \\\hline
Fitness & LL & & &  & -360.75\phantom{-} & -381.46\phantom{-}   \\
 & AIC & & &  & 723.50 & 764.91   \\
 & BIC & & &  & 728.04 & 769.45   \\ \hline
Vuong Test & &  &  &  &  & 6.517   \\
p-value & & &  &  &  & $<0.001\phantom{<}$  \\ \hline
\end{tabular}
\begin{tablenotes}
\footnotesize
\item Note: The equilibrium and the empirical frequencies of $C$ at each 
information set are denoted as $\sigma_i^*$ and $\hat{\sigma}_i$, respectively.
There are 294 games (rounds $\times$ groups).
\end{tablenotes}

\end{threeparttable}

\end{table}

Table \ref{tab:two_person_estimation_all} reports the estimation results of the Poisson-DCH model. 
Additionally, I estimate the standard Poisson-CH model by \cite{camerer2004cognitive} 
as a benchmark.\footnote{The logit-AQRE proposed by \cite{mckelvey1998quantal} is also estimated. 
The likelihood scores between Poisson-DCH and logit-AQRE are not significantly different. See 
Appendix \ref{sec:old_dirty_appendix} for the details.} Comparing the fitness of these models, 
I find that the log-likelihood of DCH is significantly higher than standard CH (Vuong test
p-value $<0.001$),
suggesting that DCH outperforms the standard CH in capturing the empirical pattern.
Besides, the estimated $\tau$ of Poisson-DCH falls within the range of commonly observed $\tau$
in various environments, with a value of 1.269.
In the following, I will design the experiment by treating Poisson(1.269) 
as the true prior distribution of levels.

\subsection{Games and Hypotheses}\label{subsection:design}

In this experiment, I employ a between-subject design where each participant 
is assigned to either the ``sequential treatment'' (using the direct-response method) 
or the ``simultaneous treatment'' (using the strategy method). To observe potential 
heterogeneity in stopping periods, I set the maximum length to be $T=5$ for both treatments.

Assessing whether the difference between the two treatments is challenging because
level 1 players---the most common types of players according to the calibration 
result---will behave the same under the two treatments. 
When observing a dirty face, they will always wait in both the sequential and simultaneous games. 
Therefore, to diagnose the predictivity of DCH,
the game parameters are chosen to make level 2 players behave 
differently under different representations. The behavioral change of level 2 players 
(around 22.6\% based on the calibration) is anticipated to yield a sizable treatment effect.

\begin{figure}[htbp!]
\centering
\noindent%
\makebox[\textwidth]{
\includegraphics[width=\textwidth]{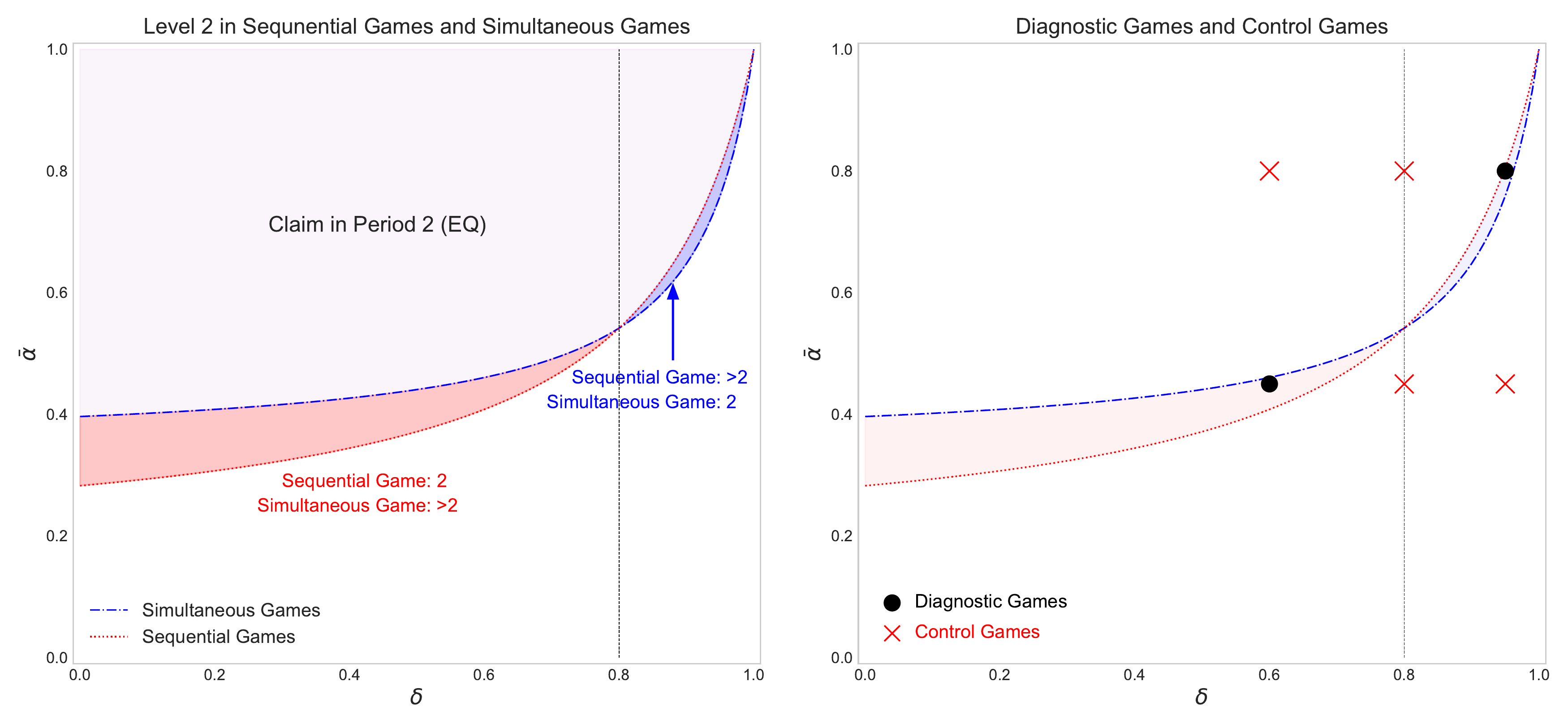}
}
\caption{(Left) The set of dirty-faces games where at information set $(2,X)$,
    level 2 players behave differently in the two versions when $T=5$ and the 
    distribution of levels follows Poisson(1.269). (Right) The two diagnostic games 
    and the four control games in the experiment.}
    \label{fig:design_illustration}
\end{figure}

According to Proposition \ref{prop:dirty_representation}, DCH predicts the existence of a set of dirty-faces 
games in which level 2 players exhibit different behavior in the sequential and simultaneous games 
at information set $(2, X).$
The left panel of Figure 
\ref{fig:design_illustration} illustrates this set of games when the distribution of 
levels follows Poisson(1.269).
From this figure, we can observe the following:
\begin{itemize}
    \item[(1)] For $\delta<0.8$, there is a range of games (red area)
where level 2 players choose to claim at $(2,X)$ in the sequential games
but not in the simultaneous games.
    \item[(2)] For $\delta=0.8$, level 2 (and more sophisticated) players behave the same 
    in the sequential and simultaneous games.
    \item[(3)] For $\delta>0.8$, there is a range of games (blue area)
where level 2 players choose to claim at $(2,X)$ in the simultaneous games but not in the sequential games.
\end{itemize}
Guided by DCH, I consider the following six dirty-faces games $(\delta, \bar{\alpha})$ as depicted in the right panel 
of Figure \ref{fig:design_illustration}.

The set of games consists of two \emph{diagnostic games} where $(\delta,\bar{\alpha})
=(0.6, 0.45)$ and $(0.95, 0.8)$ and four \emph{control games} where $(\delta,\bar{\alpha})
=(0.6, 0.8), (0.8, 0.45), (0.8, 0.8)$ and $(0.95, 0.45)$. 
DCH predicts in the diagnostic games, 
level 2 players will behave differently in two treatments, but not in the control games.
This variation allows us to examine the interplay between 
the violation of invariance under strategic equivalence varies with the payoff structures.

DCH makes several predictions about the comparative statics. First, by Proposition \ref{prop_extensive_dirty}
and \ref{prop_strategic_dirty}, DCH predicts that no matter in sequential games or in 
simultaneous games, when observing a dirty face, players will choose to claim \emph{earlier} 
when $\delta$ is smaller or $\bar{\alpha}$ is higher. 
An implication is that in both treatments, at information set $(2,X)$, players are more 
likely to claim when $\delta$ decreases or $\bar{\alpha}$ increases. 

\begin{hypothesis}
In both the sequential and simultaneous treatments, at information set $(2,X)$,
the empirical frequency of choosing $C$ is higher when $\delta$ decreases or $\bar{\alpha}$ increases.
\end{hypothesis}

Besides, DCH makes a specific prediction regarding the relative magnitude of the 
treatment effect among these six games. First, in the DCH solution, part of the treatment 
effect is attributed to the difference in level 0 players' strategies between the two treatments.
In the sequential games, level 0 players uniformly randomize at every information set, resulting in a
conditional probability to claim at $(2,X)$ is $1/2$. Yet, in the simultaneous games, level 0 players
uniformly randomize across all reduced contingent strategies, leading to a  
conditional probability to claim at $(2,X)$ is $1/5$. In other words, the difference in level 0 players' 
strategies generates a mechanical effect that increases the likelihood of players choosing to claim at $(2, 
X)$ in the sequential games.
Because in all four control games, strategic 
players behave the same at $(2,X)$ under two representations, DCH predicts the 
magnitude of the violation of invariance under
strategic equivalence will be similar in the control games.
Particularly, in the game $(\delta,\bar{\alpha})$, 
the treatment effect can be quantified by computing the difference 
between the conditional probabilities of choosing to claim at $(2,X)$ in the sequential version and the 
simultaneous version. This difference is denoted by $\Delta(\delta, 
\bar{\alpha})$.\footnote{In the sequential version, 
the observed conditional probability of claiming at $(2,X)$ 
is simply the empirical $\sigma_i(2,X)$. For the simultaneous version, the conditional probability can be 
computed from the empirical $\tilde{\sigma}_i(X)$ by 
$\tilde{\sigma}_i(2,X)\equiv \Pr(\tilde{\sigma}_i(X)=2) / 
\sum_{t=2}^6 \Pr(\tilde{\sigma}_i(X)=t).$
Therefore, the treatment effect is quantified by the (empirical) difference between 
$\sigma_i(2,X)$ and $\tilde{\sigma}_i(2,X)$, i.e., $\Delta \equiv \sigma_i(2,X)
- \tilde{\sigma}_i(2,X)$.}

Second, in the game where $(\delta,\bar{\alpha})=(0.6, 0.45)$, level 2 players will claim at $(2,X)$ in 
the sequential version but not in the simultaneous version. 
As a result, DCH predicts that the difference between the two treatments
in this diagnostic game will be stronger compared to the effect observed in the control games.
On the contrary, in the game where $(\delta,\bar{\alpha})=(0.95, 0.8)$, level 2 players will claim at 
$(2,X)$ in the simultaneous version but not in the sequential version. 
This offsets the mechanical effect caused 
by level 0 players. 
The expected differences based on the calibration results are 
summarized below.

\begin{hypothesis}
Based on the calibration results, the expected differences are:
$$\equalto{\Delta(0.6, 0.45)}{31.15\%} > 
\equalto{\Delta(0.6, 0.8)}{7.4\%} = 
\equalto{\Delta(0.8, 0.8)}{7.4\%} \approx 
\equalto{\Delta(0.8, 0.45)}{4.82\%} \approx 
\equalto{\Delta(0.9, 0.45)}{3.26\%} >  
\equalto{\Delta(0.95, 0.8)}{-18.93\%}.$$
\end{hypothesis}

\bigskip


\subsection{Experimental Procedures}\label{subsection:implementation}

The experimental sessions were conducted at the Experimental Social Science Laboratory 
(ESSL) located on the campus of the University of California, Irvine. Subjects were 
recruited from the general undergraduate population, from all majors. Experiments were conducted 
through oTree software \citep{chen2016otree}. I conducted 10
sessions with a total of 118 subjects. 
No subject participated in more than one session. Each session lasted around 45 minutes, and the 
average earnings was \$33.36, including the \$10 show-up fee (max \$52 and min \$10).

Subjects were given instructions at the beginning and the instructions were read aloud. Subjects were 
allowed to ask any questions during the whole instruction process. The questions were answered so that every 
one can hear. Afterwards, they had to answer several comprehension questions on the computer 
screen in order to proceed. 
The instructions for both the sequential treatment and the simultaneous treatments
are identical except for the instructions about the choices and the feedback after each game. 
The instructions for both treatments 
can be found in Appendix \ref{sec:instruction_appendix}.

\begin{table}[htbp!]
\centering
\caption{List of Game Parameters Implemented in the Experiment}
\label{tab:payoff_confi_list}
\begin{tabular}{cccccccccc}
\hline
 &  & \multicolumn{5}{c}{Game Parameters} &  & \multicolumn{2}{c}{Normalized Probabilities} \\
 &  & $\delta$ & $\alpha$ & $p$ &  & $\bar{\alpha}$ &  & one X, one O & two X \\ \hline
Diagnostic Game 1\phantom{'} &  & 0.60 & 0.225 & 0.67 &  & 0.45 &  & 0.25 & 0.50 \\
Diagnostic Game 1' &  & 0.60 & 0.150 & 0.75 &  & 0.45 &  & 0.20 & 0.60 \\
Diagnostic Game 2\phantom{'} &  & 0.95 & 0.400 & 0.67 &  & 0.80 &  & 0.25 & 0.50 \\
Diagnostic Game 2' &  & 0.95 & 0.267 & 0.75 &  & 0.80 &  & 0.20 & 0.60 \\ \hline
Control Game 1\phantom{'} &  & 0.60 & 0.400 & 0.67 &  & 0.80 &  & 0.25 & 0.50 \\
Control Game 1' &  & 0.60 & 0.267 & 0.75 &  & 0.80 &  & 0.20 & 0.60 \\
Control Game 2\phantom{'} &  & 0.80 & 0.225 & 0.67 &  & 0.45 &  & 0.25 & 0.50 \\
Control Game 2' &  & 0.80 & 0.150 & 0.75 &  & 0.45 &  & 0.20 & 0.60 \\
Control Game 3\phantom{'} &  & 0.80 & 0.400 & 0.67 &  & 0.80 &  & 0.25 & 0.50 \\
Control Game 3' &  & 0.80 & 0.267 & 0.75 &  & 0.80 &  & 0.20 & 0.60 \\
Control Game 4\phantom{'} &  & 0.95 & 0.225 & 0.67 &  & 0.45 &  & 0.25 & 0.50 \\
Control Game 4' &  & 0.95 & 0.150 & 0.75 &  & 0.45 &  & 0.20 & 0.60 \\ \hline
\end{tabular}
\end{table}

Each session comprised 12 games with different $(\delta, \alpha, p)$ configurations, as summarized 
in Table \ref{tab:payoff_confi_list}.\footnote{Notice that the parameters are selected such that 
each $(\delta, \bar{\alpha})$ is played twice.}
The sequence of these games was randomized, and in each game, subjects were randomly paired into groups. 
The draws for player types were independent, and the protocol was common knowledge.
To prevent any framing effect, the ``dirty face'' and the ``clean face'' were labelled as ``red''
and ``white'' in the instruction, respectively. Besides, the actions were labelled as ``I'm red''
and ``wait.'' Finally, to avoid situations where both
faces are clean, the probabilities 
were normalized to ensure that having two clean faces was impossible.\footnote{
For example, in the game with $p=2/3$, subjects were informed that 
the probability of one dirty face and one clean face was $1/4$, 
and the probability of two dirty faces was $1/2$.
Therefore, if the other's face was clean, the subject could infer that his 
own face was dirty. Conversely, if the other's face was 
dirty, the subject's belief about his own face being dirty was $2/3$.}

After observing the other's face type, subjects were asked to simultaneously choose their actions. 
In the sequential treatment, subjects simultaneously chose either ``I'm red'' or ``wait.''
If both subjects chose to wait, the game would proceed to the next period and they were asked to 
choose again. The game ended after the period where some one chose ``I'm red'' or after period 5. 
On the other hand, in the simultaneous treatment, subjects simultaneously chose one of the six 
plans (the period to choose ``I'm red'' or always wait) and the plans were implemented 
by the computer. At the end of each match, the subjects were informed of their own payoffs, the true
types and the histories of the game.\footnote{To control for the amount of feedback in both treatments, 
in the simultaneous treatment, subjects would learn the other's exact plan if the other chose
``I'm red'' earlier or at the same period; otherwise, they would be told that the other subject was 
\emph{later than you}.}

Lastly, subjects were paid in cash based on their total points earned from the 12 games. 
The highest possible earnings of each game was 100 points.\footnote{That is, a correct claim in 
the first period would yield 100 points, while an incorrect 
claim in the first period would result in a penalty of $100/\alpha$ points.} The 
conversion rate was two US dollars for every 100 points. Following previous dirty face game experiments 
\citep{weber2001behavior, bayer2007dirty}, each subject was provided an endowment of 900 points 
at the beginning of the experiment to prevent early bankruptcy, and they would only receive 
the show-up fee if the total point is negative.

\section{Experimental Results}\label{section:exp_result}

\subsection{Aggregate-Level Analysis}

The data includes two treatments (sequential and simultaneous) with 60 subjects in the sequential treatment 
and 58 subjects in the simultaneous treatment. Each subject participates in 12 games, resulting in 1024
observations for the sequential treatment and 1979 observations for the simultaneous treatment at the 
information set level.\footnote{For the simultaneous treatment, the choice data at the information set level 
are implied by the contingent strategies. For instance, choosing the contingent strategy ``claim at period 4''
implies that the subject will wait from period 1 to period 3 and claim in period 4.}
Figure \ref{fig:CDF_stopping_OX} provides a comprehensive overview of the data by 
plotting the distribution of stopping periods in both treatments, aggregating 
across all payoff configurations. This analysis considers scenarios where 
players encounter either a clean face or a dirty face.\footnote{In the 
sequential treatment, the cumulative density of stopping periods is 
derived from the choice probability at each information set. For example, the probability 
of stopping in period 1 corresponds to the empirical frequency of choosing $C$. Similarly, 
the probability of stopping in period 2 is the product of the empirical frequency of 
choosing $W$ in period 1 and the empirical frequency of choosing $C$ in period 2. 
The probabilities for other stopping periods are calculated in a similar manner.}

A few key observations emerge from this figure. First, when players see a clean face, 
their behavior is consistent across both treatments. A majority of players seem to 
understand that their face is dirty and claim in period 1. Second, when players encounter
a dirty face, it is evident that the distribution of stopping periods in the 
simultaneous treatment first-order stochastically dominates the distribution 
in the sequential treatment, implying that players are more inclined to 
claim earlier in the sequential treatment. Furthermore, 
a striking pattern in the right panel of Figure \ref{fig:CDF_stopping_OX} is the
prevalence of the ``always wait'' strategy in the simultaneous treatment, 
chosen by approximately 36.72\% of participants. This is in stark contrast to 
the sequential treatment, where the proportion of participants employing the 
``always wait'' strategy is only about 13.88\%.

\begin{figure}[htbp!]
    \centering
    \includegraphics[width=\textwidth]{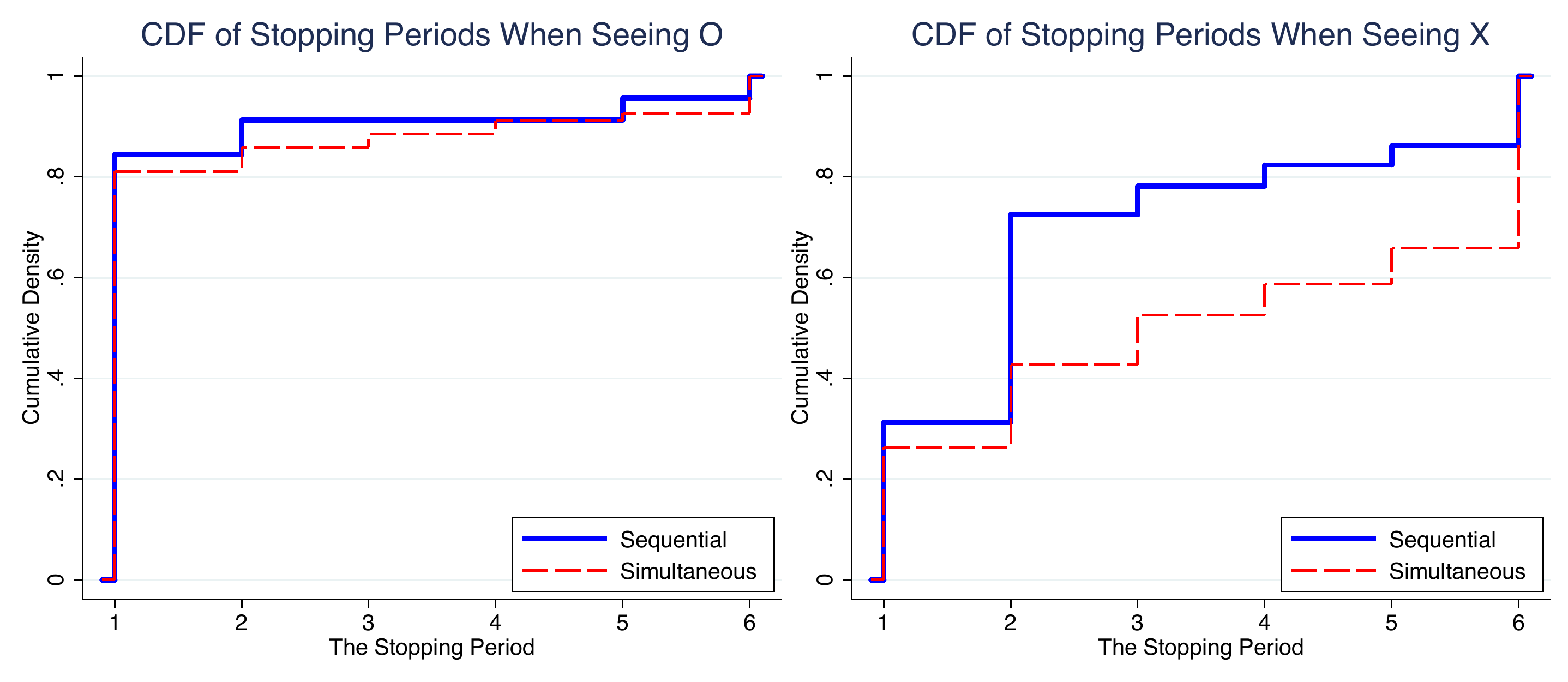}
    \caption{The CDFs of the stopping periods when players see either a clean face (left panel) 
    or a dirty face (right panel). The blue solid and the red dashed are the distributions in 
    the sequential and the simultaneous treatments, respectively.}
    \label{fig:CDF_stopping_OX}
\end{figure}

Focusing on data from the first two periods, Figure 
\ref{fig:aggregate_info_set_main} displays the empirical frequencies of 
choosing $C$ at each information set during the first two periods. From 
the figure, it's evident that at information set $(1,O)$, the behavior in 
the sequential treatment is not significantly different from the 
simultaneous treatment (Ranksum test p-value = 0.4423). In both 
treatments, the frequencies of $C$ exceed 80\%, indicating that the 
majority of subjects understand that choosing $C$ in the first period is a 
strictly dominant strategy.

Despite the limited number of observations at information set $(2,O)$, 
it provides valuable insights into the behavioral strategies of level 0 
players. This is because, from the perspective of DCH, information set 
$(2,O)$ is reached \emph{only} when a player is level 0.
As depicted in Figure \ref{fig:aggregate_info_set_main}, the frequencies 
of $C$ in the sequential and simultaneous treatments are 43.8\% and 25\%, 
respectively.\footnote{A similar pattern is also found in \cite{bayer2007dirty}. In their dataset, the frequency of choosing 
$C$ at information set $(2,O)$ is exactly 50\%, which coincides with
the prediction of DCH.}
These results align with DCH, which predicts that the 
frequencies of $C$ in the sequential and simultaneous treatments should be 
50\% and 20\%, respectively. This suggests that uniform 
randomization is a reasonable specification for 
level 0 players' behavior.

\begin{figure}[htbp!]
    \centering
    \includegraphics[width=0.95\textwidth]{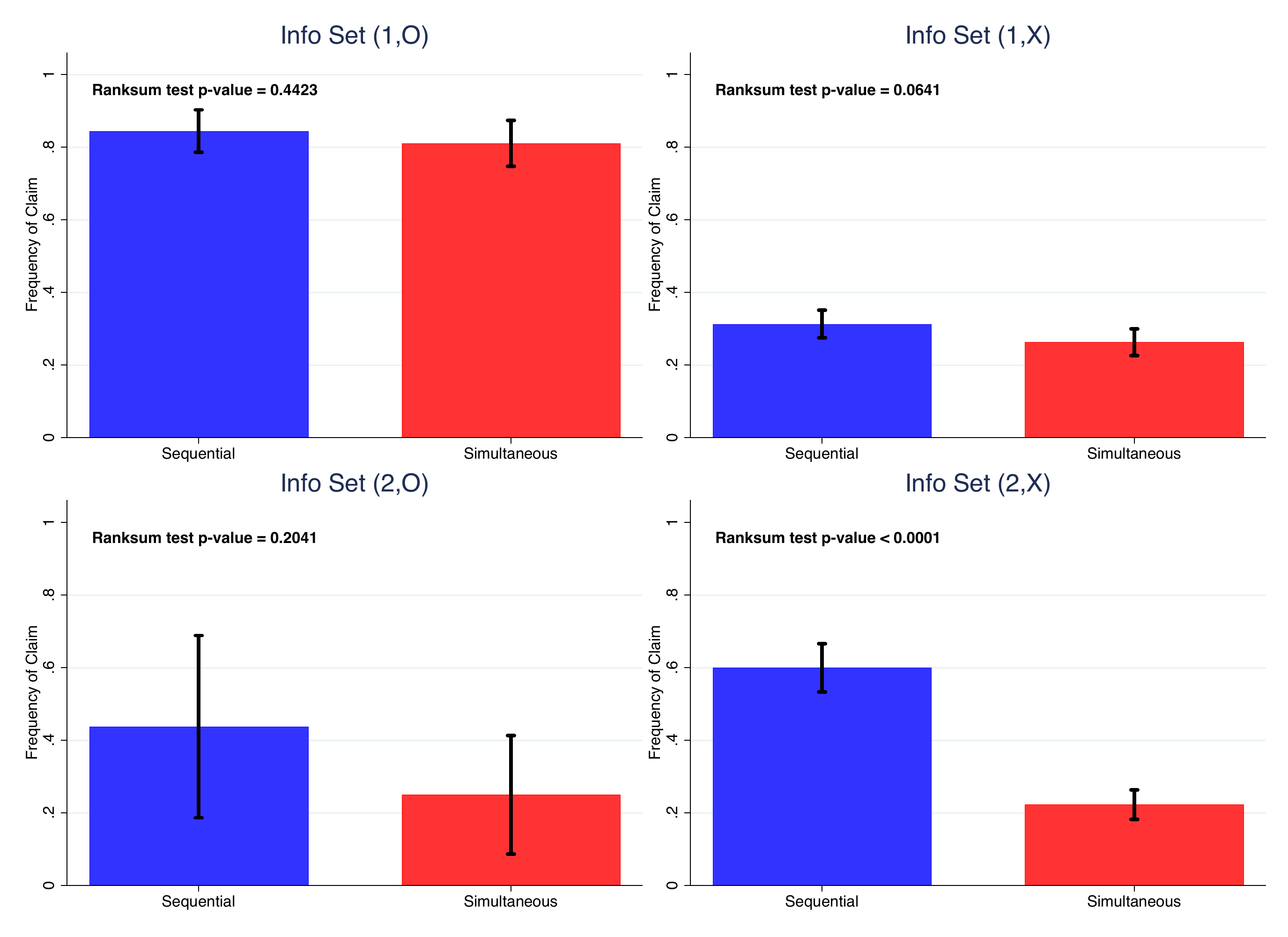}
    \caption{The empirical frequencies of $C$ and 95\% CI at each information set in period 1 and 2, 
    aggregating across all configurations. Each panel represents an information set. The blue 
    bars are the frequencies of the sequential treatment and the red bars are the frequencies of 
    the simultaneous treatment.}
    \label{fig:aggregate_info_set_main}
\end{figure}



On the other hand, when players see a dirty face, they need to make inferences about their own 
faces either from the opponent's actions or hypothetically. However, regardless of the treatment, 
claiming in period 1 is strictly dominated. Comparing the empirical frequencies of choosing $C$ at 
information set $(1,X)$, we find that players in the simultaneous treatment are less likely to choose $C$ 
(Ranksum test p-value = 0.0641). This observation is consistent with DCH, as level 0 players in the 
simultaneous treatment are less likely to claim at information set $(1,X)$. Furthermore, 
a strong treatment effect is detected at period 2. The frequency of choosing $C$ at 
information set $(2,X)$ in the sequential treatment is 60.00\%, while the frequency in the 
simultaneous treatment is 22.28\% (Ranksum test p-value $<0.0001$).

\bigskip

\noindent\textbf{Result 1.} \emph{(1) When observing a clean face in both 
treatments, over 80\% of the subjects choose $C$ in period 1, the strictly 
dominant strategy. Additionally, the behavior at information set $(2,O)$ 
aligns with the prediction of DCH about level 0 players' behavior.
(2) When players observe a dirty face, 
a significant difference emerges: they are more likely to claim at information set $(2,X)$ 
in the sequential treatment. Furthermore, the most prevalent strategy 
in the simultaneous treatment when players see a dirty face is to ``always wait.'' }

\bigskip

The supplementary analysis can be found in Appendix \ref{sec:robustness_appendix}.
In the following, I will focus on information set $(2,X)$, where a strong treatment effect is found, and 
I will test two hypotheses related to the sensitivity of behavior
to the payoff structures and the interplay between the payoff structures and the magnitude of the effect.

\subsection{The Payoff Effect}

Focusing on information set $(2,X)$, DCH predicts that in both treatments, players' behavior 
is sensitive to the payoff configurations. Specifically, DCH predicts that 
the empirical frequencies of choosing $C$ at information set $(2,X)$ 
will exhibit a monotonic relationship with $\delta$ and $\bar{\alpha}$.
To test this prediction, I perform Kruskal-Wallis ranksum tests 
on the sequential and simultaneous treatments separately.

\begin{figure}[htbp!]
    \centering
    \includegraphics[width=\textwidth]{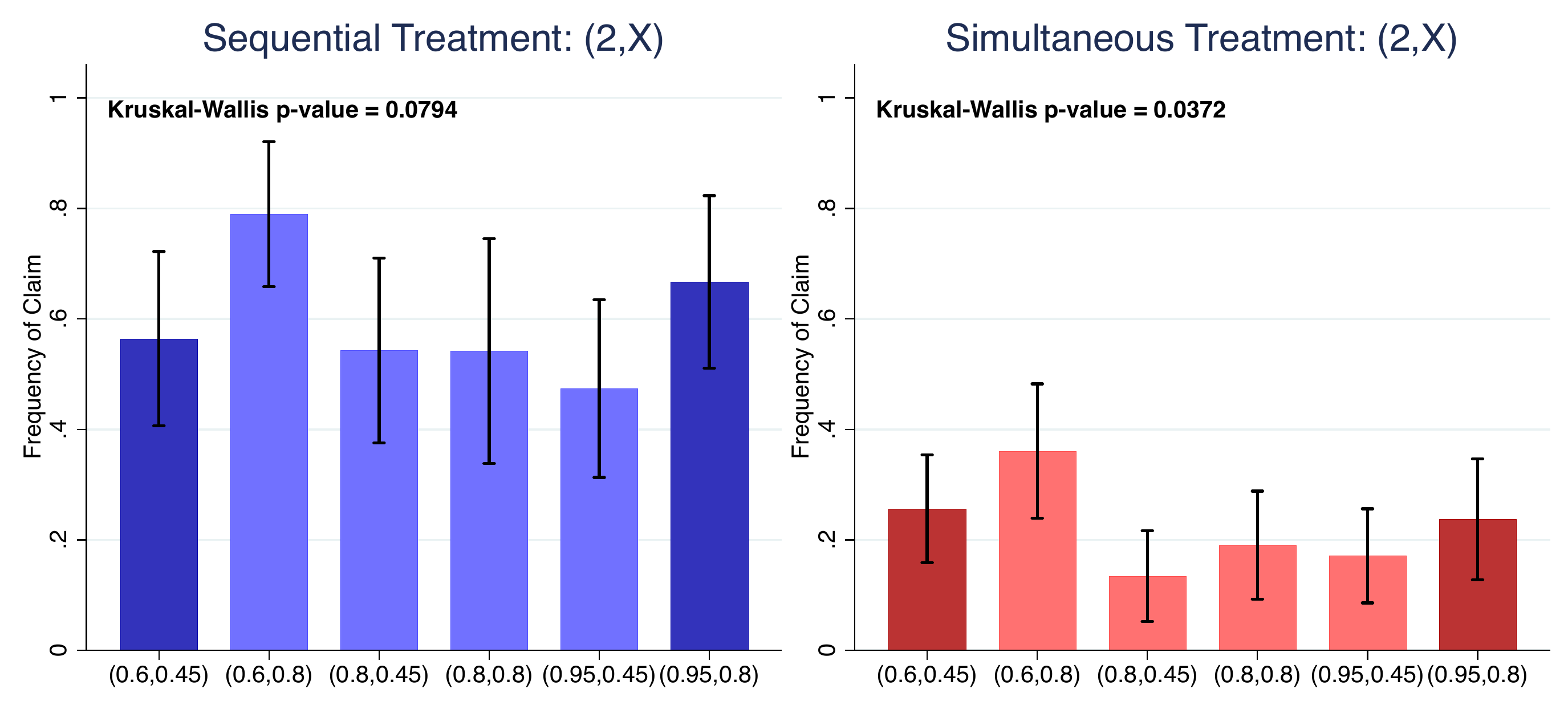}
    \caption{The empirical frequencies of $C$ and 95\% CI at information set $(2,X)$ 
    in each payoff configuration $(\delta, \bar{\alpha})$. The data from the sequential and 
    the simultaneous treatments are plotted in the left and the right panel, respectively.}
    \label{fig:payoff_six_config}
\end{figure}

In the sequential treatment, we find that the null hypothesis is marginally rejected
($\chi^2(5) = 9.856$, p-value = 0.0794), suggesting that behavior is influenced by variations 
in payoff structures. Furthermore, we can observe from the left panel of Figure
\ref{fig:payoff_six_config} that the frequency of choosing $C$ weakly 
increases with $\bar{\alpha}$ for any $\delta$. This monotonic pattern aligns with the prediction of DCH.

Similarly, the null hypothesis is rejected for the simultaneous treatment ($\chi^2(5) = 11.831$, 
p-value = 0.0372), indicating that behavior in the simultaneous treatment 
significantly varies with the payoff parameters. Once again, 
we can observe from Figure \ref{fig:payoff_six_config} that for 
each $\delta$, the frequency of choosing $C$ weakly increases with $\bar{\alpha}$, aligning with DCH.

\bigskip

\noindent\textbf{Result 2.} \emph{The behavior at information set $(2,X)$ in both treatments
significantly varies with payoffs, aligning with the qualitative predictions of DCH.}

\subsection{The Violation of Invariance under Strategic Equivalence}

The behavior in both treatments significantly varies with the payoff structures. Additionally, 
the difference in behavior between the two treatments also varies 
with the payoff structures. This variability allows us to examine the predictions of DCH.

First, the left panel of Figure \ref{fig:rep_effect_2x} displays the joint 
distribution of the empirical frequencies of choosing $C$ at $(2,X)$ 
between the two treatments, where each point represents one payoff configuration. 
From the figure, we can observe that all six points are below the 45-degree line, 
implying that players are more likely to claim at $(2,X)$ in the sequential 
treatment than in the simultaneous treatment, regardless of the payoff configuration.

\begin{figure}[htbp!]
    \centering
    \includegraphics[width=\textwidth]{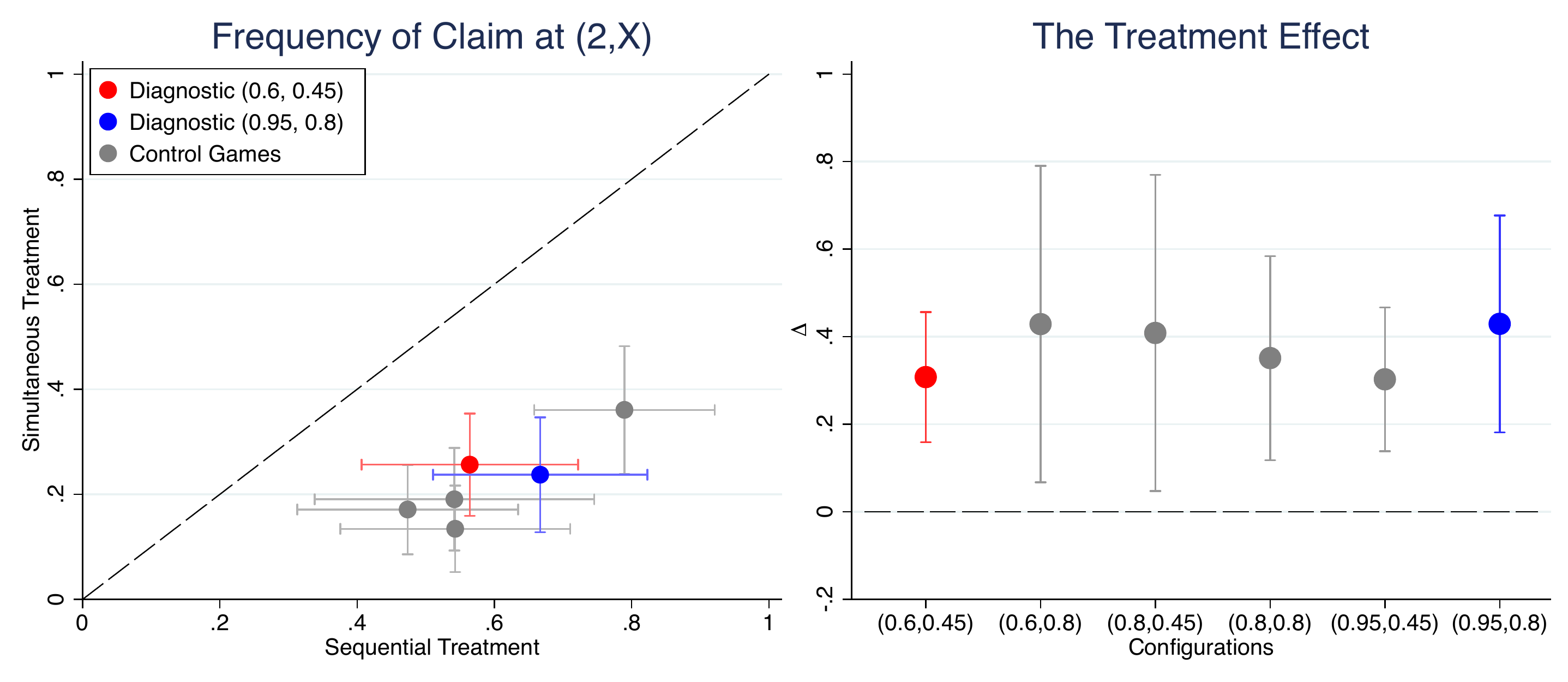}
    \caption{(Left) The empirical frequencies of $C$ and 95\% CI at information set $(2,X)$ in each 
    payoff configuration $(\delta, \bar{\alpha})$. (Right) The difference in the frequencies of $C$ at 
    information set $(2,X)$ between the sequential and simultaneous treatments for each payoff configuration 
    $(\delta, \bar{\alpha})$ with the 95\% CIs. The standard errors are clustered at the 
    session level.}
    \label{fig:rep_effect_2x}
\end{figure}

Second, for each payoff configuration $(\delta, \bar{\alpha})$, I calculate $\Delta(\delta, \bar{\alpha})$, 
which represents the difference in the empirical frequencies of choosing $C$ at $(2,X)$ 
between both treatments. The results are shown in the right panel of Figure \ref{fig:rep_effect_2x}.
Focusing on the diagnostic game where $(\delta, \bar{\alpha})=(0.6, 0.45)$, 
we observe a significant treatment effect with a magnitude of 
$\Delta(0.6,0.45) = 30.77\%$ (95\% CI = $[15.90\%, 45.64\%]$, p-value = 0.001). This is 
highly consistent with the prediction of the calibrated DCH. 

Furthermore, when focusing on the four control games, we find 
that the magnitudes of the treatment effects are similar in these games, 
which aligns with the qualitative predictions of DCH. However, these magnitudes are 
much stronger than the predictions of calibrated DCH, ranging from $\Delta(0.95, 0.45) = 30.26\%$ 
(95\% CI = $[13.85\%, 46.68\%]$, p-value = 0.002) to $\Delta(0.6, 0.8) = 42.88\%$ 
(95\% CI = $[6.73\%, 79.03\%]$, p-value = 0.025). Lastly, 
in the second diagnostic game with $(\delta, \bar{\alpha})=(0.95, 0.8)$, 
DCH predicts a negative treatment effect based on the calibration results. 
However, the observed empirical difference for this game is $\Delta(0.95, 0.8) = 42.94\%$ 
(95\% CI = $[18.17\%, 67.71\%]$ and p-value = 0.004), which is inconsistent with
the quantitative prediction of the calibrated DCH.

\bigskip

\noindent\textbf{Result 3.} \emph{In the diagnostic game with $(\delta, \bar{\alpha})=(0.6, 0.45)$,
the frequency of $C$ at information set $(2,X)$ is 30.77\% higher in the sequential 
treatment compared to the simultaneous treatment. In the diagnostic game with 
$(\delta, \bar{\alpha})=(0.95, 0.8)$, the difference is 42.94\%. Furthermore, treatment 
effects are detected in all control games, with magnitudes exceeding the predictions of calibrated DCH.}

\bigskip

In summary, when analyzing the interplay between the violation of strategic 
equivalence and payoff structures, we observe that while calibrated DCH 
captures some qualitative patterns, the observed magnitudes are significantly larger. 
This suggests that the observed behavior might result from both the violation 
of mutual consistency and other behavioral biases. 
To delve deeper into this aspect, in the next subsection, I will compare 
DCH with other behavioral models that relax other requirements of the 
standard equilibrium theory.



\subsection{Structural Estimation and Model Comparison}\label{subsec:structural_est}

DCH relaxes the requirement of mutual consistency in sequential equilibrium while
still adhering to the requirements of best response and Bayesian inference. Can the
empirical pattern be better explained by relaxing other requirements? To assess the
relaxation of two other requirements, I estimate the ``Quantal Cursed Sequential
Equilibrium (QCSE),''\footnote{See Appendix \ref{subsec:likelihood_qcse_qdch} for a detailed 
description of the model.} which is a hybrid model combining AQRE and CSE, thereby relaxing
the requirements of best response and Bayesian inference.

QCSE assumes that players are unable to fully understand how other players' actions depend 
on their private information.\footnote{In the context of the 
dirty-faces game, QCSE assumes that players do not fully recognize 
how the other player's actions depend on the observed face type.}
In particular, for any strategy profile $\sigma$, any
player $i$ and any information set $\mathcal{I}_i = (\theta_i, h^{t-1})$, the 
\emph{average behavioral strategy of player $-i$} is 
$$\bar{\sigma}_{-i}\left(a_{-i}^{t} | \theta_i, h^{t-1}\right) = 
\sum_{\theta_{-i}'}\mu_i (\theta_{-i}' | \theta_i, h^{t-1})\sigma_{-i}\left(a_{-i}^{t+1} | \theta_{-i}', h^{t-1}\right).$$
In QCSE, there is a parameter $\chi\in[0,1]$. For any $\chi$, $\chi$-cursed player $i$
believes the other players are playing the behavioral strategy:
$$\sigma_{-i}^\chi(a_{-i}^{t} |\theta,  h^{t-1}) = \chi \bar{\sigma}_{-i}(a_{-i}^{t} | \theta_i, h^{t-1})
 + (1-\chi)\sigma_{-i}(a_{-i}^{t} |\theta_{-i},  h^{t-1}), $$
which is a linear combination between the average behavioral strategy (with $\chi$ weight) 
and the true behavioral strategy (with $1-\chi$ weight). When $\chi=0$, players have correct
perceptions about others' behavioral strategies. On the other extreme, when $\chi=1$, players 
fail to understand the correlation between others' actions and types. As the game progresses, 
players update their beliefs via Bayes' rule, believing that other players are using $\sigma_{-i}^\chi$ 
instead of the true behavioral strategy $\sigma_{-i}$. As shown by \cite{fong2023cursed}, 
at any history $h^{t} = (h^{t-1}, a^t)$, player $i$'s $\chi$-cursed belief is
\begin{align*}
    \mu_i^\chi (\theta_{-i} |\theta_i,  h^t) = \chi \mu_i^\chi (\theta_{-i} | \theta_i, h^{t-1})
    + (1-\chi) \left[\frac{\mu_i^\chi (\theta_{-i} | \theta_i, h^{t-1})
    \sigma_{-i}(a_{-i}^t | \theta_{-i}, h^{t-1})}{\sum_{\theta'_{-i}}\mu_i^\chi (\theta'_{-i} | \theta_i, h^{t-1})
    \sigma_{-i}(a_{-i}^t | \theta'_{-i}, h^{t-1})} \right],
\end{align*}
which is a linear combination between the belief from the previous period (with $\chi$ weight) and the 
Bayesian belief (with $1-\chi$ weight).

Moreover, in QCSE, players make quantal responses rather than best responses. In particular, players make 
make logit quantal responses, and the precision is determined by a parameter 
$\lambda\in [0,\infty)$. Consider any information set $\mathcal{I}_i$. For any 
$a_i\in A_i(\mathcal{I}_i)$, let $\bar{u}_{a_i}$ denote the continuation value of $a_i$ in QCSE. 
The choice probability of $a_i$ is given by a multinomial logit distribution:
$$\sigma_i(a_i | \mathcal{I}_i) = \frac{e^{\lambda \bar{u}_{a_i} }}{\sum_{a' 
\in A_i(\mathcal{I}_i)} e^{\lambda \bar{u}_{a'}}}.$$ 
When $\lambda = 0$, players become insensitive to the payoffs, behaving like level 0 players.
As $\lambda$ increases, players' behavior becomes more sensitive to the payoffs. 
In the limit as $\lambda \rightarrow \infty$, 
players become fully rational and make best responses. In summary, 
QCSE relaxes the requirements of best response and Bayesian inferences with two parameters, 
$\lambda\in [0, \infty)$ and $\chi\in[0,1]$. 

\begin{remark}
    When $\chi=0$, QCSE reduces to AQRE, and as $\lambda \rightarrow \infty$, it reduces to CSE.
\end{remark}

To enable a fair comparison between DCH and QCSE, I estimate a 
Quantal DCH model (QDCH) where the prior distribution of levels follows Poisson($\tau$), 
and all levels ($k\geq 1$) of players make logit quantal responses instead of best responses. In essence, 
Quantal DCH relaxes the requirements of best response and 
mutual consistency with two parameters, $\lambda\in [0, \infty)$ and $\tau\in[0,\infty)$.
A description of QDCH can be found in Appendix \ref{subsec:likelihood_qcse_qdch}. 

\begin{remark}
    When $\lambda \rightarrow \infty$, QDCH reduces to DCH.
\end{remark}

In addition to QDCH and QCSE, I also estimate DCH and AQRE, 
which are nested within QDCH and QCSE, respectively.\footnote{CSE cannot be estimated 
independently as it lacks an error structure in the model.} 
These models are estimated using maximum likelihood estimation, and the 
construction of the likelihood functions can be found in Appendix 
\ref{subsec:likelihood_qcse_qdch}. Table \ref{tab:estimate_main_compare} 
presents the estimation results for both the sequential and simultaneous 
treatments.\footnote{In the simultaneous treatment, due to the 
flatness of the log-likelihood functions for both QDCH and QCSE at the
MLE estimates, the square roots of the inverse Hessian matrices are not well-defined.} 
The comparison between the models is summarized in Figure \ref{fig:model_comparison}.

\begin{table}[htbp!]
\centering
\caption{Estimation Results for the Sequential and the Simultaneous Treatment}
\label{tab:estimate_main_compare}
\begin{adjustbox}{width=\columnwidth,center}
\begin{tabular}{cccccccccccc}
\hline
 &  &  & \multicolumn{4}{c}{Sequential Treatment} &  & \multicolumn{4}{c}{Simultaneous Treatment} \\ \cline{4-7} \cline{9-12} 
 &  &  & QDCH & DCH & QCSE & AQRE &  & QDCH & DCH & QCSE & AQRE \\ \hline
Parameters & $\lambda$ &  & 12.371\phantom{1} &  & 5.672 & 5.484 &  & 1774.5 &  & 8.740 & 3.839 \\
 & S.E. &  & (2.062) &  & (0.821) & (0.426) &  & ---  &  & --- & (0.452) \\
 & $\tau$ &  & 1.309 & 0.277 &  &  &  & 0.388 & 0.389 &  &  \\
 & S.E. &  & (0.220) & (0.043) &  &  &  & --- & (0.015) &  &  \\
 & $\chi$ &  &  &  & 0.101 &  &  &  &  & 1.000 &  \\
 & S.E. &  &  &  & (0.364) &  &  &  &  & --- &  \\ \hline
Fitness & LL &  & \phantom{1}-634.72 & \phantom{1}-671.78 & \phantom{1}-648.36 & \phantom{1}-648.40 &  & -1100.76 & -1100.76 & -1167.68 & -1211.08 \\
 & AIC &  & \phantom{-}1273.43 & \phantom{-}1345.57 & \phantom{-}1300.73 & \phantom{-}1298.80 &  & \phantom{-}2205.51  & \phantom{-}2203.52 & \phantom{-}2339.37 & \phantom{-}2424.16 \\
 & BIC &  & \phantom{-}1283.29 & \phantom{-}1350.50 & \phantom{-}1310.59 & \phantom{-}1303.73 &  & \phantom{-}2214.61 & \phantom{-}2208.07 & \phantom{-}2348.46 & \phantom{-}2428.70 \\ \hline
\end{tabular}
\end{adjustbox}
\end{table}

Comparing these four models, we first observe that in both the sequential and 
simultaneous treatments, QDCH fits the data significantly better than QCSE 
(Sequential: Vuong Test p-value = 0.0056; Simultaneous: Vuong Test p-value $< 0.0001$). 
Without relaxing the best response requirement, DCH's fitness is significantly better 
than QCSE in the simultaneous treatment (Vuong Test p-value $< 0.0001$). 
However, in the sequential treatment, QCSE fits the data significantly better than DCH (Vuong Test p-value = $0.0245$).

QDCH outperforms other models in both treatments, indicating that the observed 
violation of strategic equivalence is primarily due to the relaxation of mutual consistency. 
However, there is evidence of the violation of other behavioral biases. 
In the sequential treatment, a significant quantal response effect is observed 
(QDCH vs. DCH: Likelihood Ratio Test p-value $< 0.0001$), but not in the simultaneous treatment 
(Likelihood Ratio Test p-value = 0.9340). Furthermore, in the simultaneous treatment, a significant cursed
effect is detected ($\hat{\chi} = 1.000$, p-value $< 0.0001$), but not in the sequential treatment 
($\hat{\chi} = 0.101$, p-value = 0.7846). This suggests that, players struggle to accurately 
understand how other players' actions depend on their 
private information and update their beliefs accordingly in the simultaneous treatment, 
but not in the sequential treatment.

Lastly, it's worth noting that DCH estimates a significantly lower $\hat{\tau}$ 
in the sequential treatment compared to the simultaneous treatment (Sequential: $\hat{\tau} = 0.277$;
Simultaneous: $\hat{\tau} = 0.389$). In contrast, when introducing quantal responses 
into DCH, we observe a significantly higher $\hat{\tau}$ in the sequential treatment 
compared to the simultaneous treatment (Sequential: $\hat{\tau} = 1.309$; 
Simultaneous: $\hat{\tau} = 0.389$). This suggests that in the simultaneous treatment, 
all of the randomness can be attributed to level 0 behavior, 
whereas in the sequential treatment, some randomness is attributable to the mistakes of higher-level players.


\begin{figure}[htbp!]
    \centering
    \includegraphics[width= \textwidth]{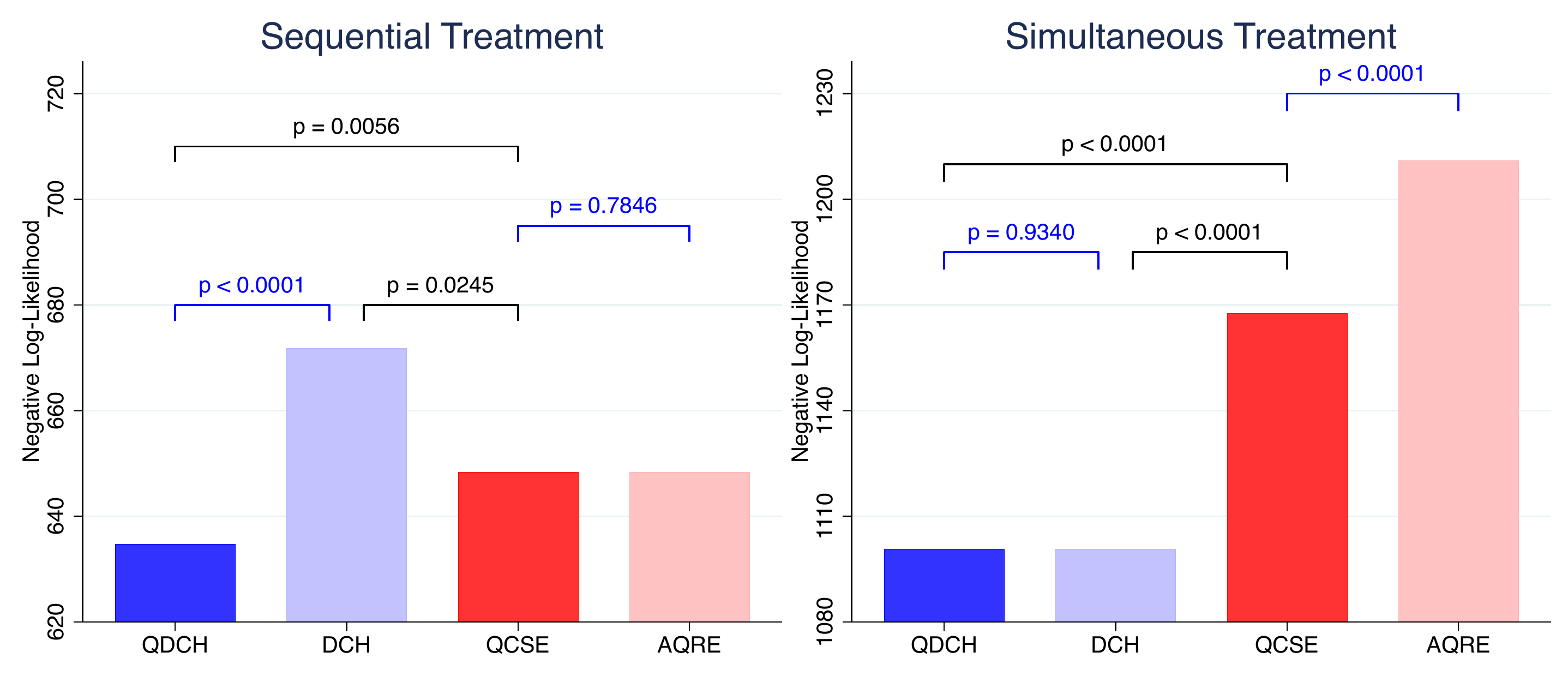}
    \caption{Negative log-likelihood of each model. The likelihood ratio test is
    performed when comparing two nested models, while the Vuong test is performed
    when comparing QDCH and DCH with QCSE.}
    \label{fig:model_comparison}
\end{figure}

\bigskip

\noindent\textbf{Result 4.} \emph{(1) In both the sequential and simultaneous treatments, 
QDCH outperforms QCSE in explaining the data. Additionally, in the sequential treatment,
the fitness of DCH is not significantly different from QCSE and AQRE. In the simultaneous 
treatment, DCH significantly outperforms QCSE and AQRE. (2) In both treatments, there is 
evidence of the failure of Bayesian inferences. Additionally, in the sequential treatment, 
evidence of quantal responses is present, while it is not observed in the simultaneous treatment.}

\subsection{The Analysis of Reaction Times}

Besides the choice data, it is also interesting to see how long it 
takes individuals to make decisions. There is evidence suggesting 
that people tend to take an action faster if they adopt some simple 
decision-making heuristics or have strong preferences over the 
action (see, for example, \cite{rubinstein2007instinctive, chabris2009allocation, 
konovalov2019revealed, lin2020evidence} and \cite{gill2023strategic}).

\begin{figure}[htbp!]
    \centering
    \includegraphics[width=\textwidth]{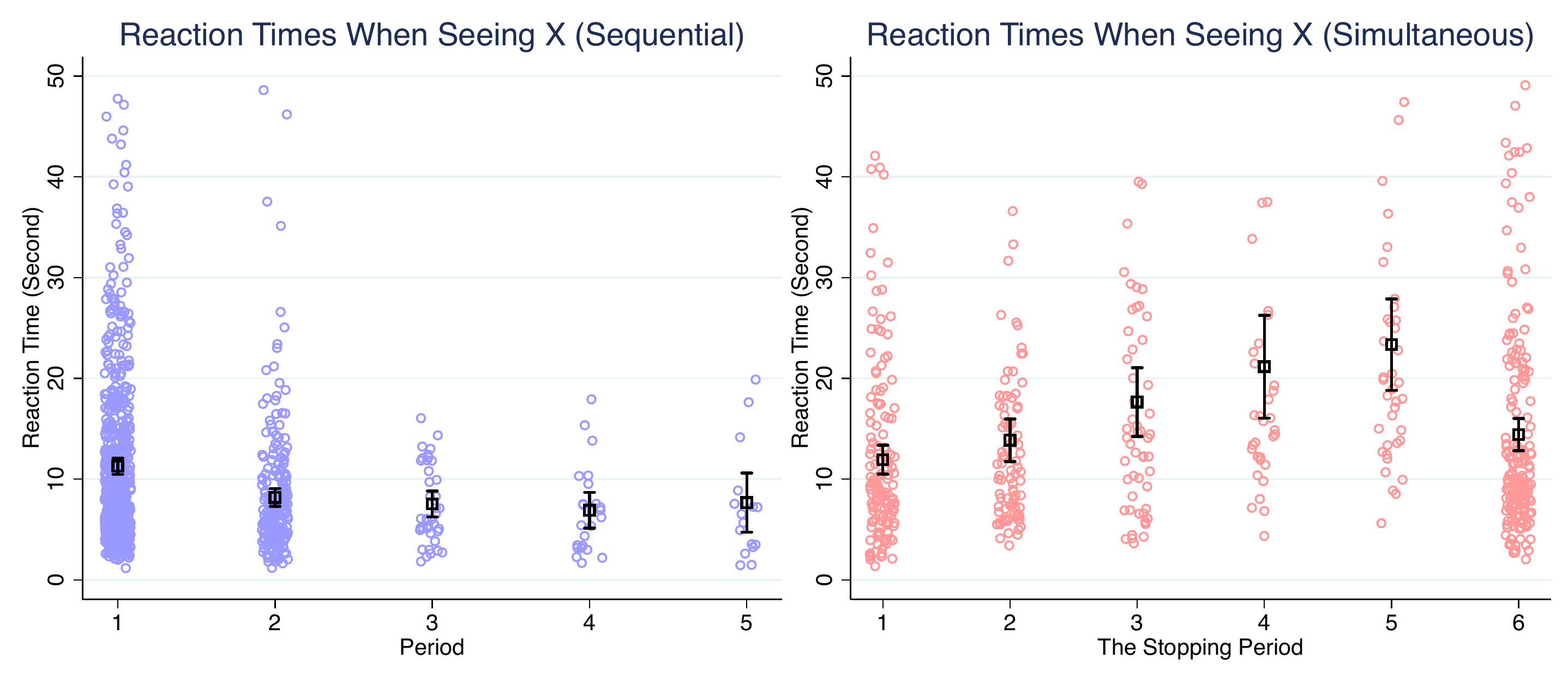}
    \caption{(Left) The reaction time in the sequential treatment. The scatterplot of reaction time 
    conditional on the current period is shown by the blue dots. The mean and the 95\% CIs are overlaid.
    (Right) The reaction time in the simultaneous treatment when seeing a dirty face. 
    The scatterplot of reaction time conditional on the choice of the stopping periods 
    is shown by the red dots. The mean and the 95\% CIs are overlaid. }
    \label{fig:static_rt}
\end{figure}

Focusing on the case where players observe a dirty face, Figure \ref{fig:static_rt} 
presents two panels. The left panel illustrates the distribution of reaction times
at each period of the sequential games. The right panel displays the distribution of 
reaction times for each stopping strategy in the simultaneous games.
In the sequential treatment, we can observe that the reaction 
times at each period are significantly different (Kruskal-Wallis ranksum 
test: $\chi^2(4) = 32.519$ and p-value = $ 0.0001$). Moreover, the reaction 
time decreases as the game progresses to later periods, dropping from 
11.29 seconds in period 1 to 7.66 seconds in period 5.
Combined with the low frequencies of $C$ in later periods (around 21.43\%), we can 
conclude that players quickly decide to wait in later periods.

In the simultaneous treatment, we once again observe that the reaction
times for each stopping strategy are significantly different 
(Kruskal-Wallis ranksum test: $\chi^2(5) = 54.291$ and p-value = $ 0.0001$).
The right panel of Figure \ref{fig:static_rt} reveals 
a monotonic pattern: players take longer to decide to claim in 
later periods, with average reaction time of 11.93 seconds for period 1 
and 23.34 seconds for period 5. However, it only takes players approximately 
14.42 seconds to decide to always wait.

The empirical patterns from both treatments provide suggestive evidence that the heuristic 
of choosing to ``always wait'' differs from the heuristic of claiming at a specific period. 
This finding aligns with the rationale of DCH---level 1 players will always wait upon seeing 
a dirty face, regardless of the payoff configurations. Conversely, higher-level strategic 
players will make inferences to determine their stopping strategies. Lastly, the observed monotonic 
increase in reaction times across stopping strategies in the simultaneous treatment 
aligns with the idea that choosing to claim at later periods requires more steps of reasoning.

\bigskip

\noindent\textbf{Result 5.} \emph{(1) In the sequential treatment, when players see a dirty face, 
their reaction time is shorter in later periods. (2) In the simultaneous treatment,
the reaction time of choosing to claim at some period when players 
see a dirty face is monotonically increasing in the stopping 
periods. However, players take much less time to decide to always 
wait.}

\section{Conclusion}\label{sec:conclusion}

This paper theoretically and experimentally studies the DCH
solution, an alternative model that relaxes the mutual consistency requirement of the 
standard equilibrium theory, in multi-stage games of incomplete information. 
Instead of mutual consistency, DCH posits that players are heterogeneous 
with respect to their levels of sophistication and incorrectly believe others are strictly less sophisticated than
they are. As the dynamic game progresses, strategic players will update their beliefs about others' types and levels.

In this paper, I characterize some general properties of the DCH belief system in multi-stage games 
of incomplete information. Proposition \ref{prop:independence} 
guarantees that the DCH belief system is a product measure across players when every player's payoff-relevant type is 
independently determined. On the other hand, when the prior distribution of types is correlated across players, 
Proposition \ref{prop:corr_type} demonstrates the existence of a unique corresponding game, 
where the types are independently drawn, resulting in the DCH solution being invariant in both games.
While solving the DCH solution does not require a fixed point argument, it could be computationally challenging 
in principle, especially when there are more players or information sets involved. To this end, 
Proposition \ref{prop:independence} and \ref{prop:corr_type} simplify the computation, preserving the 
tractability of DCH. In addition, Proposition \ref{prop:support_evo} shows that strategic players always consider 
the possibility of others being non-strategic, causing the lack of common knowledge of rationality in DCH.

Furthermore, another feature of DCH is the violation of invariance under
strategic equivalence, which arises because level 0 players' behavioral
strategies are not always outcome-equivalent in different strategically equivalent games, 
leading to different behavior of higher-level players. To demonstrate the violation of invariance and contrast 
DCH with the standard equilibrium theory, I characterize the DCH solutions of the sequential and 
simultaneous two-person dirty-faces games. Despite the two versions of the game sharing the same reduced normal form, the DCH 
solutions of the two versions differ in a specific way, as characterized by Proposition \ref{prop:dirty_representation}. 
In summary, DCH predicts that higher-level (level $k\geq 2$) players tend to claim earlier in the sequential version
when they are sufficiently impatient, and vice versa in the simultaneous version when they are patient enough.

To test the predictions of DCH, I design and run a laboratory experiment on two-person dirty-faces 
games where I manipulate both the timing (sequential vs. simultaneous) and payoff structures. 
The experimental design is guided by DCH, wherein I first calibrate the model
using an existing dirty-faces game experimental dataset and choose the payoff parameters to maximize the 
diagnositicity. Considering that the prior distributions of levels might significantly 
vary among different subject pools, this experimental design to some extent serves 
as a stress test for assessing the external validity of DCH.

Overall, a significant treatment effect is detected: players tend 
to claim earlier in the sequential treatment than in the simultaneous 
treatment. Some interesting patterns emerge from the data. First, players' 
behavior significantly varies with payoff structures in both treatments, 
aligning with the qualitative predictions of DCH. Second, players take longer 
to choose higher-level stopping strategies. Third, when comparing the fitness of 
different behavioral models, we find that QDCH outperforms other behavioral models
in both treatments. Moreover, there is some evidence of the failure of 
best responses and Bayesian inferences.

Lastly, when comparing the observed treatment effects with the predictions of the calibrated DCH, 
we find in one of the diagnostic games where $(\delta, \bar{\alpha})=(0.6, 0.45)$, 
the empirical frequency of choosing to claim at period 2 with the observation of 
a dirty face is 30.77\% higher in the sequential treatment than in the 
simultaneous treatment, which is highly consistent with the calibrated DCH (approximately 31.15\%).
However, in all control games and the other diagnostic game, 
the treatment effect is significantly higher than the predictions of the calibrated DCH.
Along with the estimation results, we can conclude that while the observed 
violation of invariance in the data is primarily attributed to the 
relaxation of mutual consistency, it is a joint consequence of 
the relaxation of all equilibrium requirements.


The key contribution of this paper is establishing the theoretical and empirical foundations of the DCH
solution. However, there are considerable extensions and applications that might be fruitful for future research. 
The first extension worth pursuing is to endogenize the levels of sophistication, possibly using the 
cost-benefit analysis proposed by \cite{alaoui2016endogenous}. This extension could be challenging in dynamic games because each player's 
level might vary in different information sets of the same game. Additionally, players not only form beliefs about others' 
current levels but also about their cognitive bounds, which might make the model less tractable.

Second, while the assumption of uniform randomization of level 0 players has some distinct 
advantages, exploring the actual behavior of level 0 
players is another direction worth investigating. Inspired by \cite{li2022predictable}, an alternative 
assumption for level 0 players is that they will randomize across \emph{visually salient} actions at each information set.
In particular, due to the rapid development of machine learning algorithms, how visually salient an action is can be
quantified even before any behavioral data is collected.

Finally, this last section lists several potential applications of dynamic games of incomplete information where the 
mutual consistency requirement is easily violated and the DCH solution might provide some new insights. 
\begin{enumerate}

    \item \emph{Social learning}: In social learning games with repeated actions, 
    players make inferences about the true state based on 
    their private signals and publicly observed actions (see \cite{bala1998learning} 
    and \citealt{harel2021rational}). The DCH solution posits that players 
    do not commonly believe others are able to make correct inferences. 
    Specifically, level 0 players' actions do not convey any information 
    about the true states, while level 1 players will always obey their private 
    signals. For higher-level players, they will constantly update their beliefs about 
    the true state and other players' levels of sophistication. An open question is whether 
    higher-level players will eventually learn the true state.

    \item \emph{Sequential bargaining}: The equilibrium of a sequential bargaining game
    was first characterized by \cite{rubinstein1982perfect}. To reach the perfect equilibrium, 
    players are required to choose the optimal proposal among a continuum of choices at 
    every subgame, and believe the other player to optimally respond to each proposal.
    Later, \cite{mckelvey1993engodeneity, mckelvey1995holdout} 
    considered a two-person multi-stage bargaining game 
    where each players has a private payoff-relevant type and makes a binary decide 
    (whether to give in or hold out) in every period. The game continues until 
    at least one of the players gives in. In this game, it is strictly dominant for
    the strong type of players to hold out forever, but not for the weak type---the weak type players need to 
    trade-off between the reward of giving in earlier and the reputational benefit from mimicking the high type. 
    This reasoning is behaviorally challenging. In contrast, DCH is not a solution of a fixed point problem 
    but solved iteratively from lower to higher levels. Therefore, the DCH solution is expected to 
    be sharply different from the standard equilibrium in the sequential bargaining game.

    \item \emph{Signaling:} In a multi-stage signaling game, an informed player will have a persistent 
    type and interact with an uninformed player repeatedly. \cite{kaya2009repeated} analyzed such an environment, 
    finding that the set of equilibrium signal sequences includes a large class of possibly complex 
    signal sequences. In contrast, in the DCH solution, the uninformed player will learn about the 
    informed player's true type and level when observing a new signal, and the informed player will also 
    learn about the uninformed player's level at each stage. Given that the DCH solution is unique, 
    it would be interesting to characterize the signal sequence of each level of informed 
    players and test whether this is consistent with the behavioral data.

    \item \emph{Sequential voting:} There is a large class of voting rules that includes multiple rounds, 
    such as sequential voting over agendas \citep{baron1989bargaining} or elections based on repeated 
    ballots and elimination of one candidate in each round \citep{bag2009multi}. To reach Condorcet consistent outcomes, 
    players are required to behave strategically. However, in cases where voters are not strategic or believe others
    might not be strategic, the DCH solution becomes an ideal solution concept. 
    In the DCH solution, voters will update their beliefs about others' preferences and levels 
    of sophistication simultaneously, and vote according to their posterior beliefs in each round. 
    Since the common knowledge of rationality is violated in DCH, it is natural to 
    conjecture that higher-level players will vote more sincerely in DCH than in the equilibrium.
\end{enumerate}

\bibliography{reference.bib}

\newpage
\appendix
\section{Omitted Proofs for General Properties}
\label{appendix:proof_general}

\subsection*{Proof of Lemma \ref{lemma_uniqueness}}

The uniqueness of DCH can be proven by induction on levels. That is, it suffices to 
prove that for every level $k\geq 1$, the optimal behavior strategy profile is unique. 
Consider any player $i\in N$ any type $\theta_i\in \Theta_i$.
Level 1 type $\theta_i$ player $i$ believes all other players are level 0, and 
will uniformly randomize at every history. Therefore, at every history $h^t$, level 1
type $\theta_i$ player $i$'s DCH belief is 
$\mu_i^1(\theta_{-i}, (0,...,0) \mid \theta_i, h^t) = \mathcal{F}(\theta_{-i}|\theta_i)$.
By one-deviation principle, level 1 type $\theta_i$ player $i$'s best response satisfies 
that for any history $h^t$ and action $a\in A_i(h^t)$, $\sigma_i^1(a | \theta_i, h^t)>0$ if 
and only if 
$$a \in \arg\max_{a'\in A_i(h^t)} \mathbb{E}u_i^1((\sigma_{-i}^0, \bar{\sigma}_i^1(a')) \mid \theta_i, h^t)$$
where $\bar{\sigma}_i^1$ agrees with $\sigma_i^1$ except at 
$(\theta_i, h^t)$ where $\bar{\sigma}_i^1(a')$ chooses $a'$ with probability 1.
Since players are assumed to uniformly randomize over optimal
actions when they are indifferent, $\sigma_i^1$ is uniquely pinned down.

Suppose there is $K>2$ such that the optimal strategy profiles for level $1$ to $K-1$ are 
unique. In this case, level $K$ player $i$'s conjecture about other players' 
behavior strategy profile $\tilde{\sigma}_{-i}^{-K}$ is also unique and totally mixed.
By one-deviation principle, level $K$ type $\theta_i$ player $i$'s best response satisfies 
that for any history $h^t$ and $a\in A_i(h^t)$, $\sigma_i^K(a | \theta_i, h^t)>0$ if 
and only if 
$$a \in \arg\max_{a'\in A_i(h^t)} \mathbb{E}u_i^K((\tilde{\sigma}_{-i}^{-K}, 
\bar{\sigma}_i^K(a')) \mid \theta_i, h^t).$$
Since players are assumed to uniformly randomize over optimal
actions when they are indifferent, $\sigma_i^K$ is again uniquely pinned down.
This completes the proof. $\blacksquare$

\subsection*{Proof of Proposition \ref{prop:independence}}

To prove this proposition, I first characterize the posterior beliefs in 
Lemma \ref{lemma_belief_closed_form} then prove that the beliefs are 
independent across players if the types are independently drawn.

\begin{lemma}\label{lemma_belief_closed_form}
Consider any multi-stage game with observed actions $\Gamma$,
any $i\in N $, $\theta_i \in \Theta_i$, $h\in \mathcal{H}\backslash\mathcal{H}^T$, 
and every level $k\in \mathbb{N}$.
For every information set $\mathcal{I}_i = (\theta_i, h)$, level $k$ player $i$'s 
belief at $\mathcal{I}_i$ can be characterized as follows.
\begin{itemize}
    \item[1.] Level $k$ player $i$'s prior belief about other players' 
        types and levels are independent. 
        That is, $\mu_i^k(\theta_{-i}, \tau_{-i} | \theta_i, h_{\emptyset})
        = \mathcal{F}(\theta_{-i}|\theta_{i}) \prod_{j\neq i} \hat{P}_{ij}^k
        (\tau_{j})$.
    \item[2.]  For any $1\leq t < T$, and $h^t\in \mathcal{H}^t$, level $k$ player $i$'s 
        belief at information set $(\theta_i, h^t)$ 
        where $h^t = (a^1, \ldots, a^t)$ is
        $$\mu_i^k(\theta_{-i},\tau_{-i}| \theta_i, h^t) =\frac{\mathcal{F}(\theta_{-i}|\theta_i)
        \prod_{j\neq i}\left\{ \hat{P}_{ij}^k(\tau_j)\prod_{l=1}^t\sigma_j^{\tau_j}(a_j^l|\theta_j, h^{l-1}) \right\}}{\sum_{\theta'_{-i}}\sum_{\{\tau'_{-i}:\tau'_j<k \; \forall j\neq i\}}
        \mathcal{F}(\theta'_{-i}|\theta_i)
        \prod_{j\neq i}\left\{ \hat{P}_{ij}^k(\tau'_j)\prod_{l=1}^t\sigma_j^{\tau'_j}(a_j^l|\theta'_j, h^{l-1}) \right\}}.$$
\end{itemize}
\end{lemma}

\bigskip

\noindent\textbf{Proof of Lemma \ref{lemma_belief_closed_form}:}

\bigskip

\noindent \textbf{\emph{1.}} At the beginning of the game, 
the only information available to 
player $i$ is his own type $\theta_i$ and his level of sophistication
$\tau_i = k$.
Therefore, the prior belief is the probability of 
the opponents' types and levels conditional on $\theta_i$ and $\tau_i$,
which is 
\begin{align*}
\mu_i^k(\theta_{-i}, \tau_{-i} | \theta_i, h_{\emptyset}) &= 
    \Pr(\theta_{-i}, \tau_{-i} | \theta_i , \tau_i = k) \\
    & = \Pr(\theta_{-i}| \theta_i)\Pr(\tau_{-i}|\tau_i=k)\\
    & = \mathcal{F}(\theta_{-i}|\theta_i)\prod_{j\neq i}\hat{P}_{ij}^k(\tau_j).
\end{align*}
The second equality holds because the types and levels are 
independently drawn.

\bigskip

\noindent \textbf{\emph{2.}}
This can be shown by induction on $t$. 
Consider any available history at period 2, $h^1\in \mathcal{H}^1$.
Level $k$ player $i$'s belief at information set $(\theta_i, h^1)$ is 
\begin{equation}\tag{A.1}\label{eq:prop_1_recursive_def}
\mu_i^k(\theta_{-i},\tau_{-i}|\theta_i, h^1)
=\frac{\mu_{i}^k(\theta_{-i},\tau_{-i}|\theta_i, h_\emptyset)\prod_{j\neq i} \sigma_{j}^{\tau_j}(a_j^1|\theta_j,h_\emptyset) }{\sum_{\theta'_{-i}}\sum_{\{\tau'_{-i}:\tau'_j<k \; \forall j\neq i\}}\mu_{i}^k(\theta'_{-i},\tau'_{-i}|\theta_i, h_\emptyset)
\prod_{j\neq i}\sigma_{j}^{\tau'_j}(a_j^1|\theta'_j,h_\emptyset)}.      
\end{equation}
By step 1, we know 
$\mu_{i}^k(\theta_{-i},\tau_{-i}|\theta_i, h_\emptyset) = 
\mathcal{F}(\theta_{-i}|\theta_i)\prod_{j\neq i}
\hat{P}_{ij}^k(\tau_j)$. Plugging in Equation (\ref{eq:prop_1_recursive_def}),
we can obtain that 
\begin{align*}
\mu_i^k(\theta_{-i},\tau_{-i}|\theta_i, h^1)
&=\frac{\mu_{i}^k(\theta_{-i},\tau_{-i}|\theta_i, h_\emptyset)\prod_{j\neq i} \sigma_{j}^{\tau_j}(a_j^1|\theta_j,h_\emptyset) }{\sum_{\theta'_{-i}}\sum_{\{\tau'_{-i}:\tau'_j<k \; \forall j\neq i\}}\mu_{i}^k(\theta'_{-i},\tau'_{-i}|\theta_i, h_\emptyset)
\prod_{j\neq i}\sigma_{j}^{\tau'_j}(a_j^1|\theta'_j,h_\emptyset)}\\
&=\frac{\mathcal{F}(\theta_{-i}|\theta_i)\prod_{j\neq i} \left\{\hat{P}_{ij}^k(\tau_j)\sigma_{j}^{\tau_j}(a_j^1|\theta_j,h_\emptyset)\right\} }{\sum_{\theta'_{-i}}\sum_{\{\tau'_{-i}:\tau'_j<k \; \forall j\neq i\}}\mathcal{F}(\theta'_{-i}|\theta_i)
\prod_{j\neq i}\left\{\hat{P}_{ij}^k(\tau'_j)\sigma_{j}^{\tau'_j}(a_j^1|\theta'_j,h_\emptyset)\right\}}.
\end{align*}

Next, suppose there is $t'$ such that the statement holds for 
every period $t=2,\ldots, t'$.
Consider period $t'+1$ and any history available at period $t'+1$, 
$h^{t'}\in \mathcal{H}^{t'}$ where $h^{t'} = (a^1, 
\ldots, a^{t'})$.
Then level $k$ player $i$'s belief at information set
$(\theta_i, h^{t'})$ is 
\begin{align*}
&\mu_i^k(\theta_{-i},\tau_{-i}|\theta_i, h^{t'}) = 
\frac{\mu_{i}^k(\theta_{-i},\tau_{-i}|\theta_i, h^{t'-1})\prod_{j\neq i} \sigma_{j}^{\tau_j}(a_j^{t'}|\theta_j,h^{t'-1}) }{\sum_{\theta'_{-i}}\sum_{\{\tau'_{-i}:\tau'_j<k \; \forall j\neq i\}}\mu_{i}^k(\theta'_{-i},\tau'_{-i}|\theta_i, h^{t'-1})
\prod_{j\neq i}\sigma_{j}^{\tau'_j}(a_j^{t'}|\theta'_j,h^{t'-1})}\\
=& \; \frac{\mathcal{F}(\theta_{-i}|\theta_i)\prod_{j\neq i}\left\{ 
\hat{P}_{ij}^k(\tau_j)\prod_{l=1}^{t'-1}\sigma_j^{\tau_j}(a_j^{l}|\theta_j, h^{l-1})
\right\}\prod_{j\neq i} \sigma_{j}^{\tau_j}(a_j^{t'}|\theta_j,h^{t'-1}) }{\sum_{\theta'_{-i}}\sum_{\{\tau'_{-i}:\tau'_j<k \; \forall j\neq i\}}\mathcal{F}(\theta'_{-i}|\theta_i)\prod_{j\neq i}\left\{ 
\hat{P}_{ij}^k(\tau'_j)\prod_{l=1}^{t'-1}\sigma_j^{\tau'_j}(a_j^{l}|\theta'_j, h^{l-1})
\right\}
\prod_{j\neq i}\sigma_{j}^{\tau'_j}(a_j^{t'}|\theta'_j,h^{t'-1})}\\
=& \; \frac{\mathcal{F}(\theta_{-i}|\theta_i)\prod_{j\neq i}\left\{ 
\hat{P}_{ij}^k(\tau_j)\prod_{l=1}^{t'}\sigma_j^{\tau_j}(a_j^{l}|\theta_j, h^{l-1})
\right\}}{\sum_{\theta'_{-i}}\sum_{\{\tau'_{-i}:\tau'_j<k \; \forall j\neq i\}}\mathcal{F}(\theta'_{-i}|\theta_i)\prod_{j\neq i}\left\{ 
\hat{P}_{ij}^k(\tau'_j)\prod_{l=1}^{t'}\sigma_j^{\tau'_j}(a_j^{l}|\theta'_j, h^{l-1})
\right\}}.
\end{align*}
The second equality holds because of the induction hypothesis,
as desired. $\square$

\bigskip

\noindent\textbf{Proof of Proposition
\ref{prop:independence}:}

\bigskip

\noindent We prove this by induction on $t$. Let $\sigma$ be any level-dependent 
strategy profile and $\mathcal{F}$ and $P$ be any distributions of 
types and levels. First, consider $t=1$. By Lemma \ref{lemma_belief_closed_form}, we know $\mu_i^k(\theta_{-i},\tau_{-i}|\theta_i, h_{\emptyset}) = 
\mathcal{F}(\theta_{-i}|\theta_i)\prod_{j\neq i}\hat{P}^k_{ij}(\tau_j)$. 
As the prior distribution of types is independent across players, 
we can obtain that 
\begin{align*}
    \mu_i^k(\theta_{-i},\tau_{-i}|\theta_i, h_{\emptyset}) &= \mathcal{F}(\theta_{-i}|\theta_i)\prod_{j\neq i}\hat{P}^k_{ij}(\tau_j)\\
    &= \prod_{j\neq i} \mathcal{F}_j(\theta_j) \prod_{j\neq i}\hat{P}^k_{ij}(\tau_j)\\
    & = \prod_{j \neq i} \left[\mathcal{F}_j(\theta_j)\hat{P}^k_{ij}(\tau_j)\right]
    = \prod_{j\neq i}\mu_{ij}^k(\theta_j, \tau_j| \theta_i, h_{\emptyset}).
\end{align*}
Therefore, we know the result is true at $t=1$. Next, suppose there is 
$t'>1$ such that the result holds for all $t=1,\ldots, t'$. 
We want to show that the result holds at period $t'+1$. 
Let $h^{t'} \in \mathcal{H}^{t'}$ be any available history in period 
$t'+1$ where $h^{t'} = (h^{t'-1}, a^{t'})$.
Therefore, player $i$'s posterior belief at history $h^{t'}$ is 
$$\mu_{i}^k(\theta_{-i},\tau_{-i}|\theta_i, h^{t'})
= \frac{\mu_{i}^k(\theta_{-i}, \tau_{-i}| \theta_i, h^{t'-1}) 
\prod_{j\neq i}\sigma_j^{\tau_j}(a_j^{t'}| \theta_j, h^{t'-1})}{\sum_{\theta'_{-i}\in \Theta_{-i}}
\sum_{\{\tau'_{-i}: \tau'_j <k \; \forall j\neq i\}}
\mu_{i}^k(\theta'_{-i}, \tau'_{-i}| \theta_i, h^{t'-1}) 
\prod_{j\neq i}\sigma_j^{\tau'_j}(a_j^{t'}| \theta'_j, h^{t'-1})}.$$
By induction hypothesis, we know 
$$\mu_i^k(\theta_{-i},\tau_{-i}|\theta_i, h^{t'-1}) = \prod_{j\neq i}
\mu_{ij}^k(\theta_j, \tau_j | \theta_i, h^{t'-1}).$$
Therefore, as we rearrange the posterior belief $\mu_i^{k}(\theta_{-i},\tau_{-i}|\theta_i, h^{t'})$, we can obtain that 
\begin{align*}
\mu_{i}^k(\theta_{-i},\tau_{-i}|\theta_i, h^{t'})
&= \frac{\mu_{i}^k(\theta_{-i}, \tau_{-i}| \theta_i, h^{t'-1}) 
\prod_{j\neq i}\sigma_j^{\tau_j}(a_j^{t'}| \theta_j, h^{t'-1})}{\sum_{\theta'_{-i}\in \Theta_{-i}}
\sum_{\{\tau'_{-i}: \tau'_j <k \; \forall j\neq i\}}
\mu_{i}^k(\theta'_{-i}, \tau'_{-i}| \theta_i, h^{t'-1}) 
\prod_{j\neq i}\sigma_j^{\tau'_j}(a_j^{t'}| \theta'_j, h^{t'-1})}\\
& = \frac{\prod_{j\neq i}\left[\mu_{ij}^k(\theta_{j}, \tau_{j}| \theta_i, h^{t'-1}) \sigma_j^{\tau_j}(a_j^{t'}| \theta_j, h^{t'-1})\right]}{\sum_{\theta'_{-i}\in \Theta_{-i}}
\sum_{\{\tau'_{-i}: \tau'_j <k \; \forall j\neq i\}}
\prod_{j\neq i}\left[\mu_{ij}^k(\theta'_{j}, \tau'_{j}| \theta_i, h^{t'-1})\sigma_j^{\tau'_j}(a_j^{t'}| \theta'_j, h^{t'-1})\right]}\\
&= \frac{\prod_{j\neq i}\left[\mu_{ij}^k(\theta_{j}, \tau_{j}| \theta_i, h^{t'-1}) \sigma_j^{\tau_j}(a_j^{t'}| \theta_j, h^{t'-1})\right]}{\sum_{\theta'_{-i}\in \Theta_{-i}}
\prod_{j\neq i}\left[\sum_{\tau'_j <k}\mu_{ij}^k(\theta'_{j}, \tau'_{j}| \theta_i, h^{t'-1})\sigma_j^{\tau'_j}(a_j^{t'}| \theta'_j, h^{t'-1})\right]}\\
&= \frac{\prod_{j\neq i}\left[\mu_{ij}^k(\theta_{j}, \tau_{j}| \theta_i, h^{t'-1}) \sigma_j^{\tau_j}(a_j^{t'}| \theta_j, h^{t'-1})\right]}{
\prod_{j\neq i}\left[\sum_{\theta'_{j}\in \Theta_{j}}\sum_{\tau'_j <k}\mu_{ij}^k(\theta'_{j}, \tau'_{j}| \theta_i, h^{t'-1})\sigma_j^{\tau'_j}(a_j^{t'}| \theta'_j, h^{t'-1})\right]}.
\end{align*}
As a result, we can conclude that 
\begin{align*}
\mu_i^k(\theta_{-i},\tau_{-i}|\theta_i, h^{t'}) &= \prod_{j\neq i}\left[\frac{\mu_{ij}^k(\theta_{j}, \tau_{j}| \theta_i, h^{t'-1}) \sigma_j^{\tau_j}(a_j^{t'}| \theta_j, h^{t'-1})}{
\sum_{\theta'_{j}\in \Theta_{j}}\sum_{\tau'_j <k}\mu_{ij}^k(\theta'_{j}, \tau'_{j}| \theta_i, h^{t'-1})\sigma_j^{\tau'_j}(a_j^{t'}| \theta'_j, h^{t'-1})}\right]\\
& =\prod_{j\neq i }\mu_{ij}^k(\theta_j,\tau_j|\theta_i, h^{t'}).
\end{align*}
This completes the proof of the proposition.  $\blacksquare$

\subsection*{Proof of Proposition \ref{prop:corr_type}}

By Lemma \ref{lemma_belief_closed_form}, we know that
in the transformed game (with independent types) $\hat{\Gamma}$, 
level $k$ player $i$'s belief at $h^t\in \mathcal{H}^t$ is
\begin{align*}
\hat{\mu}_i^k(\theta_{-i},\tau_{-i}| \theta_i, h^t)
&=\frac{\hat{\mathcal{F}}(\theta_{-i}|\theta_i)
\prod_{j\neq i}\left\{ \hat{P}_{ij}^k(\tau_j)\prod_{l=1}^t\sigma_j^{\tau_j}(a_j^l|\theta_j, h^{l-1}) \right\}}{\sum_{\theta'_{-i}}\sum_{\{\tau'_{-i}:\tau'_j<k \; \forall j\neq i\}}
\hat{\mathcal{F}}(\theta'_{-i}|\theta_i)
\prod_{j\neq i}\left\{ \hat{P}_{ij}^k(\tau'_j)\prod_{l=1}^t\sigma_j^{\tau'_j}(a_j^l|\theta'_j, h^{l-1}) \right\}}\\
&=\frac{\prod_{j\neq i}\left\{ \hat{P}_{ij}^k(\tau_j)\prod_{l=1}^t\sigma_j^{\tau_j}(a_j^l|\theta_j, h^{l-1}) \right\}}{\sum_{\theta'_{-i}}\sum_{\{\tau'_{-i}:\tau'_j<k \; \forall j\neq i\}}\prod_{j\neq i}\left\{ \hat{P}_{ij}^k(\tau'_j)\prod_{l=1}^t\sigma_j^{\tau'_j}(a_j^l|\theta'_j, h^{l-1}) \right\}}.
\end{align*}
Therefore, we can obtain that 
\begin{align*}
\mu_i^k(\theta_{-i},\tau_{-i}| \theta_i, h^t)
&=\frac{\mathcal{F}(\theta_{-i}|\theta_i)
\prod_{j\neq i}\left\{ \hat{P}_{ij}^k(\tau_j)\prod_{l=1}^t\sigma_j^{\tau_j}(a_j^l|\theta_j, h^{l-1}) \right\}}{\sum_{\theta'_{-i}}\sum_{\{\tau'_{-i}:\tau'_j<k \; \forall j\neq i\}}
\mathcal{F}(\theta'_{-i}|\theta_i)
\prod_{j\neq i}\left\{ \hat{P}_{ij}^k(\tau'_j)\prod_{l=1}^t\sigma_j^{\tau'_j}(a_j^l|\theta'_j, h^{l-1}) \right\}}\\
&= \frac{\mathcal{F}(\theta_{-i}|\theta_i)
\hat{\mu}_{i}^k(\theta_{-i},\tau_{-i}|\theta_i, h^t)}{\sum_{\theta'_{-i}}\sum_{\{\tau'_{-i}:\tau'_j<k \; \forall j\neq i\}}
\mathcal{F}(\theta'_{-i}|\theta_i)
\hat{\mu}_{i}^k(\theta'_{-i},\tau'_{-i}|\theta_i, h^t)}.
\end{align*}

To complete the proof, it suffices to show that 
for each level $k$ player $i$ and every $h^t\in 
\mathcal{H}\backslash\mathcal{H}^T$, 
maximizing $\mathbb{E}u_i^k$ given belief $\mu_i^k$
and $\sigma_{-i}^{-k}$ is equivalent to 
maximizing $\mathbb{E}\hat{u}_i^k$ given belief 
$\hat{\mu}_i^k$ and $\hat{\sigma}_{-i}^{-k} = \sigma_{-i}^{-k}.$ 
This is true because the expected payoff in the original game
(with correlated types) is:
\begin{align*}
\mathbb{E}u_i^{k}(\sigma&|\theta_i,h^t)=\\   
&\sum_{h^T\in\mathcal{H}^T} \sum_{\theta_{-i}\in \Theta_{-i}}
\sum_{\{\tau_{-i}: \tau_j <k \; \forall j\neq i\}}\mu_i^{k}(\theta_{-i},\tau_{-i}|
\theta_i, h^t) P_i^{k}(h^T|h^{t}, \theta,\tau_{-i},
\sigma_{-i}^{-k},\sigma_i^{k}) 
u_i(h^T,\theta_i,\theta_{-i}),
\end{align*}
which is proportional to 
\begin{align*}
&\mathbb{E}\hat{u}_i^{k}(\sigma|\theta_i,h^t)=\\   
&\sum_{h^T\in\mathcal{H}^T} \sum_{\theta_{-i}\in \Theta_{-i}}
\sum_{\{\tau_{-i}: \tau_j <k \; \forall j\neq i\}}\mathcal{F}(\theta_{-i}|\theta_i)
\hat{\mu}_i^{k}(\theta_{-i},\tau_{-i}|
\theta_i, h^t) P_i^{k}(h^T|h^{t}, \theta,\tau_{-i},
\sigma_{-i}^{-k},\sigma_i^{k}) 
u_i(h^T,\theta_i,\theta_{-i})\\
=&\sum_{h^T\in\mathcal{H}^T} \sum_{\theta_{-i}\in \Theta_{-i}}
\sum_{\{\tau_{-i}: \tau_j <k \; \forall j\neq i\}}
\hat{\mu}_i^{k}(\theta_{-i},\tau_{-i}|
\theta_i, h^t) P_i^{k}(h^T|h^{t}, \theta,\tau_{-i},
\sigma_{-i}^{-k},\sigma_i^{k}) 
\hat{u}_i(h^T,\theta_i,\theta_{-i}).
\end{align*}
This completes the proof of the proposition. $\blacksquare$

\subsection*{Proof of Proposition \ref{prop:support_evo}}

\noindent\textbf{Proof of statement 1:}

\noindent Consider any player $i\in N$, any level $\tau_i$, any type $\theta_i$ and any
non-terminal history $h^t = (h^{t-1}, a^t)\in \mathcal{H}^t \backslash \mathcal{H}^T$. To prove the 
statement, it suffices to show that for any $\tilde{\tau}_{-i}$, if $\tilde{\tau}_{-i} \not\in 
supp_i(\tau_{-i} |\tau_i, \theta_i, h^{t-1})$, then $\tilde{\tau}_{-i} \not\in 
supp_i(\tau_{-i} |\tau_i, \theta_i, h^{t})$.

If $\tilde{\tau}_{-i}\not\in supp_i(\tau_{-i} |\tau_i, \theta_i, h^{t-1})$, then $\mu_i^{\tau_{i}}
(\theta_{-i}, \tilde{\tau}_{-i} | \theta_i, h^{t-1}) = 0$ for any $\theta_{-i}$. By Lemma 
\ref{lemma_belief_closed_form},  we can find that for any $\theta_{-i}$,
\begin{align*}
\mu_i^{\tau_i}(\theta_{-i},\tilde{\tau}_{-i} | \theta_i, h^{t}) = 
\frac{\mu_{i}^{\tau_i}(\theta_{-i},\tilde{\tau}_{-i} | \theta_i, h^{t-1})\prod_{j\neq i} \sigma_{j}^{\tilde{\tau}_j}(a_j^{t}|\theta_j,h^{t-1}) }{\sum_{\theta'_{-i}}\sum_{\{\tau'_{-i}:\tau'_j<k \; \forall j\neq i\}}\mu_{i}^{\tau_i}(\theta'_{-i},\tau'_{-i}|\theta_i, h^{t-1})
\prod_{j\neq i}\sigma_{j}^{\tau'_j}(a_j^{t}|\theta'_j,h^{t-1})}= 0,
\end{align*}
implying that $\mu_i^{\tau_{i}}( \tilde{\tau}_{-i} | \theta_i, h^{t}) = 0$ and hence 
$\tilde{\tau}_{-i} \not\in supp_i(\tau_{-i} |\tau_i, \theta_i, h^{t})$. $\square$

\bigskip
\noindent\textbf{Proof of statement 2:}
\bigskip

\noindent The second statement can be proven by induction on $t$. First, consider $t=1$. 
For any $i\in N$, $\tau_{i}\in \mathbb{N}$ and 
$\theta_i\in \Theta_i$, by Lemma \ref{lemma_belief_closed_form}, 
we know the belief about other players' types and levels is 
$\mu_i^{\tau_{i}}(\theta_{-i},\tau_{-i}|\theta_i, h_\emptyset) = 
\mathcal{F}(\theta_{-i}|\theta_i)\prod_{j\neq i}\hat{P}^{\tau_i}_{ij}(\tau_j)$.
Since $\mathcal{F}$ has full support, for any $\theta_{-i} \in \Theta_{-i}$,
\begin{align*}
\sum_{\{\tau_{-i}: \tau_j < \tau_i\; \forall j\neq i\}}
\mu_{i}^{\tau_i}(\theta_{-i},\tau_{-i}| \theta_i, h_\emptyset) = 
\sum_{\{\tau_{-i}: \tau_j < \tau_i\; \forall j\neq i\}}\mathcal{F}(\theta_{-i}|\theta_i)\prod_{j\neq i}\hat{P}^{\tau_i}_{ij}(\tau_j) 
=\mathcal{F}(\theta_{-i}|\theta_i)>0.
\end{align*}
Hence, the statement is true at period 1. 

Next, suppose there is $t'>1$ such that the result holds 
for all $t=1,\ldots, t'$. We want to show the statement holds 
at period $t'+1$. 
Let $h^{t'}$ be any available history at period $t'+1$ where
$h^{t'} = (h^{t'-1}, a^{t'})$. 
Therefore, player $i$'s posterior belief at $h^{t'}$ is 
$$\mu_{i}^{\tau_i}(\theta_{-i},\tau_{-i}|\theta_i, h^{t'})
= \frac{\mu_{i}^{\tau_i}(\theta_{-i}, \tau_{-i}| \theta_i, h^{t'-1}) 
\prod_{j\neq i}\sigma_j^{\tau_j}(a_j^{t'}| \theta_j, h^{t'-1})}{\sum_{\theta'_{-i}\in \Theta_{-i}}
\sum_{\{\tau'_{-i}: \tau'_j <\tau_i \; \forall j\neq i\}}
\mu_{i}^{\tau_i}(\theta'_{-i}, \tau'_{-i}| \theta_i, h^{t'-1}) 
\prod_{j\neq i}\sigma_j^{\tau'_j}(a_j^{t'}| \theta'_j, h^{t'-1})},$$
which is well-defined because level 0 players are always in the 
support and $\sigma^0_j(a_j^{t'}|\theta'_j, h^{t'-1}) = \frac{1}{|A_j(h^{t'-1})|}>0$ for all $j$.
By induction hypothesis, we know $supp_{i}(\theta_{-i} |\tau_i, \theta_i, h^{t'-1})=\Theta_{-i}$. 
Therefore, as we fix any $\theta_{-i}\in \Theta_{-i}$,
we know $\mu_i^{\tau_i}(\theta_{-i}, (0,\ldots, 0)|\theta_i, h^{t'-1})>0$,
suggesting that $\theta_{-i} \in supp_{i}(\theta_{-i} |\tau_i, \theta_i, 
h^{t'})$ because
\begin{align*}
\mu_{i}^{\tau_i}(\theta_{-i}|\theta_i, h^{t'})
&= \frac{\sum_{\{\tau_{-i}: \tau_j <\tau_i \; \forall j\neq i\}} \mu_{i}^{\tau_i}(\theta_{-i}, \tau_{-i}| \theta_i, h^{t'-1}) 
\prod_{j\neq i}\sigma_j^{\tau_j}(a_j^{t'}| \theta_j, h^{t'-1})}{\sum_{\theta'_{-i}\in \Theta_{-i}}
\sum_{\{\tau'_{-i}: \tau'_j <\tau_i \; \forall j\neq i\}}
\mu_{i}^{\tau_i}(\theta'_{-i}, \tau'_{-i}| \theta_i, h^{t'-1}) 
\prod_{j\neq i}\sigma_j^{\tau'_j}(a_j^{t'}| \theta'_j, h^{t'-1})}\\
&\geq \frac{ \mu_{i}^{\tau_i}(\theta_{-i}, (0,\ldots,0)| \theta_i, h^{t'-1}) 
\prod_{j\neq i}\sigma_j^{0}(a_j^{t'}| \theta_j, h^{t'-1})}{\sum_{\theta'_{-i}\in \Theta_{-i}}
\sum_{\{\tau'_{-i}: \tau'_j <\tau_i \; \forall j\neq i\}}
\mu_{i}^{\tau_i}(\theta'_{-i}, \tau'_{-i}| \theta_i, h^{t'-1}) 
\prod_{j\neq i}\sigma_j^{\tau'_j}(a_j^{t'}| \theta'_j, h^{t'-1})}\\
& = \frac{ \mu_{i}^{\tau_i}(\theta_{-i}, (0,\ldots,0)| \theta_i, h^{t'-1}) 
\prod_{j\neq i}\frac{1}{|A_j(h^{t'-1})|}}{\sum_{\theta'_{-i}\in \Theta_{-i}}
\sum_{\{\tau'_{-i}: \tau'_j <\tau_i \; \forall j\neq i\}}
\mu_{i}^{\tau_i}(\theta'_{-i}, \tau'_{-i}| \theta_i, h^{t'-1}) 
\prod_{j\neq i}\sigma_j^{\tau'_j}(a_j^{t'}| \theta'_j, h^{t'-1})}
&> 0.
\end{align*}
This completes the proof of the proposition. $\blacksquare$


\newpage
\section{Omitted Proofs for Two-Person Dirty-Faces Games (For Online Publication)}
\label{appendix:proof_dirty}

\subsection*{Proof of Proposition \ref{prop_extensive_dirty}}

\noindent\textbf{\emph{Step 1:}} Consider any $i\in N$. If $x_{-i}=O$, 
then player $i$ knows his face is dirty immediately. Therefore, 
$C$ is a dominant strategy, suggesting $\sigma_i^k(t, O)=1$ for all $k\geq 1$
and $1\leq t \leq T$. If $x_{-i}=X$, 
player $i$'s belief of having a dirty face at period 1 is 
$p$. Hence, the expected payoff of choosing $C$ at period 1 is 
$p\alpha - (1-p)<0$, implying $\sigma_i^k(1,X) =0$
for all $k\geq 1$.
Finally, since level 1 players believe the other player's actions don't convey any 
information about their own face types,
the expected payoff of $C$ at each period is $p\alpha - (1-p)<0$, implying 
$\sigma_i^1(t,X)=0$ for any $1\leq t \leq T$.

\bigskip

\noindent \textbf{\emph{Step 2:}} 
Consider any level $k\geq 2$, and period $2\leq t \leq T$. This step
characterizes the DCH posterior belief when $x_{-i}=X$.
When the game proceeds to period $t$, the posterior belief of 
$(x_i, \tau_{-i}) = (f,l)$ for any $f\in\{O,X\}$ and $0\leq l \leq k-1$ is:
\begin{equation}\tag{A.2}\label{eq:post_belief}
    \mu_i^k(f,l|t,X) = \frac{\left[\prod_{t'=1}^{t-1}(1-\sigma_{-i}^l(t',f))\right]p_l\Pr(f)}
    {\sum_{x \in \{O,X\}}\sum_{j=0}^{k-1}
    \left[\prod_{t'=1}^{t-1}(1-\sigma_{-i}^j(t',x))\right]p_j\Pr(x)}.
\end{equation}
By step 1, since strategic players will claim immediately when seeing a clean face, 
$\sigma_{-i}^l(t',O)=1$ for all $1\leq t' \leq t-1$. 
Therefore, as the game proceeds to period $t\geq 2$, level $k$ player $i$ would know that 
\emph{it is impossible for the other player to observe a dirty face and have a positive 
level of sophistication at the same time.}
Furthermore, let $\mathcal{M}_i^k(t)$ be the support of level $k$ player $i$'s marginal belief of 
$\tau_{-i}$ at period $t$. For any $0\leq l \leq k-1$,
$$l\in \mathcal{M}_i^k(t) \iff \sum_{x_i\in\{O,X\}}
\prod_{t'=1}^{t-1}(1-\sigma_{-i}^l(t',x_i))>0, $$
and we let $\mathcal{M}_{i+}^k(t) \equiv \mathcal{M}_i^k(t) \backslash \{0\} $.
Therefore, equation (\ref{eq:post_belief}) implies that for any $t\geq 2$,
\begin{align*}
    &\mu_i^k(X,0|t,X) = \frac{\left(\frac{1}{2}\right)^{t-1}p p_0}{\left(\frac{1}{2}\right)^{t-1} p_0 + p \sum_{j\in \mathcal{M}_{i+}^k(t)}p_j},
    &\mu_i^k(O,0|t,X) = \frac{\left(\frac{1}{2}\right)^{t-1}(1-p) p_0}{\left(\frac{1}{2}\right)^{t-1} p_0 + p \sum_{\in \mathcal{M}_{i+}^k(t)}p_j}.    
\end{align*}
Moreover, for any $1\leq k'\leq k-1$, $\mu_i^k(O,k'|t,X) = 0,$ and 
for any $l \in \mathcal{M}_{i+}^k(t)$,
\begin{align*}
    \mu_i^k(X,l|t,X) = \frac{p p_{l}}{\left(\frac{1}{2}\right)^{t-1} p_0 + 
    p \sum_{j\in \mathcal{M}_{i+}^k(t)}p_j}.
\end{align*}
Consequently, the marginal belief of having a dirty face at period 
$2\leq t \leq T$ is: 
\begin{align*}
\mu_{i}^k(X|t,X) = \sum_{j=0}^{k-1} \mu_{i}^k(X,j|t,X)
= \frac{p\left[\left(\frac{1}{2}\right)^{t-1} p_0 + \sum_{j\in \mathcal{M}_{i+}^k(t)}p_j \right]}{\left(\frac{1}{2}\right)^{t-1} p_0 + p \sum_{j\in \mathcal{M}_{i+}^k(t)}p_j}.
\end{align*}
Thus, the expected payoff of choosing $C$ at period $t$ is 
$\delta^{t-1}\left[(1+\alpha)\mu_i^k(X|t,X) - 1\right]$,
which equals to $\mathbb{E}u_i^k(C|t, X)=$
\begin{equation}\tag{A.3}\label{eq:general_expected_payoff}
\frac{\delta^{t-1}}{\left(\frac{1}{2}\right)^{t-1}p_0 
            +p\sum_{j\in \mathcal{M}_{i+}^k(t)}p_j}
            \left\{p\alpha \left[\left(\frac{1}{2}\right)^{t-1}p_0 + \sum_{j\in \mathcal{M}_{i+}^k(t)}p_j\right]-(1-p)\left[\left(\frac{1}{2}\right)^{t-1}p_0 \right] \right\}.    
\end{equation}
Finally, at period $t$, level $k$ player $i$ believes the other player will wait with probability
\begin{equation}\tag{A.4}\label{eq:general_u_prob}
    \frac{1}{2}\mu_i^k(0|t,X) + \sum_{j\in \mathcal{M}_{i+}^k(t+1)} \mu_i^k(j|t,X) =  
    \frac{\left(\frac{1}{2}\right)^{t}p_0 + p\sum_{j\in \mathcal{M}_{i+}^k(t+1)}p_j}{\left(\frac{1}{2}\right)^{t-1}p_0 + p\sum_{j\in \mathcal{M}_{i+}^k(t)}p_j}.
\end{equation}

\bigskip

\noindent\textbf{\emph{Step 3:}} This step proves a monotonicity result: if 
$\sigma_i^k(t,X)=1$, then $\sigma_i^{k+1}(t,X)=1$ for any 
$k\geq 2$ and $2\leq t\leq T$. The proof consists of two cases. First consider 
period $T$. Equation (\ref{eq:general_expected_payoff}) implies $\sigma_i^k(T,X)=1$ 
if and only if 
\begin{align*}
&\frac{\delta^{T-1}}{\left(\frac{1}{2}\right)^{T-1}p_0 +p\sum_{j\in \mathcal{M}_{i+}^k(T)} p_j}\left\{p\alpha \left[\left(\frac{1}{2}\right)^{T-1}p_0 + \sum_{j\in \mathcal{M}_{i+}^k(T)} p_j\right]-(1-p)\left[\left(\frac{1}{2}\right)^{T-1}p_0 \right] \right\}\geq 0\\
&\qquad\qquad\qquad\qquad\qquad\qquad\qquad\qquad\qquad\qquad\qquad\;
\iff \bar{\alpha} \; \geq \; \frac{\left(\frac{1}{2}\right)^{T-1}p_0}{\left(\frac{1}{2}\right)^{T-1}p_0 + \sum_{j\in \mathcal{M}_{i+}^k(T)} p_j}.
\end{align*}
Because $ \mathcal{M}_{i}^k(T) \subseteq\mathcal{M}_{i}^{k+1}(T)$, it can be proven that 
\begin{align*}
   \bar{\alpha} \; \geq \;  \frac{\left(\frac{1}{2}\right)^{T-1}p_0}{\left(\frac{1}{2}\right)^{T-1}p_0 + \sum_{j\in \mathcal{M}_{i+}^k(T)} p_j} 
   \; \geq \; \frac{\left(\frac{1}{2}\right)^{T-1}p_0}{\left(\frac{1}{2}\right)^{T-1}p_0 + \sum_{j\in \mathcal{M}_{i+}^{k+1}(T)} p_j},
\end{align*}
implying that if it is optimal for level $k$ player $i$ to claim at period $T$, 
it is also optimal for level $k+1$ player $i$ to claim.

Second, consider any period $2\leq t \leq T-1$. Note that if
level $k$ players would choose $C$ at period $t$, $k\not\in \mathcal{M}_i^{k+1}(t+1)$ and hence
$\mathcal{M}_{i+}^{k}(t') = \mathcal{M}_{i+}^{k+1}(t')$ for any $t+1\leq t' \leq T$.
Therefore, as the game proceeds beyond period $t$, level $k$ and level $k+1$
players will have the same continuation value. Let $V^{\tilde{k}}_{\tilde{t}}$ 
be level $\tilde{k}$ players' continuation value 
at period $\tilde{t}$. The observation implies $V_{t+1}^k = V_{t+1}^{k+1}$.
Coupled with that $\mathcal{M}_{i+}^{k}(t) \subseteq \mathcal{M}_{i+}^{k+1}(t)$, 
level $k+1$ player $i$'s expected payoff of choosing $W$ at period $t$ satisfies
\begin{align*}
    \frac{\left(\frac{1}{2}\right)^{t}p_0 + p\sum_{j\in \mathcal{M}_{i+}^{k+1}(t+1)}p_j}{\left(\frac{1}{2}\right)^{t-1}p_0 + p\sum_{j\in \mathcal{M}_{i+}^{k+1}(t)}p_j} V^{k+1}_{t+1} \leq
    \frac{\left(\frac{1}{2}\right)^{t}p_0 + p\sum_{j\in \mathcal{M}_{i+}^{k}(t+1)}p_j}{\left(\frac{1}{2}\right)^{t-1}p_0 + p\sum_{j\in \mathcal{M}_{i+}^{k}(t)}p_j} V^k_{t+1},
\end{align*}
where the RHS is level $k$ player's expected payoff of choosing $W$ at period $t$.
The inequality shows level $k$ player's expected payoff of 
choosing $W$ is weakly higher than level $k+1$ player's expected payoff of choosing $W$. 
To complete the proof, it suffices to argue that level $k+1$ player's expected payoff
of $C$ at period $t$ is higher than level $k$ player's expected payoff of $C$. 
This is true because $\mathcal{M}_{i+}^{k}(t) \subseteq \mathcal{M}_{i+}^{k+1}(t)$ implies 
$\mu_i^{k+1}(X|t,X) \geq \mu_i^{k}(X|t,X)$, and hence, 
\begin{align*}
\delta^{t-1}\left[(1+\alpha)\mu_i^{k+1}(X|t,X) - 1\right] \geq \delta^{t-1}\left[(1+\alpha)\mu_i^{k}(X|t,X) - 1\right].
\end{align*}

\bigskip

\noindent \textbf{\emph{Step 4:}} The proposition is proven by induction on $k$. 
This step establishes the base case for level 2 players.
By step 1, $\sigma_i^1(t,X)=0$ for all $1\leq t\leq T$, and hence 
$\mathcal{M}_{i+}^2(t)=\{1 \}$ for all $1\leq t\leq T$.
Therefore, equation (\ref{eq:general_expected_payoff}) suggests the 
expected payoff of $C$ at period $T$ is 
$$ \mathbb{E}u_i^2(C|T,X) = 
\frac{\delta^{T-1}}{\left(\frac{1}{2}\right)^{T-1}p_0 +pp_1}
        \left\{p\alpha \left[\left(\frac{1}{2}\right)^{T-1}p_0 + p_1\right]
            -(1-p)\left[\left(\frac{1}{2}\right)^{T-1}p_0 \right] \right\}.$$
Therefore, $C$ is optimal at period $T$ if and only if 
\begin{align*}
    \mathbb{E}u_i^2(C|T,X) \geq 0 \iff
    \bar{\alpha} \; \geq \; \frac{\left(\frac{1}{2}\right)^{T-1}p_0}{\left(\frac{1}{2}\right)^{T-1}p_0 + p_1}.
\end{align*}

For any period $2\leq t\leq T-1$, I first prove the direction of necessity.
Equation (\ref{eq:general_u_prob}) implies level 2 player $i$'s belief about the 
other player choosing $W$ at period $t$ is:
\begin{align*}
    \frac{1}{2}\mu_i^2(0 | t,X) + \mu_i^2(1|t,X) = \frac{\left(\frac{1}{2}\right)^{t}p_0 +pp_1}{\left(\frac{1}{2}\right)^{t-1}p_0 +pp_1}. 
\end{align*}
Therefore, the expected payoff of $W$ at period $t$ is \emph{at least}
$\left[\frac{\left(\frac{1}{2}\right)^{t}p_0 +pp_1}{\left(\frac{1}{2}\right)^{t-1}p_0 +pp_1} \right]\mathbb{E}u_i^2(C|t+1, X) = $
\begin{align*}
    \frac{\delta^{t}}{\left(\frac{1}{2}\right)^{t-1}p_0 +pp_1}
            \left\{p\alpha \left[\left(\frac{1}{2}\right)^{t}p_0 + p_1\right]
            -(1-p)\left[\left(\frac{1}{2}\right)^{t}p_0 \right] \right\}.
\end{align*}
Since $W$ is always available, 
$C$ is strictly dominated at period $t$ for level 2 player $i$ if 
\begin{align*}
    \mathbb{E}u_i^2(C|t,X)<
    \left[\frac{\left(\frac{1}{2}\right)^{t}p_0 +pp_1}{\left(\frac{1}{2}\right)^{t-1}p_0 +pp_1} \right]&\mathbb{E}u_i^2(C|t+1, X)\\
    \iff & \bar{\alpha} \; < \; \frac{\left[ \left(\frac{1}{2}\right)^{t-1} - \left(\frac{1}{2}\right)^{t}\delta \right]p_0}{\left[ \left(\frac{1}{2}\right)^{t-1} - \left(\frac{1}{2}\right)^{t}\delta \right]p_0 + (1-\delta)p_1}.
\end{align*}
This proves the direction of necessity.

Second, the sufficiency is proven by induction on the periods. Namely, 
I show the sufficiency holds for any period $T-t'$ where $1\leq t' \leq T-2$.
Consider the base case for period $T-1$.
Because
\begin{align*}
    \bar{\alpha} \; \geq \; \frac{\left[ \left(\frac{1}{2}\right)^{T-2} - \left(\frac{1}{2}\right)^{T-1}\delta \right]p_0}{\left[ \left(\frac{1}{2}\right)^{T-2} - \left(\frac{1}{2}\right)^{T-1}\delta \right]p_0 + (1-\delta)p_1} 
    \;> \; \frac{\left(\frac{1}{2}\right)^{T-1}p_0}{\left(\frac{1}{2}\right)^{T-1}p_0 + p_1},
\end{align*}
level 2 players will choose $C$ at period $T$, 
so it is optimal to choose $C$ at period $T-1$ if
\begin{align*}
    \mathbb{E}u_i^2(C|T-1,X)\geq
    \left[\frac{\left(\frac{1}{2}\right)^{T-1}p_0 +pp_1}{\left(\frac{1}{2}\right)^{T-2}p_0 +pp_1} \right]&\mathbb{E}u_i^2(C|T,X)\\
    &\iff \bar{\alpha} \geq \frac{\left[ \left(\frac{1}{2}\right)^{T-2} - \left(\frac{1}{2}\right)^{T-1}\delta \right]p_0}{\left[ \left(\frac{1}{2}\right)^{T-2} - \left(\frac{1}{2}\right)^{T-1}\delta \right]p_0 + (1-\delta)p_1}.
\end{align*}

Now, suppose there is $t'\leq T-2$ such that the 
statement holds at any period $T-t$ where $1\leq t\leq t'-1$.
It can be proven that the sufficiency also holds at period $T-t'$. Because 
\begin{align*}
    \bar{\alpha} \;\geq\; \frac{\left[ \left(\frac{1}{2}\right)^{T-t'-1} - \left(\frac{1}{2}\right)^{T-t'}\delta \right]p_0}{\left[ \left(\frac{1}{2}\right)^{T-t'-1} - \left(\frac{1}{2}\right)^{T-t'}\delta \right]p_0 + (1-\delta)p_1} 
    \;>\; \frac{\left[ \left(\frac{1}{2}\right)^{T-t'} - \left(\frac{1}{2}\right)^{T-t'+1}\delta \right]p_0}{\left[ \left(\frac{1}{2}\right)^{T-t'} - \left(\frac{1}{2}\right)^{T-t'+1}\delta \right]p_0 + (1-\delta)p_1},
\end{align*}
level 2 players will choose $C$ at period $T-t'+1$ by induction 
hypothesis. Therefore, it is optimal to choose $C$ at period $T-t'$ if
\begin{align*}
    \mathbb{E}u_i^2(C|T-t', X)\geq
    \left[\frac{\left(\frac{1}{2}\right)^{T-t'}p_0 +pp_1}{\left(\frac{1}{2}\right)^{T-t'-1}p_0 +pp_1} \right]&\mathbb{E}u_i^2(C|T-t'+1, X)\\
    \iff &\bar{\alpha} \; \geq \; \frac{\left[ \left(\frac{1}{2}\right)^{T-t'-1} - \left(\frac{1}{2}\right)^{T-t'}\delta \right]p_0}{\left[ \left(\frac{1}{2}\right)^{T-t'-1} - \left(\frac{1}{2}\right)^{T-t'}\delta \right]p_0 + (1-\delta)p_1} .
\end{align*}
This completes the proof of sufficiency.

\bigskip

\noindent\textbf{\emph{Step 5:}} 
Step 4 establishes the base case for $k=2$. 
Now, suppose there is $K>2$ such that 
the statement holds for all $2\leq k \leq K$.
It suffices to prove the statement holds for level $K+1$ players. 
The proof for period $T$ is straightforward. 
From step 3, we know if $\sigma_i^{K}(T,X)=1$, then $\sigma_i^{K+1}(T,X)=1$.
Hence, the only case that needs to be considered is when  
$$\bar{\alpha} \; < \; \frac{\left(\frac{1}{2}\right)^{T-1}p_0}{\left(\frac{1}{2}\right)^{T-1}p_0 + \sum_{j=1}^{K-1}p_j}.$$
By induction hypothesis, $\sigma_{-i}^l(t,X)=0$ for all $1\leq l \leq K$ 
and for all $1\leq t\leq T$. Therefore, $\sigma_{i}^{K+1}(T,X)=1$ if and only if 
$\mathbb{E}u_i^{K+1}(C|T,X)\geq 0$, which is equivalent to 
\begin{align*}
    \bar{\alpha} \; \geq \; \frac{\left(\frac{1}{2}\right)^{T-1}p_0}{\left(\frac{1}{2}\right)^{T-1}p_0 + \sum_{j=1}^K p_j} .
\end{align*}
 
For any period $2\leq t\leq T-1$, I first prove the direction of necessity. If
$$\bar{\alpha} \;  < \; \frac{\left[ \left(\frac{1}{2}\right)^{t-1} - \left(\frac{1}{2}\right)^{t}\delta \right]p_0}{\left[ \left(\frac{1}{2}\right)^{t-1} - \left(\frac{1}{2}\right)^{t}\delta \right]p_0 + (1-\delta)\sum_{j=1}^{K}p_j},$$
then by induction hypothesis, $\sigma_{-i}^l(t',X)=0$ for all $1\leq l \leq K$ and $1\leq t'\leq t$, 
implying that $\mathcal{M}_{i+}^{K+1}(t) = \{1,\ldots, K \}$. Then 
equation (\ref{eq:general_expected_payoff}) suggests 
that the expected payoff of $C$ at period $t$ is 
\begin{align*}
     \frac{\delta^{t-1}}{\left(\frac{1}{2}\right)^{t-1}p_0 +p\sum_{j=1}^K p_j}\left\{p\alpha \left[\left(\frac{1}{2}\right)^{t-1}p_0 + \sum_{j=1}^K p_j\right]-(1-p)\left[\left(\frac{1}{2}\right)^{t-1}p_0 \right] \right\}. 
\end{align*} 
Furthermore, equation (\ref{eq:general_u_prob}) suggests level $K+1$ players 
believe the other player will wait at period $t$ with probability 
$$\frac{1}{2}\mu_i^{K+1}(0|t,X) + \sum_{j=1}^K \mu_i^{K+1}(l|t,X)
= \frac{\left(\frac{1}{2}\right)^{t}p_0 +p\sum_{j=1}^{K}p_j}{\left(\frac{1}{2}\right)^{t-1}p_0 +p\sum_{j=1}^{K}p_j}.$$
Therefore, by similar calculation as in step 4, choosing $C$ is strictly dominated if
\begin{align*}
&\frac{\delta^{t-1}}{\left(\frac{1}{2}\right)^{t-1}p_0 +p\sum_{j=1}^{K}p_j}
            \left\{p\alpha \left[\left(\frac{1}{2}\right)^{t-1}p_0 + \sum_{j=1}^{K}p_j\right]
            -(1-p)\left[\left(\frac{1}{2}\right)^{t-1}p_0 \right] \right\}\\
<\;&  \frac{\delta^{t}}{\left(\frac{1}{2}\right)^{t-1}p_0 +p\sum_{j=1}^{K}p_j}
            \left\{p\alpha \left[\left(\frac{1}{2}\right)^{t}p_0 + \sum_{j=1}^{K}p_j\right]
            -(1-p)\left[\left(\frac{1}{2}\right)^{t}p_0 \right] \right\},
\end{align*}
which is implied by  
$$\bar{\alpha} \; < \; \frac{\left[ \left(\frac{1}{2}\right)^{t-1} - \left(\frac{1}{2}\right)^{t}\delta \right]p_0}{\left[ \left(\frac{1}{2}\right)^{t-1} - \left(\frac{1}{2}\right)^{t}\delta \right]p_0 + (1-\delta)\sum_{j=1}^{K}p_j}.$$
This proves the direction of necessity.

Second, the sufficiency is proven by induction on the periods. Namely, 
I show the sufficiency holds for any period $T-t'$ where $1\leq t' \leq T-2$.
Consider the base case for period $T-1$. By step 3,
if $\sigma_i^K(T-1,X)=1$, then $\sigma_i^{K+1}(T-1,X)=1$. Therefore,
it suffices to consider the case where 
\begin{align*}
\frac{\left[ \left(\frac{1}{2}\right)^{T-2} - \left(\frac{1}{2}\right)^{T-1}\delta \right]p_0}{\left[ \left(\frac{1}{2}\right)^{T-2} - \left(\frac{1}{2}\right)^{T-1}\delta \right]p_0 + (1-\delta)\sum_{j=1}^{K}p_j} \; \leq \; \bar{\alpha} 
\; < \; \frac{\left[ \left(\frac{1}{2}\right)^{T-2} - \left(\frac{1}{2}\right)^{T-1}\delta \right]p_0}{\left[ \left(\frac{1}{2}\right)^{T-2} - \left(\frac{1}{2}\right)^{T-1}\delta \right]p_0 + (1-\delta)\sum_{j=1}^{K-1}p_j}.
\end{align*}
By induction hypothesis, $\sigma_{-i}^{l}(t,X)=0$ for all $1\leq t\leq T-1$ 
and $1\leq l \leq K$.
Moreover, $\sigma_i^{K+1}(T,X)=1$ because
$$\frac{\left[ \left(\frac{1}{2}\right)^{T-2} - \left(\frac{1}{2}\right)^{T-1}\delta \right]p_0}{\left[ \left(\frac{1}{2}\right)^{T-2} - \left(\frac{1}{2}\right)^{T-1}\delta \right]p_0 + (1-\delta)\sum_{j=1}^{K}p_j} 
> \frac{\left(\frac{1}{2}\right)^{T-1}p_0}{\left(\frac{1}{2}\right)^{T-1}p_0 + \sum_{j=1}^K p_j}.$$
Therefore, by a similar calculation as in step 4, 
it can be proven that it is optimal for level $K+1$ players to choose $C$ at period $T-1$ if
\begin{align*}
    \mathbb{E}u_i^{K+1}(C|T-1, X)\geq
    \left[\frac{\left(\frac{1}{2}\right)^{T-1}p_0 +p\sum_{j=1}^{K}p_j}{\left(\frac{1}{2}\right)^{T-2}p_0 +p\sum_{j=1}^{K}p_j} \right]&\mathbb{E}u_i^{K+1}(C|T, X)\\
    \iff \bar{\alpha} \; \geq \; &\frac{\left[ \left(\frac{1}{2}\right)^{T-2} - \left(\frac{1}{2}\right)^{T-1}\delta \right]p_0}{\left[ \left(\frac{1}{2}\right)^{T-2} - \left(\frac{1}{2}\right)^{T-1}\delta \right]p_0 + (1-\delta)\sum_{j=1}^{K}p_j}.
\end{align*}

Now, suppose there is $t'\leq T-2$ such that the statement holds for any period $T-t$ where
$1\leq t\leq t'-1$. It can be proven that the statement also holds at period  
period $T-t'$.
By step 3, if $\sigma_i^K(T-t',X)=1$, then $\sigma_i^{K+1}(T-t',X)=1$ and 
it suffices to consider the case:
\begin{align*}
&\frac{\left[ \left(\frac{1}{2}\right)^{T-t'-1} - \left(\frac{1}{2}\right)^{T-t'}\delta \right]p_0}{\left[ \left(\frac{1}{2}\right)^{T-t'-1} - \left(\frac{1}{2}\right)^{T-t'}\delta \right]p_0 + (1-\delta)\sum_{j=1}^{K}p_j} \;  \leq \; \bar{\alpha} \\
&\qquad\qquad\qquad\qquad\qquad\qquad\qquad\qquad\qquad < \;\frac{\left[ \left(\frac{1}{2}\right)^{T-t'-1} - \left(\frac{1}{2}\right)^{T-t'}\delta \right]p_0}{\left[ \left(\frac{1}{2}\right)^{T-t'-1} - \left(\frac{1}{2}\right)^{T-t'}\delta \right]p_0 + (1-\delta)\sum_{j=1}^{K-1}p_j} .
\end{align*}
By induction hypothesis, $\sigma_{-i}^{l}(t,X)=0$ for all $1\leq t\leq T-t'$ and $1\leq l \leq K$, and $\sigma_i^{K+1}(T-t'+1, X)=1$. Therefore, by a similar calculation as in step 4, 
it can be proven that it is optimal for level $K+1$ players to choose $C$ at period $T-t'$ if
$$\bar{\alpha} \; \geq \; \frac{\left[ \left(\frac{1}{2}\right)^{T-t'-1} - \left(\frac{1}{2}\right)^{T-t'}\delta \right]p_0}{\left[ \left(\frac{1}{2}\right)^{T-t'-1} - \left(\frac{1}{2}\right)^{T-t'}\delta \right]p_0 + (1-\delta)\sum_{j=1}^{K}p_j}.$$
This completes the proof of the proposition. 
$\blacksquare$

\subsection*{Proof of Corollary \ref{coro_extensive_dirty}}

By Proposition \ref{prop_extensive_dirty}, we know $\sigma_i^k(t,O)=1$ for all $t$ 
and $k\geq 1$, and $\sigma_i^1(t,X) = 0$ for all $t$. Then by Definition \ref{def_optimal_stopping}, we can 
obtain that $\hat{\sigma}_i^k(O)=1$ for all $k\geq 1$,
and $\hat{\sigma}_i^1(X)=T+1$.  
In addition, since $\sigma_i^k(1,X) = 0$ for all $k\geq 2$, $\hat{\sigma}_i^k(X)\neq 1$.
Moreover, the DCH solution can be equivalently characterized by  
optimal stopping periods because for any $t\geq 2$ and $k\geq 2$,
\begin{align*}
\hat{\sigma}_i^k(X) = t &\iff \sigma_i^k(t-1,X)= 0 \;\;
    \mbox{ and }\;\; \sigma_i^k(t,X)= 1,    \\
\hat{\sigma}_i^k(X) = T+1 &\iff \sigma_i^k(t',X)=0 \;\; \mbox{ for any } 1\leq t'\leq T.
\end{align*}

Lastly, to show the monotonicity, it suffices to show that for any $k' > k \geq 2$ and any $2\leq t \leq T$, 
if $\sigma_i^k(t,X)= 1$, then $\sigma_i^{k'}(t,X)= 1$. The discussion is separated into two cases. 
First, if $t=T$, then by Proposition \ref{prop_extensive_dirty}, $\sigma_i^k(T,X)= 1$ suggests 
\begin{align*}
\bar{\alpha} \; \geq \; \frac{\left(\frac{1}{2}\right)^{T-1}p_0}{\left(\frac{1}{2}\right)^{T-1}p_0 + \sum_{j=1}^{k-1}p_j}
\; > \; \frac{\left(\frac{1}{2}\right)^{T-1}p_0}{\left(\frac{1}{2}\right)^{T-1}p_0 + \sum_{j=1}^{k'-1}p_j},
\end{align*}
implying $\sigma_i^{k'}(T,X)= 1$. Second, if $2\leq t \leq T-1$, 
by Proposition \ref{prop_extensive_dirty}, $\sigma_i^k(t,X)= 1$ suggests 
\begin{align*}
    \bar{\alpha} \; \geq \;  \frac{\left[ \left(\frac{1}{2}\right)^{t-1} - \left(\frac{1}{2}\right)^{t}\delta \right]p_0}{\left[ \left(\frac{1}{2}\right)^{t-1} - \left(\frac{1}{2}\right)^{t}\delta \right]p_0 + (1-\delta)\sum_{j=1}^{k-1}p_j}  
    \; > \; \frac{\left[ \left(\frac{1}{2}\right)^{t-1} - \left(\frac{1}{2}\right)^{t}\delta \right]p_0}{\left[ \left(\frac{1}{2}\right)^{t-1} - \left(\frac{1}{2}\right)^{t}\delta \right]p_0 + (1-\delta)\sum_{j=1}^{k'-1}p_j},
\end{align*}
implying $\sigma_i^{k'}(t,X)= 1$. This completes the proof. $\blacksquare$

\subsection*{Proof of Proposition \ref{prop_strategic_dirty}}

\noindent\textbf{\emph{Step 1:}} 
Consider any $i\in N$. If $x_{-i}=O$, player $i$ knows his face is dirty immediately,
suggesting 1 is a dominant strategy and $\tilde{\sigma}_i^k(O)=1$ for any $k\geq 1$.
If $x_{-i}=X$, the expected payoff of 1 is $p\alpha - (1-p)<0$, 
implying $\tilde{\sigma}_i^k(X)\geq 2$ for any $k\geq 1$.
Moreover, level 1 players believe the other player is level 0, 
so when observing $X$, the expected payoff of $2\leq j\leq T$ is 
\begin{align*}
     p\left[\frac{T+2-j}{T+1}\delta^{j-1}\alpha \right] - (1-p)\left[\frac{T+2-j}{T+1}\delta^{j-1} \right]
    = \delta^{j-1}\left(\frac{T+2-j}{T+1} \right)\left[p\alpha -(1-p)\right] <0.
\end{align*}
implying $\tilde{\sigma}_i^1(X)=T+1$.

\bigskip

\noindent\textbf{\emph{Step 2:}} This step proves for any $K>1$,
if $\tilde{\sigma}_i^{l+1}(X) \leq \tilde{\sigma}_i^{l}(X)$ for all 
$1\leq l \leq K-1$, then $\tilde{\sigma}_i^{K+1}(X)\leq \tilde{\sigma}_i^{K}(X)$.
Note that if $\tilde{\sigma}_i^{K}(X)=T+1$, then there is nothing to prove.
Let $s^* \equiv \tilde{\sigma}_{i}^{K}(X)$ and focus on the case where 
$2\leq s^*\leq T$. If $s^*=T$, then level $K+1$ player's expected payoff of choosing $T$ is 
\begin{align*}
&\delta^{T-1} \left[p\alpha \left(\frac{2}{T+1}\frac{p_0}{\sum_{j=0}^K p_j} + \frac{\sum_{j=1}^K p_j}{\sum_{j=0}^K p_j} \right) -(1-p) \left(\frac{2}{T+1}\frac{p_0}{\sum_{j=0}^K p_j} \right)  \right]\\
>\; &\delta^{T-1} \left[p\alpha \left(\frac{2}{T+1}\frac{p_0}{\sum_{j=0}^{K-1} p_j} + \frac{\sum_{j=1}^{K-1} p_j}{\sum_{j=0}^{K-1} p_j} \right) -(1-p) \left(\frac{2}{T+1}\frac{p_0}{\sum_{j=0}^{K-1} p_j} \right)  \right] \geq 0. 
\end{align*}
The last inequality holds because it is optimal for level $K$ players
to choose $T$. This suggests that $T+1$ is dominated by $T$ and hence
$\tilde{\sigma}_{i}^{K+1}(X) \leq T = \tilde{\sigma}_{i}^{K}(X)$.

On the other hand, consider $2\leq s^* \leq T-1$. 
If level $K+1$ player $i$ chooses some $s$ where $s^* < s< T+1$ that
yields a non-negative expected payoff, then choosing $s$ is strictly suboptimal because

\begin{align*}
&\delta^{s-1} \left[p\alpha \left(\frac{T+2-s}{T+1}\frac{p_0}{\sum_{j=0}^{K} p_j} + \frac{\sum_{j=1}^{K-1} p_j}{\sum_{j=0}^K p_j} \right) -(1-p) \left(\frac{T+2-s}{T+1}\frac{p_0}{\sum_{j=0}^K p_j} \right)  \right]\\
 &< \delta^{s-1} \left[p\alpha \left(\frac{T+2-s}{T+1}\frac{p_0}{\sum_{j=0}^{K-1} p_j} + \frac{\sum_{j=1}^{K-1} p_j}{\sum_{j=0}^{K-1} p_j} \right) -(1-p) \left(\frac{T+2-s}{T+1}\frac{p_0}{\sum_{j=0}^{K-1} p_j} \right)  \right] \\
&\quad\leq  \delta^{s^*-1} \left[p\alpha \left(\frac{T+2-s^*}{T+1}\frac{p_0}{\sum_{j=0}^{K-1} p_j} + \frac{\sum_{j=1}^{K-1} p_j}{\sum_{j=0}^{K-1} p_j} \right) -(1-p) \left(\frac{T+2-s^*}{T+1}\frac{p_0}{\sum_{j=0}^{K-1} p_j} \right)  \right] \\
&\qquad<  \delta^{s^*-1} \left[p\alpha \left(\frac{T+2-s^*}{T+1}\frac{p_0}{\sum_{j=0}^{K} p_j} + \frac{\sum_{j=1}^{K} p_j}{\sum_{j=0}^{K} p_j} \right) -(1-p) \left(\frac{T+2-s^*}{T+1}\frac{p_0}{\sum_{j=0}^{K} p_j} \right)  \right].
\end{align*}
Note that the second inequality holds because 
$s^*$ is level $K$ player's optimal choice, and the RHS of the 
last inequality is level $K+1$ player's expected
payoff of choosing $s^*$. 
These inequalities show that 
it is not optimal for level $K+1$ players to choose any $s>s^*$, 
suggesting that 
$\tilde{\sigma}_{i}^{K+1}(X) \leq \tilde{\sigma}_{i}^{K}(X)$.

\bigskip

\noindent\textbf{\emph{Step 3:}}  
The proposition is proven by induction on $k$. This step establishes the base case
for level 2 players. For any $2\leq j \leq T$, the expected payoff of choosing $j$ is
$\mathbb{E}u_i^2(j|X) = $
\begin{align*}
p\left[\left(\frac{T+2-j}{T+1}\delta^{j-1}\alpha\right)\frac{p_0}{p_0 + p_1} 
+ \left(\delta^{j-1}\alpha\right)\frac{p_1}{p_0 + p_1}  \right] - (1-p)\left[\left(\frac{T+2-j}{T+1}\delta^{j-1}\right)\frac{p_0}{p_0 + p_1} \right].
\end{align*}
For level 2 players and any $2\leq j \leq T-1$, let $\Delta_j^2\equiv 
\mathbb{E}u_i^2(j|X) - \mathbb{E}u_i^2(j+1|X) $ 
be the difference of expected payoffs between $j$ and $j+1$. That is,
\begin{align*}
\Delta^2_j & = \delta^{j-1}p\alpha\left[\left(\frac{T+2-j}{T+1} - \frac{T+1-j}{T+1}\delta\right)\frac{p_0}{p_0 + p_1} + (1-\delta)\frac{p_1}{p_0 + p_1} 
\right] \\
& \qquad\qquad\qquad\qquad\qquad\qquad - \delta^{j-1}(1-p) \left[\left(\frac{T+2-j}{T+1} - \frac{T+1-j}{T+1}\delta\right)\frac{p_0}{p_0 + p_1}\right],
\end{align*}
suggesting that $j$ dominates $j+1$ if and only if 
$$\Delta^2_j \geq 0 \iff \bar{\alpha} \geq \frac{\left[ \frac{T+2-j}{T+1} - \frac{T+1-j}{T+1}\delta \right]p_0}{\left[ \frac{T+2-j}{T+1} - \frac{T+1-j}{T+1}\delta \right]p_0 + (1-\delta)p_1}.$$
Because the RHS is decreasing function in $j$, 
$\Delta^2_j \geq 0$ implies $\Delta^2_{j+1} \geq 0$. Moreover, since
$$\mathbb{E}u_i^2(j|X)\geq 0 \iff 
\bar{\alpha} \; \geq \; \frac{\frac{T+2-j}{T+1} p_0}{ \frac{T+2-j}{T+1}p_0 + p_1},$$
$\Delta^2_j \geq 0$ implies $\mathbb{E}u_i^2(j|X)\geq 0$ because
\begin{align*}
 \bar{\alpha} \; \geq \; \frac{\left[ \frac{T+2-j}{T+1} - \frac{T+1-j}{T+1}\delta \right]p_0}{\left[ \frac{T+2-j}{T+1} - \frac{T+1-j}{T+1}\delta \right]p_0 + (1-\delta)p_1} \; > \; 
 \frac{\frac{T+2-j}{T+1} p_0}{ \frac{T+2-j}{T+1}p_0 + p_1}.
\end{align*}
As a result, 
$\tilde{\sigma}_{i}^{2}(X)\leq T$ if and only if $\mathbb{E}u_i^2(T|X)\geq 0$, 
which is equivalent to 
$$\bar{\alpha} \; \geq \; \frac{\frac{2}{T+1} p_0}{ \frac{2}{T+1}p_0 + p_1},$$
and for any other $2\leq t \leq T-1$, $ \tilde{\sigma}_{i}^{2}(X) \leq t$ if and only if
$$\Delta^2_t \geq 0 \iff \bar{\alpha} \; \geq\; \frac{\left[ \frac{T+2-t}{T+1} - \frac{T+1-t}{T+1}\delta \right]p_0}{\left[ \frac{T+2-t}{T+1} - \frac{T+1-t}{T+1}\delta \right]p_0 + (1-\delta)p_1}.$$

\bigskip

\noindent\textbf{\emph{Step 4:}}  
Step 3 establishes the base case where $k=2$. Now suppose 
there is $K>2$ such that the 
statement holds for any $2\leq k \leq K$. It suffices to prove that 
the statement also holds for level $K+1$ players.
By step 1, $\tilde{\sigma}_i^{K+1}(X)\geq 2$. Besides, note that 
for any $1\leq t\leq T$ and $1\leq l \leq K$, if $\tilde{\sigma}_{-i}^{l}(X)> t $, 
then level $K+1$ player $i$'s expected 
payoff of choosing $2\leq j \leq t+1$ is $\mathbb{E}u_i^{K+1}(j|X)=$
\begin{align*}
\delta^{j-1} \left[p\alpha \left(\frac{T+2-j}{T+1}\frac{p_0}{\sum_{j=0}^K p_j} + \frac{\sum_{j=1}^K p_j}{\sum_{j=0}^K p_j} \right) -(1-p) \left(\frac{T+2-j}{T+1}\frac{p_0}{\sum_{j=0}^K p_j} \right)  \right].
\end{align*}
Similar to step 3, we define
$\Delta^{K+1}_{t'}$ for any $2\leq t' \leq t$ where $\Delta^{K+1}_{t'}$ is the 
difference of expected payoff between choosing $t'$ and $t'+1$. That is, 
\begin{align*}
\Delta^{K+1}_{t'} &\equiv \delta^{t'-1}p\alpha\left[\left(\frac{T+2-t'}{T+1} - \frac{T+1-t'}{T+1}\delta\right)\frac{p_0}{\sum_{j=0}^K p_j} + (1-\delta)\frac{\sum_{j=1}^K p_j}{\sum_{j=0}^K p_j} 
\right] \\
& \qquad\qquad\qquad\qquad\qquad\qquad - \delta^{t'-1}(1-p) \left[\left(\frac{T+2-t'}{T+1} - \frac{T+1-t'}{T+1}\delta\right)\frac{p_0}{\sum_{j=0}^K p_j}\right].
\end{align*}
By the same argument as in step 3, 
$\Delta^{K+1}_{t'}<0$ implies $\Delta^{K+1}_{t'-1}<0$.
Therefore, if $\tilde{\sigma}_{-i}^{l}(X)> t $
for any $1\leq l \leq K$, it is strictly dominated 
for level $K+1$ players to choose $t'$ (and all strategies $s<t'$) 
where $2\leq t' \leq t$ if 
\begin{equation}\tag{A.5}\label{eq:strategic_dirty}
\bar{\alpha} \; < \; \frac{\left[ \frac{T+2-t'}{T+1} - \frac{T+1-t'}{T+1}\delta \right]p_0}{\left[ \frac{T+2-t'}{T+1} - \frac{T+1-t'}{T+1}\delta \right]p_0 + (1-\delta)\sum_{j=1}^{K}p_j},
\end{equation}
and by a similar argument as in step 3, $\Delta^{K+1}_{t'}\geq 0$ implies
$\mathbb{E}u_i^{K+1}(t'|X)\geq 0$.

The proof for period $T$ is straightforward. The implication of  
the induction hypothesis is that $\tilde{\sigma}_{i}^{l+1}(X) \leq 
\tilde{\sigma}_{i}^{l}(X)$ for all $1\leq l\leq K-1$. 
By step 2, $\tilde{\sigma}_{i}^{K+1}(X) \leq T$ if $\tilde{\sigma}_{i}^{K}(X)\leq T$.
Thus, it suffices to consider the case where 
$$\bar{\alpha} \; <\; \frac{\frac{2}{T+1} p_0}{ \frac{2}{T+1}p_0 + \sum_{j=1}^{K-1}p_j}.$$
By induction hypothesis, $\tilde{\sigma}_{i}^{l}(X)= T+1$ for all $1\leq l \leq K$,
so $\tilde{\sigma}_{i}^{K+1}(X)\leq T$ if and only if 
$$\mathbb{E}u_i^{K+1}(T|X)\geq 0 \iff 
\bar{\alpha} \; \geq\; \frac{\frac{2}{T+1} p_0}{ \frac{2}{T+1}p_0 + \sum_{j=1}^{K}p_j} .$$
Next, consider any $2\leq t\leq T-1$. By induction hypothesis and 
step 2, if $\tilde{\sigma}_{i}^{K}(X)\leq t$, then $\tilde{\sigma}_{i}^{K+1}(X)\leq t$.
Hence, it suffices to complete the proof by considering 
\begin{align*}
    \bar{\alpha} \; < \; \frac{\left[ \frac{T+2-t}{T+1} - \frac{T+1-t}{T+1}\delta \right]p_0}{\left[ \frac{T+2-t}{T+1} - \frac{T+1-t}{T+1}\delta \right]p_0 + (1-\delta)\sum_{j=1}^{K-1}p_j}.
\end{align*}
In this case, $t<\tilde{\sigma}_i^{l+1}(X) \leq \tilde{\sigma}_i^{l}(X)$ 
for all $1\leq l\leq K-1$.
Therefore, inequality (\ref{eq:strategic_dirty}) implies that 
$\tilde{\sigma}_i^{K+1}(X)\leq t$ if and only if 
$$ \bar{\alpha} \; \geq \; \frac{\left[ \frac{T+2-t}{T+1} - \frac{T+1-t}{T+1}\delta \right]p_0}{\left[ \frac{T+2-t}{T+1} - \frac{T+1-t}{T+1}\delta \right]p_0 + (1-\delta)\sum_{j=1}^{K}p_j}.$$
This completes the proof of this proposition. 
$\blacksquare$

\subsection*{Proof of Corollary \ref{coro_normal_dirty}}

It suffices to prove the monotonicity by showing for all $k'>k\geq 2$, if $\tilde{\sigma}_i^k(X)\leq t$, then 
$\tilde{\sigma}_i^{k'}(X)\leq t$ for any $2\leq t\leq T$. 
We can separate the analysis into two cases. First, if $t=T$, then by Proposition \ref{prop_strategic_dirty},
$\tilde{\sigma}_i^k(X)\leq T$ suggests
\begin{align*}
    \bar{\alpha} \; \geq \; \frac{\frac{2}{T+1} p_0}{ \frac{2}{T+1}p_0 + \sum_{j=1}^{k-1}p_j} \; > \; \frac{\frac{2}{T+1} p_0}{ \frac{2}{T+1}p_0 + \sum_{j=1}^{k'-1}p_j},
\end{align*}
implying $\tilde{\sigma}_i^{k'}(X)\leq T$. Second, for any $2\leq t \leq T-1$, by 
Proposition \ref{prop_strategic_dirty}, $\tilde{\sigma}_i^k(X)\leq t$ suggests
\begin{align*}
    \bar{\alpha}  \; \geq \; \frac{\left[ \frac{T+2-t}{T+1} - \frac{T+1-t}{T+1}\delta \right]p_0}{\left[ \frac{T+2-t}{T+1} - \frac{T+1-t}{T+1}\delta \right]p_0 + (1-\delta)\sum_{j=1}^{k-1}p_j}
     \; > \; \frac{\left[ \frac{T+2-t}{T+1} - \frac{T+1-t}{T+1}\delta \right]p_0}{\left[ \frac{T+2-t}{T+1} - \frac{T+1-t}{T+1}\delta \right]p_0 + (1-\delta)\sum_{j=1}^{k'-1}p_j},
\end{align*}
implying $\tilde{\sigma}_i^{k'}(X)\leq t$. This completes the proof. $\blacksquare$

\subsection*{Proof of Proposition \ref{prop:dirty_representation}}

First, for any $k\geq 2$, it suffices to prove $\mathcal{S}_{T}^k \subset \mathcal{E}_{T}^k$ by showing 
if $\tilde{\sigma}_i^k(X)\leq T$, then $\hat{\sigma}_i^k(X)\leq T$.
This is true because 
$$\frac{\frac{2}{T+1} p_0}{ \frac{2}{T+1}p_0 + \sum_{j=1}^{k-1}p_j} \; > \;
\frac{\left(\frac{1}{2}\right)^{T-1}p_0}{\left(\frac{1}{2}\right)^{T-1}p_0 + \sum_{j=1}^{k-1}p_j} .$$
Similarly, for $2\leq t\leq T-1$, we can first observe that $\mathcal{S}_{t}^k \subset \mathcal{E}_{t}^k$
if and only if
\begin{align}\label{eq:proposition_rep}
    \frac{\left[ \frac{T+2-t}{T+1} - \frac{T+1-t}{T+1}\delta \right]p_0}{\left[ \frac{T+2-t}{T+1} - \frac{T+1-t}{T+1}\delta \right]p_0 + (1-\delta)\sum_{j=1}^{k-1}p_j} \;  &\geq \;
    \frac{\left[ \left(\frac{1}{2}\right)^{t-1} - \left(\frac{1}{2}\right)^{t}\delta \right]p_0}{\left[ \left(\frac{1}{2}\right)^{t-1} - \left(\frac{1}{2}\right)^{t}\delta \right]p_0 + (1-\delta)\sum_{j=1}^{k-1}p_j}  \nonumber   \\ 
     \iff \delta &\leq \; \frac{(2^t-2)(T+1)-(t-1)2^t}{(2^t-1)(T+1)-t2^t}
    \equiv \overline{\delta}(T,t). \tag{A.6}
\end{align}
where $\overline{\delta}(T,t)>0$ because $(2^t-2)(T+1)-(t-1)2^t \geq 2(T+1)-4>0$ and 
$(2^t-1)(T+1)-t2^t\geq 3(T+1)-8>0$.
If $\overline{\delta}(T,t)>1$, then 
the inequality holds for any $\delta\in(0,1)$, and 
hence $\mathcal{S}_{t}^k \subset \mathcal{E}_{t}^k$. 
Otherwise, if $\overline{\delta}(T,t)<1$, 
the inequality does not hold for all $\delta$, implying 
there is no set inclusion relationship between 
$\mathcal{S}_{t}^k$ and $\mathcal{E}_{t}^k$.
In addition, inequality (\ref{eq:proposition_rep}) suggests 
$\hat{\sigma}_i^k(X) < \tilde{\sigma}_i^k(X)$
if $\delta < \overline{\delta}(T,t)$ and 
$\hat{\sigma}_i^k(X) > \tilde{\sigma}_i^k(X)$ if $\delta > \overline{\delta}(T,t)$. 
Lastly, as we rearrange the inequality, we can obtain that 
$$\overline{\delta}(T,t) < 1 \iff
\frac{(2^t-2)(T+1)-(t-1)2^t}{(2^t-1)(T+1)-t2^t}<1
\iff t<\frac{\ln(T+1)}{\ln(2)}.$$
This completes the proof of this proposition. 
$\blacksquare$

\subsection*{Proof of Corollary \ref{coro:limit_representation}}

By Proposition \ref{prop:dirty_representation}, we know for any 
$k\geq 2$, there is no set inclusion relationship between 
$\mathcal{S}_t^k$ and $\mathcal{E}_t^k$ if $2\leq t < 
[\ln(T+1)/\ln(2)]$. 
When $T\rightarrow \infty$, this condition holds for any $t\geq 2$. Moreover, from
Proposition \ref{prop:dirty_representation}, we
can obtain that 
\begin{align*}
    \overline{\delta}^*(t) = \lim_{T\rightarrow\infty} \overline{\delta}(T,t) = 
    \lim_{T\rightarrow\infty} \frac{(2^t-2)(T+1)-(t-1)2^t}{(2^t-1)(T+1)-t2^t} = \frac{2^t-2}{2^t-1}.
\end{align*}
This completes the proof. $\blacksquare$

\subsection*{Additional Result for Poisson-DCH}

One feature of the Poisson-DCH model is that as $\tau\rightarrow \infty$, the aggregate choice frequencies 
converge to the equilibrium prediction. This provides a second interpretation for the parameter $\tau$: 
the higher the value of $\tau$, the closer the predictions are to the equilibrium. 
It is worth noting that, as highlighted by \cite{camerer2004cognitive}, this convergence
property does not hold for general classes of games.

For the sake of simplicity, I will prove the result for sequential two-person games. A 
similar argument holds for the simultaneous version. 
For any two-person dirty faces game, conditional on there is an announcement,
there are two possible states: one dirty face or two dirty faces, which are denoted as 
$\Omega = \{OX, XX \}$.  For each $\omega\in \Omega$, equilibrium predicts a deterministic 
terminal period. We use $F_\omega^*(t)$ to express the 
(degenerated) distribution of terminal periods at the equilibrium. 
The equilibrium predicts that players will choose
$C$ at period 1 when seeing $O$, and choose $W$ at period 1 and $C$ at period 2 when seeing $X$.
Therefore, when $\omega=OX$, the game will end at period 1, and when $\omega=XX$,
the game will end at period 2. In other words,
\begin{align*}
F_{OX}^*(t) = 
\begin{cases}
0 \;\;\mbox{if}\;\; t<1 \\
1 \;\;\mbox{if}\;\; t\geq 1,
\end{cases}
\mbox{ and }\quad
F_{XX}^*(t) = 
\begin{cases}
0 \;\;\mbox{if}\;\; t<2 \\
1 \;\;\mbox{if}\;\; t\geq 2. 
\end{cases}
\end{align*}

In contrast, given any $\tau>0$ and $\omega\in \Omega$, 
the Poisson-DCH model predicts a non-degenerated distribution over 
all possible terminal periods. We use $F_{w}^D(t|\tau)$ to denote the 
distribution predicted by the Poisson-DCH.
Proposition \ref{prop:covergence_poisson} states that when $\tau\rightarrow\infty$,
the max norm between $F_{\omega}^D(t|\tau)$ and
$F_{\omega}^*(t)$ will converge to 0 for any $\omega\in\Omega$.

\begin{proposition}\label{prop:covergence_poisson}

Consider any sequential two-person dirty faces game. When the prior distribution of levels
follows Poisson($\tau$), for any $\omega\in \Omega$,
$$\lim_{\tau\rightarrow\infty} \left\lVert F_\omega^*(t) - F_\omega^D(t|\tau) \right\rVert_{\infty} = 0.$$

\end{proposition}

\noindent\textit{Proof.}

\noindent When $\omega=OX$, a strategic player that sees a clean face will choose $C$ in period 1. Therefore, 
\begin{align*}
    F_{OX}^D(1|\tau) = 1 - \left(\frac{1}{2}e^{-\tau}\right)\left(1-\frac{1}{2}e^{-\tau} \right).
\end{align*}
To show $\left\lVert F_{OX}^*(t) - F_{OX}^D(t|\tau) \right\rVert_{\infty} \rightarrow 0$, it suffices to show 
$F_{OX}^D(1|\tau) \rightarrow 1$ as $\tau \rightarrow \infty$. This is 
true because
\begin{align*}
\lim_{\tau\rightarrow\infty} F_{OX}^D(1|\tau) = 
\lim_{\tau\rightarrow\infty} 1 - \left(\frac{1}{2}e^{-\tau}\right)\left(1-\frac{1}{2}e^{-\tau} \right) =1.
\end{align*}

When $\omega=XX$, it suffices to prove the convergence by showing 
$F_{XX}^D(1|\tau) \rightarrow 0$ and $F_{XX}^D(2|\tau) \rightarrow 1$
as $\tau \rightarrow \infty$. 
Since every level $k\geq 1$ will choose $W$ in
period 1 when seeing a dirty face, 
$F_{XX}^D(1|\tau) = 1 - \left[1 - (1/2)e^{-\tau} \right]^2$, implying that 
\begin{align*}
\lim_{\tau\rightarrow\infty} F_{XX}^D(1|\tau) = 
\lim_{\tau\rightarrow\infty} 1 - \left[1 - \frac{1}{2}e^{-\tau} \right]^2 =0.    
\end{align*}
Lastly, let $K^*(\tau)$ be the lowest level of players 
to choose $C$ at period 2 when seeing a dirty face with the prior distribution of 
levels being Poisson($\tau$). By Proposition \ref{prop_extensive_dirty}, 
$K^*(\tau)$ is weakly decreasing in $\tau$, and $K^*(\tau) \rightarrow 2$ as 
$\tau \rightarrow \infty$. Hence, 
$$F_{XX}^D(2|\tau) = 1 - \left[(1/4)e^{-\tau} +
\sum_{j=1}^{K^*(\tau)-1}e^{-\tau}\tau^j/j! \right]^2,$$ 
suggesting the limit is 
\begin{align*}
\lim_{\tau\rightarrow\infty} F_{XX}^D(2|\tau) = 
\lim_{\tau\rightarrow\infty} 1 - \left[\frac{1}{4}e^{-\tau} +
\sum_{j=1}^{K^*(\tau)-1}\frac{e^{-\tau}\tau^j}{j!} \right]^2 = 
\lim_{\tau\rightarrow\infty} 1 - \left[\frac{1}{4}e^{-\tau} +
\tau e^{-\tau} \right]^2 =1.    \; \blacksquare
\end{align*}

\newpage
\section{Three-Person Three-Period Dirty-Faces Games (For Online Publication)}
\label{sec:dirty_exp_appendix}

In this section, I characterize the DCH solutions for the sequential and simultaneous 
three-person three-period dirty-faces games. In a three-person dirty-faces game, each player $i$
will be random assigned a face type type, either clean ($O$) or dirty ($X$). 
The face types are i.i.d. drawn from the distribution $p = \Pr(x_i = X) = 1 - \Pr(x_i = O)$ 
where $p>0$. After the face 
types are drawn, each player $i$ can observe the other two players' faces $x_{-i}$ but 
not their own face. If there is at least one player 
having a dirty face, a public announcement is made, informing all players of this fact.

There are up to 3 periods. In each period, all three players simultaneously choose either to 
claim to have a dirty face ($C$) or wait ($W$) and the actions are revealed at the 
end of each period. Similar to the two-person games, the game will end after any period where 
some player chooses $C$ or after period 3. A player's payoff depends on their own face types and 
their actions in the terminal period. If player $i$ waits in the terminal period, his payoff is 0 
regardless of his face type. If player $i$ chooses $C$ in the terminal period, say period $t$, 
then he will receive $\delta^{t-1}\alpha$ if his face is dirty but $-\delta^{t-1}$ if his 
face is clean.
Following the analysis of two-person games, I will focus on the case where there is a public
announcement. Moreover, the assumption that $0 < \bar{\alpha} \equiv 
\alpha p / (1-p) < 1$ is maintained so it is strictly dominated to choose $C$ in period
1 when seeing one or two dirty faces.

\subsection{DCH Solution for the Sequential Games}

In the sequential three-person three-period dirty-faces game, a behavioral strategy for player $i$
is a mapping from the period and the observed face types ($x_{-i} \in \{OO, OX, XX \}$) to the 
probability of claiming to have a dirty face. The behavioral strategy is denoted by 
$$\sigma_i: \{1,2,3\} \times \{OO,OX,XX \}
\rightarrow [0,1].$$
For the sake of simplicity, I assume that each player $i$'s level is i.i.d. drawn from the 
the distribution $p = \left(p_k \right)_{k=0}^\infty$ where $p_k>0$ for all $k$. 
Proposition \ref{prop:three_p_ext} characterizes the 
DCH solution for the sequential three-person three-period dirty-faces games.

\begin{proposition}\label{prop:three_p_ext}
For any sequential three-person three-period dirty-faces game, 
the level-dependent strategy profile of the DCH solution satisfies that for any $i\in N$, 
\begin{itemize}
    \item[1.] $\sigma^k_i(t,OO)=1$ for all $k\geq 1$ and 
     $1\leq t\leq 3$.
     \item[2.] $\sigma_i^1(t,OX)=0$ for any $1\leq t \leq 3$. Moreover, for any $k\geq 2$,
     \begin{itemize}
         \item[(1)]  $\sigma_i^k(1,OX)=0$,
         \item[(2)]  $\sigma_i^k(2,OX)=1$ if and only if 
         \begin{align*}
            \bar{\alpha} \; \geq \; \frac{\left(\frac{1}{2} - \frac{1}{4}\gamma_k \delta \right)p_0}{\left(\frac{1}{2} - \frac{1}{4}\gamma_k \delta \right)p_0 + (1-\gamma_k\delta)\sum_{j=1}^{k-1}p_j} 
        \end{align*}
        where $\gamma_k \equiv \left[\frac{1}{4}p_0 + \sum_{j=1}^{k-1}p_j \right] /
        \left[\frac{1}{2}p_0 + \sum_{j=1}^{k-1}p_j \right]$,
        \item[(3)] $\sigma_{i}^{k}(3,OX) = 1 $ if and only if
        \begin{align*}
        \bar{\alpha} \; \geq \; \frac{ \frac{1}{4}p_0}{\frac{1}{4}p_0 + \sum_{j=1}^{k-1}p_j},  
        \end{align*}
    \end{itemize}
    \item[3.] $\sigma_i^1(t,XX)=\sigma_i^2(t,XX)=0$ for any $1\leq t \leq 3$. Moreover, for any $k\geq 3$,
    \begin{itemize}
        \item[(1)] $\sigma_i^k(1,XX) = \sigma_i^k(2,XX)=0$,
        \item[(2)] $\sigma_i^k(3,XX)=1$ if and only if there exists $2\leq l \leq k-1$ such that $\sigma_i^l(2,OX)=1$ with 
        $L_k^* \equiv \min_j\left\{j<k: \sigma_i^j(2,OX)=1 \right\}$, and 
        \begin{align*}
            \bar{\alpha} \geq \max\left\{ \frac{\left(\frac{1}{2} - \frac{1}{4}\gamma_{L_k^*} \delta \right)p_0}{\left(\frac{1}{2} - \frac{1}{4}\gamma_{L_k^*} \delta \right)p_0 + (1-\gamma_{L_k^*}\delta)\sum_{j=1}^{L_k^*-1}p_j} \; , \; 
            \left(\frac{ \frac{1}{4}p_0 + \sum_{j=1}^{L_k^*-1}p_j}{\frac{1}{4}p_0 +  \sum_{j=1}^{k-1}p_j} \right)^2 \right\}.
        \end{align*}
    \end{itemize}
\end{itemize}
\end{proposition}

\bigskip

\noindent\emph{Proof.}

\noindent\textbf{\emph{Step 1:}} Consider any $i\in N$. If $x_{-i}=OO$, 
then player $i$ knows his face is dirty immediately. Therefore, 
$C$ is a dominant strategy, suggesting $\sigma_i^k(t, OO)=1$ for all $k\geq 1$
and $1\leq t \leq 3$. If $x_{-i}=OX$, 
player $i$'s belief of having a dirty face at period 1 is 
$p$. Hence, the expected payoff of choosing $C$ at period 1 is 
$p\alpha - (1-p)<0$, implying $\sigma_i^k(1,OX) =0$
for all $k\geq 1$. Similarly, if $x_{-i}=XX$, 
the beliefs of having a dirty face at period 1 and 2 are $p$, which suggests
$\sigma_i^k(1,XX) = \sigma_i^k(2,XX) =0$ for all $k\geq 1$.

In addition, level 1 players believe other players' actions don't convey any 
information about their own face types, so 
$\sigma_i^1(t,OX)= \sigma_i^1(t,XX)=0$ for any $1\leq t \leq 3$.
Since level 1 players behave exactly the same when observing 
$OX$ and $XX$, level 2 player $i$'s belief about having a dirty face 
at period 3 is still $p$ when $x_{-i}=XX$, implying $\sigma_i^2(3,XX)=0$.

\bigskip

\noindent\textbf{\emph{Step 2:}} In this step, I prove that 
\begin{align*}
    \frac{\left(\frac{1}{2} - \frac{1}{4}\gamma_k \delta \right)p_0}{\left(\frac{1}{2} - \frac{1}{4}\gamma_k \delta \right)p_0 + (1-\gamma_k\delta)\sum_{j=1}^{k-1}p_j}
\end{align*}
is decreasing in $k$ for all $k\geq 2$ where $\gamma_k \equiv \left[\frac{1}{4}p_0 + 
\sum_{j=1}^{k-1}p_j \right] /\left[\frac{1}{2}p_0 + \sum_{j=1}^{k-1}p_j \right]$.
To prove this, it suffices to prove that for any $l\geq 2$, 
\begin{align*}
&\frac{\left(\frac{1}{2} - \frac{1}{4}\gamma_l \delta \right)p_0}{\left(\frac{1}{2} - \frac{1}{4}\gamma_l \delta \right)p_0 + (1-\gamma_l\delta)\sum_{j=1}^{l-1}p_j}  \geq 
\frac{\left(\frac{1}{2} - \frac{1}{4}\gamma_{l+1} \delta \right)p_0}{\left(\frac{1}{2} - \frac{1}{4}\gamma_{l+1} \delta \right)p_0 + (1-\gamma_{l+1}\delta)\sum_{j=1}^{l}p_j} \\
\iff & \left(-\frac{1}{4}\gamma_{l+1}\delta + \frac{1}{4}\gamma_l\delta \right)\sum_{j=1}^{l-1}p_j 
+ \left(1-\gamma_{l+1}\delta \right)\left(\frac{1}{2} - \frac{1}{4}\gamma_l\delta \right)p_l \geq 0
\end{align*}
Notice that the LHS of the inequality is decreasing in $\delta$ since
\begin{align*}
&\frac{d}{d\delta}\left[\left(-\frac{1}{4}\gamma_{l+1}\delta + \frac{1}{4}\gamma_l\delta \right)\sum_{j=1}^{l-1}p_j 
+ \left(1-\gamma_{l+1}\delta \right)\left(\frac{1}{2} - \frac{1}{4}\gamma_l\delta \right)p_l \right] \\
= & \underbrace{\left(-\frac{1}{4}\gamma_{l+1} + \frac{1}{4}\gamma_l \right)}_{<0}\sum_{j=1}^{l-1}p_j
 + \underbrace{\left(-\frac{1}{2}\gamma_{l+1} - \frac{1}{4}\gamma_l + \frac{1}{2}\gamma_l\gamma_{l+1}\delta\right)}_{
 \substack{{< -\frac{1}{2}\gamma_{l+1} - \frac{1}{4}\gamma_l + \frac{1}{2}\gamma_l\gamma_{l+1}\;\;\;\;\;}\\{\leq -\sqrt{\frac{1}{2}\gamma_l\gamma_{l+1}} + \frac{1}{2}\gamma_l\gamma_{l+1} \; < \;0}} }p_l<0.
\end{align*}
Therefore, it suffices to prove that the inequality holds when $\delta=1$. That is, 
$$\left(-\frac{1}{4}\gamma_{l+1} + \frac{1}{4}\gamma_l \right)\sum_{j=1}^{l-1}p_j 
+ \left(1-\gamma_{l+1} \right)\left(\frac{1}{2} - \frac{1}{4}\gamma_l \right)p_l \geq 0,$$
which holds because the inequality is equivalent to
\begin{align*}
\frac{p_l}{\sum_{j=1}^{l-1}p_j} \geq \frac{\frac{1}{4}\left(\gamma_{l+1} - \gamma_l \right)}{\left(1 - \gamma_{l+1}\right)\left(\frac{1}{2} - \frac{1}{4}\gamma_{l}\right)}=\frac{p_l}{\frac{3}{4}p_0 + \sum_{j=1}^{l-1}p_j}.
\end{align*}

\bigskip

\noindent\textbf{\emph{Step 3:}} In this step, I characterize 
level $k$ player $i$'s behavior when $x_{-i}=OX$ for all $k\geq 2$ by induction on $k$. 
I first prove the base case where $k=2$. At period 3, level 2 player $i$'s belief about having 
a dirty face is 
\begin{align*}
    \mu_{i}^2(X|3,OX) = \sum_{\tau_{-i}}
    \mu_i^2(X,\tau_{-i}|3,OX) = \frac{p\left( \frac{1}{4}p_0 + p_1\right)}{\frac{1}{4}p_0 + p p_1}.
\end{align*}
Therefore, it is optimal to choose $C$ at period 3 if and only if 
\begin{align*}
    \mu_{i}^2(X|3,OX)\alpha - (1-\mu_{i}^2(X|3,OX)) \geq 0 \iff \bar{\alpha} \; \geq \;
    \frac{\frac{1}{4}p_0}{\frac{1}{4}p_0 + p_1}.
\end{align*}
At period 2, level 2 player $i$'s belief about having a dirty face is 
\begin{align*}
    \mu_{i}^2(X|2,OX) = \sum_{\tau_{-i}}
    \mu_i^2(X,\tau_{-i}|2,OX) = \frac{p\left( \frac{1}{2}p_0 + p_1\right)}{\frac{1}{2}p_0 + p p_1},
\end{align*}
and the belief about that the two other players wait at period 2 is 
\begin{align*}
    \frac{\left( \frac{1}{4}p_0 + p_1\right)\left( \frac{1}{4}p_0 + pp_1\right)}{\left( \frac{1}{2}p_0 + p_1\right)\left( \frac{1}{2}p_0 + pp_1\right)} \equiv \gamma_2
    \left(\frac{\frac{1}{4}p_0 + pp_1}{\frac{1}{2}p_0 + pp_1} \right).
\end{align*}
Conditional on reaching period 3, the payoff of waiting is 0, and 
the expected payoff of $C$ is 
$$\frac{\delta^2}{\frac{1}{4}p_0 + pp_1}\left[p\alpha\left(\frac{1}{4}p_0 + p_1 \right) 
- (1-p)\left(\frac{1}{4}p_0 \right) \right].$$
Therefore, it is optimal to choose $C$ at period 2 if and only if 
\begin{align*}
\frac{\delta}{\frac{1}{2}p_0 + pp_1}&\left[
p\alpha\left(\frac{1}{2}p_0 + p_1 \right) 
- (1-p)\left(\frac{1}{2}p_0 \right) \right]\\
\geq \max&\left\{  \gamma_2 \left(\frac{\frac{1}{4}p_0 + pp_1}{\frac{1}{2}p_0 + pp_1} \right)\frac{\delta^2}{\frac{1}{4}p_0 + pp_1}\left[
p\alpha\left(\frac{1}{4}p_0 + p_1 \right) 
- (1-p)\left(\frac{1}{4}p_0 \right) \right], \; 0 \right\}\\
\iff  \bar{\alpha} \geq \max&\left\{ 
\frac{\left(\frac{1}{2} - \frac{1}{4}\gamma_2 \delta \right)p_0}{\left(\frac{1}{2} - \frac{1}{4}\gamma_2 \delta \right)p_0 + (1-\gamma_2\delta)p_1} \; , \;
\frac{\frac{1}{2}p_0}{\frac{1}{2}p_0 + p_1} \right\}.
\end{align*}
Furthermore, because for any $\delta\in (0,1)$,
\begin{align*}
    \frac{\left(\frac{1}{2} - \frac{1}{4}\gamma_2 \delta \right)p_0}{\left(\frac{1}{2} - \frac{1}{4}\gamma_2 \delta \right)p_0 + (1-\gamma_2\delta)p_1}
    > \frac{\frac{1}{2}p_0}{\frac{1}{2}p_0 + p_1},
\end{align*}
it is optimal for level 2 players to claim at period 2 if and only if
\begin{align*}
   \bar{\alpha} \; \geq \; \frac{\left(\frac{1}{2} - \frac{1}{4}\gamma_2 \delta \right)p_0}{\left(\frac{1}{2} - \frac{1}{4}\gamma_2 \delta \right)p_0 + (1-\gamma_2\delta)p_1} . 
\end{align*}
This completes the proof for level 2 players. 

Suppose there is $K>2$ such that the statement holds for any level $2\leq k \leq K$. I now
prove the statement holds for level $K+1$ players.
By the same argument as in the proof of Proposition \ref{prop_extensive_dirty},
level $K+1$ players would choose $C$ when it is already 
optimal for level $K$ players to choose $C$.
Therefore, for period 3, it suffices to consider the case where 
\begin{align*}
    \bar{\alpha} \; < \; \frac{ \frac{1}{4}p_0}{\frac{1}{4}p_0 + \sum_{j=1}^{K-1}p_j}.
\end{align*}
By induction hypothesis, 
we know for every level $1\leq k \leq K$ players, 
they will wait for three periods when observing one dirty face. 
Therefore, level $K+1$ player $i$'s belief about having a dirty face at period 3 
when $x_{-i}=OX$ is 
\begin{align*}
    \mu_{i}^{K+1}(X|3,OX) = \sum_{\tau_{-i}}
    \mu_i^{K+1}(X,\tau_{-i}|3,OX) = \frac{p\left( \frac{1}{4}p_0 + \sum_{j=1}^Kp_j\right)}{\frac{1}{4}p_0 + p\sum_{j=1}^Kp_j}.
\end{align*}
Consequently, level $K+1$ players would choose $C$
at period 3 if and only if 
\begin{align*}
    \mu_{i}^{K+1}(X|3,OX)\alpha - (1-\mu_{i}^{K+1}(X|3,OX)) \geq 0 \iff \bar{\alpha}\; \geq\; \frac{ \frac{1}{4}p_0}{\frac{1}{4}p_0 + \sum_{j=1}^{K}p_j}.
\end{align*}
For period 2, by step 2 and the induction hypothesis, 
it suffices to consider
\begin{align*}
    \bar{\alpha} \; < \; \frac{\left(\frac{1}{2} - \frac{1}{4}\gamma_{K+1} \delta \right)p_0}{\left(\frac{1}{2} - \frac{1}{4}\gamma_{K+1} \delta \right)p_0 + (1-\gamma_{K+1}\delta)\sum_{j=1}^{K}p_j} ;
\end{align*}
otherwise, level $K$ players would choose $C$ at period 2 and so do level $K+1$ players.
By similar argument, level $K+1$ player $i$ will choose $C$ at period 2 
if and only if 
\begin{align*}
\frac{\delta}{\frac{1}{2}p_0 + p\sum_{j=1}^{K}p_j}&\left[
p\alpha\left(\frac{1}{2}p_0 + \sum_{j=1}^{K}p_j \right) 
- (1-p)\left(\frac{1}{2}p_0 \right) \right]  \\
\geq \max&\left\{  \gamma_{K+1} \left(\frac{\delta^2}{\frac{1}{2}p_0 + p\sum_{j=1}^{K}p_j} \right)
\left[p\alpha\left(\frac{1}{4}p_0 + \sum_{j=1}^{K}p_j\right) 
- (1-p)\left(\frac{1}{4}p_0 \right) \right], \; 0\right\},
\end{align*}
which is equivalent to 
\begin{align*}
    \bar{\alpha} \; \geq \; \frac{\left(\frac{1}{2} - \frac{1}{4}\gamma_{K+1} \delta \right)p_0}{\left(\frac{1}{2} - \frac{1}{4}\gamma_{K+1} \delta \right)p_0 + (1-\gamma_{K+1}\delta)\sum_{j=1}^{K}p_j} .
\end{align*}

\bigskip

\noindent\textbf{\emph{Step 4:}} This step characterizes level $k$ player $i$'s behavior when 
$x_{-i}=XX$ for all $k\geq 3$. Consider any level $k\geq 3$.
For level $k$ players, they update their beliefs about having a dirty face at period 3 
only if there is some lower level of players that
chooses $C$ at period 2 when observing one dirty face.
That is, $\sigma_i^k(3,XX)=1$ only if there is $2\leq l \leq k-1$ such that  
\begin{align*}
    \bar{\alpha} \; \geq \; \frac{\left(\frac{1}{2} - \frac{1}{4}\gamma_l \delta \right)p_0}{\left(\frac{1}{2} - \frac{1}{4}\gamma_l \delta \right)p_0 + (1-\gamma_l\delta)\sum_{j=1}^{l-1}p_j}.
\end{align*}
If there exists such level of players, let $L_k^*$ denote the lowest level below $k$ 
that would choose $C$ at period 2 when 
observing one dirty face. In this case, 
level $k$ player $i$'s belief about having a dirty face at period 3 is 
\begin{align*}
    \mu_{i}^k(X|3,XX) = \frac{p\left(\frac{1}{4}p_0 + \sum_{j=1}^{k-1}p_j \right)^2}{p\left(\frac{1}{4}p_0 + \sum_{j=1}^{k-1}p_j \right)^2 + (1-p)\left(\frac{1}{4}p_0 + \sum_{j=1}^{L_k^*-1}p_j  \right)^2},
\end{align*}
and expected payoff of $C$ is greater than 0 if and only if
\begin{align*}
    \mu_{i}^{k}(X|3,XX)\alpha - (1-\mu_{i}^{k}(X|3,XX)) \geq 0 \iff \bar{\alpha} \; \geq \; \left(\frac{\frac{1}{4}p_0 + \sum_{j=1}^{L_k^*-1}p_j}{\frac{1}{4}p_0 + \sum_{j=1}^{k-1}p_j}\right)^2.
\end{align*}
Therefore, we can conclude that $\sigma_i^k(3,XX)=1$ if and only if 
\begin{align*}
    \bar{\alpha} \geq \max\left\{ \frac{\left(\frac{1}{2} - \frac{1}{4}\gamma_{L_k^*} \delta \right)p_0}{\left(\frac{1}{2} - \frac{1}{4}\gamma_{L_k^*} \delta \right)p_0 + (1-\gamma_{L_k^*}\delta)\sum_{j=1}^{L_k^*-1}p_j} \; , \; 
     \left(\frac{ \frac{1}{4}p_0 + \sum_{j=1}^{L_k^*-1}p_j}{\frac{1}{4}p_0 +  \sum_{j=1}^{k-1}p_j} \right)^2 \right\}.
\end{align*}
This completes the proof of step 4 and this proposition. $\blacksquare$


\subsection{DCH Solution for Simultaneous Games}

The strategically equivalent simultaneous three-person three-period dirty-faces game is a one-period
game, in which all three players simultaneously choose an action from the set 
$S = \{1,2,3,4\}$. Action $t\leq 3$ represents the plan to wait from period 1 to $t-1$ and 
claim in period $t$. Action 4 is the plan to always wait. In the simultaneous 
three-person three-period dirty-faces game, a mixed strategy is a mapping from the observed face type 
($x_{-i} \in \{OO, OX, XX \}$) to a probability distribution over the action set. That is, 
$$\tilde{\sigma}_i: \{OO,OX,XX\}\rightarrow \Delta(S).$$

Suppose $(s_i, s_{-i})$ is the realized action profile. If $s_i$ is the smallest number, then the payoff
for player $i$ is computed as the case where player $i$ claims to have a dirty face at period $s_i$; 
otherwise, player $i$'s payoff is 0. The equilibrium analysis for the simultaneous game is the 
same as the sequential game. However, as characterized by Proposition \ref{prop:three_p_strategic},
the DCH solution for the simultaneous games differs from the DCH solution for the 
sequential games.

\begin{proposition}\label{prop:three_p_strategic}

For any simultaneous three-person three-period dirty-faces game, the level-dependent strategy 
profile of the DCH solution satisfies that for $i\in N$, 
\begin{itemize}
    \item[1.] $\tilde{\sigma}_i^k(OO) = 1 $ for all $k\geq 1$.
    \item[2.] $\tilde{\sigma}_i^1(OX) = 4$. Moreover, for any $k\geq 2$,
    $\tilde{\sigma}_i^k(OX)>1$ and 
    \begin{itemize}
        \item[(1)] $\tilde{\sigma}_i^k(OX)=2$ if and only if 
        $$\bar{\alpha} \; \geq \; \frac{\frac{3}{4}p_0
        \left(\frac{3}{4}p_0 + \sum_{j=1}^{k-1}p_j \right) - \delta
        \left(\frac{1}{2}p_0\right)\left(\frac{1}{2}p_0 + \sum_{j=1}^{k-1}p_j \right)}{\left(\frac{3}{4}p_0 +
        \sum_{j=1}^{k-1}p_j \right)^2 - \delta\left(\frac{1}{2}p_0 +
        \sum_{j=1}^{k-1}p_j \right)^2 },$$
        \item[(2)] $\tilde{\sigma}_i^k(OX) \leq 3$ if and only if
        \begin{align*}
        \bar{\alpha} \; \geq \; \frac{ \frac{1}{2}p_0}{\frac{1}{2}p_0 + \sum_{j=1}^{k-1}p_j} ,  
        \end{align*}
    \end{itemize}
    \item[(3)] $\tilde{\sigma}_i^1(XX) = \tilde{\sigma}_i^2(XX) = 4$. Furthermore, 
    for any $k\geq 3$, $\tilde{\sigma}_i^k(XX)>2$,  and 
    $\tilde{\sigma}_i^k(XX)=3$ if and only if there exists $2\leq l \leq k-1$ such that $\tilde{\sigma}_i^l(OX)=2$ with
    $\tilde{L}_k^* = \min_j\left\{j < k: \tilde{\sigma}_i^j(OX)=2 \right\}$, and 
        \begin{align*}
            \bar{\alpha} \geq \max\left\{ \frac{\frac{3}{4}p_0
        \left(\frac{3}{4}p_0 + \sum_{j=1}^{\tilde{L}_k^*-1}p_j \right) - \delta
        \left(\frac{1}{2}p_0\right)\left(\frac{1}{2}p_0 + \sum_{j=1}^{\tilde{L}_k^*-1}p_j \right)}{\left(\frac{3}{4}p_0 +
        \sum_{j=1}^{\tilde{L}_k^*-1}p_j \right)^2 - \delta\left(\frac{1}{2}p_0 +
        \sum_{j=1}^{\tilde{L}_k^*-1}p_j \right)^2 } \;, \;
        \frac{ \frac{1}{2}p_0 + \sum_{j=1}^{\tilde{L}_k^*-1}p_j}{\frac{1}{2}p_0 +  \sum_{j=1}^{k-1}p_j}  \right\}.
        \end{align*}
\end{itemize}

\end{proposition}

\bigskip

\noindent\emph{Proof.}

\noindent\textbf{\emph{Step 1:}} 
Consider any $i\in N$. If $x_{-i}=OO$, player $i$ knows his face is dirty immediately,
suggesting 1 is a dominant strategy and $\tilde{\sigma}_i^k(OO)=1$ for any $k\geq 1$.
If $x_{-i}=OX$ or $XX$, the expected payoff of 1 is $p\alpha - (1-p)<0$, 
implying $\tilde{\sigma}_i^k(OX)\geq 2$ and 
$\tilde{\sigma}_i^k(XX)\geq 2$ for any $k\geq 1$.
Moreover, level 1 players believe all other players are level 0, 
so when observing $OX$ or $XX$, the expected payoff of $t\in \{2,3 \}$ is 
\begin{align*}
    p \left[\delta^{t-1}\alpha \left(\frac{5-t}{4}\right)^2\right]
     + (1-p)\left[-\delta^2 \left(\frac{5-t}{4}\right)^2 \right]
    = \delta^{t-1} \left(\frac{5-t}{4}\right)^2 \left[p\alpha - (1-p) \right]<0,
\end{align*}
implying $\tilde{\sigma}_i^1(OX) = \tilde{\sigma}_i^1(XX)=4$.

In addition, $\tilde{\sigma}_i^k(XX)\geq 3$ for 
all $k\geq 1$ can be proven by induction on $k$. From the previous calculation, 
we know $\tilde{\sigma}_i^1(XX)=4$, which establishes the base case.
Suppose $\tilde{\sigma}_i^k(XX)\geq 3$ for all $1\leq i \leq K$ for some $K>1$.
It suffices to prove $\tilde{\sigma}_i^{K+1}(XX)\geq 3$ by showing 2 is a strictly dominated strategy 
for level $K+1$ players. This is strictly dominated because
\begin{align*}
    &p \left[\delta\alpha \left(\frac{3}{4}\frac{p_0}{\sum_{j=0}^{K}p_j} 
    + \frac{\sum_{j=1}^K p_j}{\sum_{j=0}^{K}p_j}  \right)^2\right]
    + (1-p)\left[-\delta \left(\frac{3}{4}\frac{p_0}{\sum_{j=0}^{K}p_j} 
    + \frac{\sum_{j=1}^K p_j}{\sum_{j=0}^{K}p_j}\right)^2 \right]\\
    & \; \qquad\qquad\qquad\qquad\qquad\qquad\qquad\;
    = \delta\left(\frac{3}{4}\frac{p_0}{\sum_{j=0}^{K}p_j} 
    + \frac{\sum_{j=1}^K p_j}{\sum_{j=0}^{K}p_j}\right)^2\left[p\alpha - (1-p) \right]<0.
\end{align*}

\bigskip

\noindent\textbf{\emph{Step 2:}} This step establishes a monotonic result: for any $K>1$,
if $\tilde{\sigma}_i^{l+1}(OX) \leq \tilde{\sigma}_i^{l}(OX)$ for all 
$1\leq l \leq K-1$, then $\tilde{\sigma}_i^{K+1}(OX)\leq \tilde{\sigma}_i^{K}(OX)$.
If $\tilde{\sigma}_i^{K}(OX)=4$, then there is nothing to prove.
Suppose $\tilde{\sigma}_i^{l+1}(OX) \leq \tilde{\sigma}_i^{l}(OX)$ for all 
$1\leq l \leq K-1$.
If $\tilde{\sigma}_i^{K}(OX)=3$, then it is necessary 
that level $K$ player's 
expected payoff of choosing 3 is non-negative. Namely, 
\begin{align*}
    \delta^2 \left(\frac{1}{2}\frac{p_0}{\sum_{j=0}^{K-1}p_j} + 
    \frac{\sum_{j=1}^{K-1} p_j}{\sum_{j=0}^{K-1}p_j}\right)\left[ p\alpha \left(\frac{1}{2}\frac{p_0}{\sum_{j=0}^{K-1}p_j} + \frac{\sum_{j=1}^{K-1} p_j}{\sum_{j=0}^{K-1}p_j}\right) - (1-p)\left(\frac{1}{2}\frac{p_0}{\sum_{j=0}^{K-1}p_j} \right)\right] \geq 0, 
\end{align*}
which implies:
\begin{align*}
    \delta^2 \left(\frac{1}{2}\frac{p_0}{\sum_{j=0}^{K}p_j} + 
    \frac{\sum_{j=1}^K p_j}{\sum_{j=0}^{K}p_j}\right)\left[ p\alpha \left(\frac{1}{2}\frac{p_0}{\sum_{j=0}^{K}p_j} + \frac{\sum_{j=1}^K p_j}{\sum_{j=0}^{K}p_j}\right) - (1-p)\left(\frac{1}{2}\frac{p_0}{\sum_{j=0}^{K}p_j} \right)
    \right]>0.   
\end{align*}
suggesting $\tilde{\sigma}_i^{K+1}(OX)\leq 3$. If $\tilde{\sigma}_i^{K}(OX)=2$, 
it suffices to prove $\tilde{\sigma}_i^{K+1}(OX)=2$ as well. Notice that 
if $\tilde{\sigma}_i^{K}(OX)=2$, then it is necessary for level $K$ players that
2 dominates 3 and 4. Let $M$ be the lowest level of players that would choose 
$2$ when observing $OX$. Then level $K$ player's expected payoff of choosing $2$ 
would satisfy that 
\begin{align*}
    \delta\left(\frac{3}{4}p_0 + \sum_{j=1}^{K-1}p_j \right)
    &\left[p\alpha\left(\frac{3}{4}p_0 + \sum_{j=1}^{K-1}p_j \right) 
    - (1-p)\left(\frac{3}{4}p_0\right)\right]\\
    \geq  \max&\left\{ \delta^2\left(\frac{1}{2}p_0 + \sum_{j=1}^{M-1}p_j \right)
    \left[p\alpha\left(\frac{1}{2}p_0 + \sum_{j=1}^{M-1}p_j \right) 
    - (1-p)\left(\frac{1}{2}p_0\right)\right], \; 0\right\}, 
\end{align*}
which implies:
\begin{align*}
    \delta\left(\frac{3}{4}p_0 + \sum_{j=1}^{K}p_j \right)
    &\left[p\alpha\left(\frac{3}{4}p_0 + \sum_{j=1}^{K}p_j \right) 
    - (1-p)\left(\frac{3}{4}p_0\right)\right]\\
    \geq  \max&\left\{ \delta^2\left(\frac{1}{2}p_0 + \sum_{j=1}^{M-1}p_j \right)
    \left[p\alpha\left(\frac{1}{2}p_0 + \sum_{j=1}^{M-1}p_j \right) 
    - (1-p)\left(\frac{1}{2}p_0\right)\right], \; 0\right\}, 
\end{align*}
suggesting that $\tilde{\sigma}_i^{K+1}(OX)=2$.

\bigskip

\noindent\textbf{\emph{Step 3:}} In this step, I characterize level $k$ player $i$'s behavior
as $x_{-i}=OX$ for all $k\geq 2$ by induction on $k$.
Level 2 player $i$'s expected payoff of choosing $t\in \{2,3 \}$ is
\begin{align*}
    \delta^{t-1}\left(\frac{5-t}{4}\frac{p_0}{p_0 + p_1} + \frac{p_1}{p_0 + p_1} \right)
    \underbrace{\left[p\alpha\left(\frac{5-t}{4}\frac{p_0}{p_0 + p_1} + 
    \frac{p_1}{p_0 + p_1} \right)
    - (1-p)\left(\frac{5-t}{4}\frac{p_0}{p_0 + p_1}\right) \right]}_{\mbox{increasing in } t}.
\end{align*}
Therefore, $\tilde{\sigma}_i^2(OX)\leq3$ if and only if
\begin{align*}
    p\alpha\left(\frac{1}{2}p_0 + p_1 \right) - (1-p)\left(\frac{1}{2}p_0 \right)\geq 0
    \iff \bar{\alpha} \; \geq \; \frac{\frac{1}{2}p_0}{\frac{1}{2}p_0 + p_1}, 
\end{align*}
and $\tilde{\sigma}_i^2(OX)=2$ if and only if
\begin{align*}
\delta\left(\frac{3}{4}p_0 + p_1 \right)&\left[p\alpha\left(\frac{3}{4}p_0 + p_1 \right) 
    - (1-p)\left(\frac{3}{4}p_0\right)\right]\\
\geq \max&\left\{ \delta^2\left(\frac{1}{2}p_0 + p_1 \right)
\left[p\alpha\left(\frac{1}{2}p_0 + p_1 \right) - (1-p)\left(\frac{1}{2}p_0\right)\right], \; 0
\right\}\\
\iff \bar{\alpha} \geq \max&\left\{ 
\frac{\frac{3}{4}p_0 \left(\frac{3}{4}p_0 +p_1\right) 
- \delta\left(\frac{1}{2}p_0 \right)\left(\frac{1}{2}p_0 +p_1\right)}{\left(\frac{3}{4}p_0 +p_1\right)^2 - \delta\left(\frac{1}{2}p_0 +p_1\right)^2} \; , \;
\frac{\frac{3}{4}p_0}{\frac{3}{4}p_0 + p_1} \right\}.
\end{align*}
Since for any $\delta\in (0,1)$, 
$$\frac{\frac{3}{4}p_0 \left(\frac{3}{4}p_0 +p_1\right) - \delta\left(\frac{1}{2}p_0 \right)\left(\frac{1}{2}p_0 +p_1\right)}{\left(\frac{3}{4}p_0 +p_1\right)^2 - \delta\left(\frac{1}{2}p_0 +p_1\right)^2} \; > \; \frac{\frac{3}{4}p_0}{\frac{3}{4}p_0 + p_1}, $$
2 is optimal for level 2 players if and only if 
\begin{align*}
   \bar{\alpha} \; \geq \; 
    \frac{\frac{3}{4}p_0 \left(\frac{3}{4}p_0 +p_1\right) 
    - \delta\left(\frac{1}{2}p_0 \right)\left(\frac{1}{2}p_0 +p_1\right)}{\left(\frac{3}{4}p_0 +p_1\right)^2 - \delta\left(\frac{1}{2}p_0 +p_1\right)^2}.
\end{align*}

Now suppose there is $K>1$ such that the statement holds for 
any $1\leq k\leq K$. We want to show it also holds for level $K+1$ players.
Notice that by induction hypothesis, $\tilde{\sigma}_i^{l+1}(OX) 
\leq \tilde{\sigma}_i^{l}(OX)$ for all $1\leq l \leq K-1$, implying
$\tilde{\sigma}_i^{K+1}(OX) \leq \tilde{\sigma}_i^{K}(OX)$
by step 2. If $\tilde{\sigma}_i^{K}(OX)\leq 3$,
then $\tilde{\sigma}_i^{K+1}(OX)\leq 3$ by step 2. Therefore, it suffices to 
focus on the case where $\tilde{\sigma}_i^{l}(OX)=4$ for all $1\leq l \leq K$.
In this case, level $K+1$ player's expected payoff of choosing $t\in \{2,3 \}$ is:
\begin{align*}
    \delta^{t-1}\left(\frac{5-t}{4}\frac{p_0}{\sum_{j=0}^K p_j} + 
    \frac{\sum_{j=1}^K p_j}{\sum_{j=0}^K p_j} \right)
    \left[p\alpha\left(\frac{5-t}{4}\frac{p_0}{\sum_{j=0}^K p_j} + 
    \frac{\sum_{j=1}^K p_j}{\sum_{j=0}^K p_j} \right)
    - (1-p)\left(\frac{5-t}{4}\frac{p_0}{\sum_{j=0}^K p_j}\right) \right],
\end{align*}
suggesting $4$ is a dominated strategy if and only if 
$$\bar{\alpha} \; \geq \; \frac{\frac{1}{2}p_0}{\frac{1}{2}p_0 + \sum_{j=1}^{K}p_j}.$$
If $\tilde{\sigma}_i^{K}(OX) = 2$, then $\tilde{\sigma}_i^{K+1}(OX)= 2$ by step 2.  
Thus, it suffices to consider the case where $\tilde{\sigma}_{i}^l(OX)\geq 3$ for all
$1\leq l \leq K$. 
In this case, $\tilde{\sigma}_i^{K+1}(OX) = 2$ if and only if 
\begin{align*}
\delta\left(\frac{3}{4}p_0 + \sum_{j=1}^K p_j \right)&\left[p\alpha\left(\frac{3}{4}p_0 +
\sum_{j=1}^K p_j \right) - (1-p)\left(\frac{3}{4}p_0\right)\right]\\
\geq \max&\left\{ \delta^2\left(\frac{1}{2}p_0 + \sum_{j=1}^K p_j \right)
\left[p\alpha\left(\frac{1}{2}p_0 + \sum_{j=1}^K p_j \right) 
- (1-p)\left(\frac{1}{2}p_0\right)\right], \; 0 \right\}\\
\iff \bar{\alpha} \; \geq \; &
\frac{\frac{3}{4}p_0 \left(\frac{3}{4}p_0 +\sum_{j=1}^K p_j\right) 
- \delta\left(\frac{1}{2}p_0 \right)\left(\frac{1}{2}p_0 +\sum_{j=1}^K p_j\right)}{\left(\frac{3}{4}p_0 +\sum_{j=1}^K p_j\right)^2 - 
\delta\left(\frac{1}{2}p_0 +\sum_{j=1}^K p_j\right)^2}.
\end{align*}

\bigskip

\noindent\textbf{\emph{Step 4:}} Lastly, this step characterizes level $k$ player $i$'s 
behavior when $x_{-i}=XX$
for level $k\geq 3$. Consider any level $K\geq 3$. For level $k$ players,
they would choose $3$ only if there is some level $2\leq l\leq k-1$ such that 
$\tilde{\sigma}_i^l(OX)=2$. Let $\tilde{L}_k^*$ be the lowest level below $k$  
that would choose 2 when seeing one dirty face.
Then level $k$ player $i$'s expected payoff of 3 is: 
\begin{align*}
    \delta^2 \left(\frac{1}{2}\frac{p_0}{\sum_{j=0}^{k-1} p_j} + 
    \frac{\sum_{j=1}^{k-1} p_j}{\sum_{j=0}^{k-1} p_j} \right)
    \left[ p\alpha\left(\frac{1}{2}\frac{p_0}{\sum_{j=0}^{k-1} p_j} + 
    \frac{\sum_{j=1}^{k-1} p_j}{\sum_{j=0}^{k-1} p_j} \right) - (1-p)\left(\frac{1}{2}\frac{p_0}{\sum_{j=0}^{k-1} p_j} + 
    \frac{\sum_{j=1}^{\tilde{L}_k^*-1} p_j}{\sum_{j=0}^{k-1} p_j} \right)
    \right],
\end{align*}
which dominates 4 if and only if 
\begin{align*}
    \bar{\alpha} \; \geq \; \frac{\frac{1}{2}p_0 + 
    \sum_{j=0}^{\tilde{L}_k^*-1} p_j}{\frac{1}{2}p_0 + \sum_{j=0}^k p_j}.
\end{align*}
Coupled with the existence of $\tilde{L}_k^*$, $\tilde{\sigma}_i^{k}(XX)=3$
if and only if 
\begin{align*}
    \bar{\alpha} \geq \max\left\{ \frac{\frac{3}{4}p_0
    \left(\frac{3}{4}p_0 + \sum_{j=1}^{\tilde{L}_k^*-1}p_j \right) - \delta
    \left(\frac{1}{2}p_0\right)\left(\frac{1}{2}p_0 + \sum_{j=1}^{\tilde{L}_k^*-1}p_j \right)}{\left(\frac{3}{4}p_0 +
    \sum_{j=1}^{\tilde{L}_k^*-1}p_j \right)^2 - \delta\left(\frac{1}{2}p_0 +
    \sum_{j=1}^{\tilde{L}_k^*-1}p_j \right)^2 } \;, \;  
    \frac{ \frac{1}{2}p_0 + \sum_{j=1}^{\tilde{L}_k^*-1}p_j}{\frac{1}{2}p_0 +  \sum_{j=1}^{k-1}p_j}  \right\}.
\end{align*}
This completes the proof of this proposition. $\blacksquare$

\subsection{Illustrative Example}

To illustrate the representation effect in three-person games, I characterize level 3 players' 
behavior in both the sequential and simultaneous games
when the distribution of levels follows Poisson(1.5). 
Similar to the analysis of two-person games, the 
set of dirty-faces games is the unit square on the $(\delta, \bar{\alpha})$-plane.

When observing one dirty face and one clean face, level 3 players cannot tell their faces 
for sure in period 1, so they will wait, no matter in sequential games or in simultaneous games. 
In period 2, level 3 players will claim to have a dirty face if and only if the expected payoff of 
$C$ is higher than the continuation value of choosing $W$. By Proposition \ref{prop:three_p_ext} and 
\ref{prop:three_p_strategic}, level 3 players will claim in period 2 in the sequential game 
if and only if 
\begin{align*}
    \bar{\alpha} \geq \frac{\left(\frac{1}{2} - \frac{1}{4}\gamma_3 \delta \right)p_0}{\left(\frac{1}{2} - \frac{1}{4}\gamma_3 \delta \right)p_0 + (1-\gamma_3\delta)(p_1 + p_2)} = 
    \frac{100-46\delta}{625-529\delta},
\end{align*}
and choose 2 in the simultaneous game if and only if 
$$\bar{\alpha} \geq \frac{\frac{3}{4}p_0
        \left(\frac{3}{4}p_0 + p_1 + p_2 \right) - \delta
        \left(\frac{1}{2}p_0\right)\left(\frac{1}{2}p_0 + p_1 + p_2 \right)}{\left(\frac{3}{4}p_0 +
        p_1 + p_2 \right)^2 - \delta\left(\frac{1}{2}p_0 +
        p_1 + p_2 \right)^2 } = \frac{162-100\delta}{729-625\delta}.$$
At period 3, level 3 players will claim to have a dirty face if and only if the expected payoff of $C$ 
is positive. Therefore, in the sequential game, level 3 players will claim in period 3 if and only if 
\begin{equation*}
\bar{\alpha} \; \geq \; \frac{\frac{1}{4}p_0}{\frac{1}{4}p_0 + p_1
+ p_2} = \frac{2}{23}, 
\end{equation*}
while in the simultaneous game, they will not choose always wait if and only if 
\begin{equation*}
\bar{\alpha} \; \geq \; \frac{\frac{1}{4}p_0}{\frac{1}{4}p_0 + p_1
+ p_2} = \frac{4}{25}. 
\end{equation*}

When observing two dirty faces, level 3 players cannot tell their face types in the 
first two periods, so they will wait in the first two periods. At period 3, level 3 players will 
claim if and only if (1) level 2 players will claim at period 2 when seeing only one 
dirty face,\footnote{Otherwise, if both level 1 and 2 players wait at period 2 when 
seeing only one dirty face, level 3 players cannot make inferences about their face types 
when the game proceeds to period 3.}
and (2) the expected payoff of $C$ is positive. Therefore, in the sequential game, 
it is optimal to claim at period 3 if and only if 
\begin{align*}
    \bar{\alpha} \geq \max\left\{ \frac{16-7\delta}{64-49\delta}, \;\;
    \frac{196}{529} \right\}.
\end{align*}
In the simultaneous game, it is optimal to claim at period 3 when observing two dirty faces if 
and only if 
\begin{align*}
    \bar{\alpha} \geq \max\left\{ \frac{27-16\delta}{81-64\delta}, \;\;
    \frac{16}{25} \right\}.
\end{align*}

Level 3 players' DCH optimal stopping periods in both sequential and simultaneous games
are plotted in Figure \ref{fig:dynamic_CH_3_person}. 
The definition of optimal stopping periods is naturally extended to three-person games. 
From this figure, we can observe two features that are different from the 
two-person games. First, when observing one dirty face and $\delta \rightarrow 1$, 
level 3 players will claim at period 2 if $\bar{\alpha} \geq 9/16$.
However, in two-person games, when $\delta \rightarrow 1$, players will 
always wait till the last period. This is because when there are more players,
the game is more likely to be randomly terminated, causing the players to claim
earlier even if the payoff is not discounted. Second, when observing two dirty faces, level 3 players' 
behavior at period 3 depends on $\delta$ \emph{even if this is the last period.}
The reason is that level 3 players' belief at period 3 depends on 
level 2 players' behavior at period 2 which depends on $\delta$.

\begin{figure}[htbp!]
\renewcommand\thefigure{A.1}
    \centering
    \includegraphics[width=\textwidth]{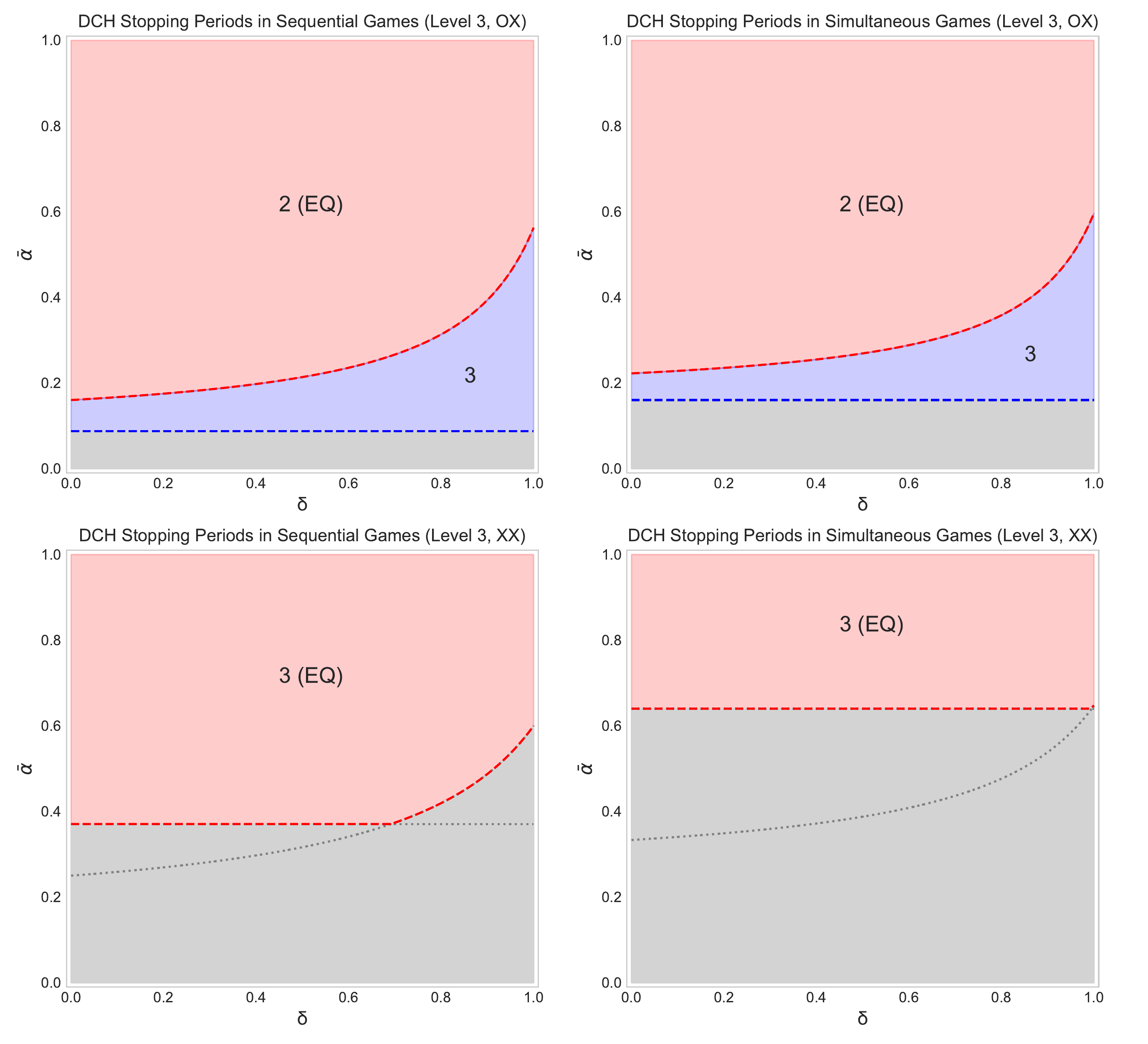}
    \caption{Level 3 players' DCH stopping periods in sequential (left column) and 
    simultaneous (right column) three-person three-period dirty-faces games where the 
    distribution of levels follows Poisson(1.5).}
    \label{fig:dynamic_CH_3_person}
\end{figure}

\begin{remark}
In this illustrative example, 
DCH predicts level 3 players tend to claim earlier in sequential games than 
in simultaneous games because (1) at information set $(2,OX)$, $(162-100\delta)/(729-625\delta) 
> (100-46\delta)/ (625 - 529\delta)$, (2) at information set $(3,OX)$, $4/25 > 2/23$, and 
(3) at information set $(3,XX)$, 
\begin{align*}
    \max\left\{ \frac{27-16\delta}{81-64\delta}, \;\;
    \frac{16}{25} \right\} > \max\left\{ \frac{16-7\delta}{64-49\delta}, \;\;
    \frac{196}{529}\right\}.
\end{align*}
\end{remark}

To summarize, the analysis of three-person three-period games demonstrates
how the DCH solution varies with the representations and the number of players in a game. 
The prediction of the equilibrium theory only depends on the number of dirty faces, 
not the number of players. This sharply contrasts with DCH. 
The intuition is that when there are more players, the game 
is more likely to be randomly terminated by level 0 players, and hence
strategic players' behavior is affected.

\newpage
\section{Detailed Analysis of Bayer and Chan (2007) Data (For Online Publication)}
\label{sec:old_dirty_appendix}

\subsection{Data Description}

This section revisits the dirty-faces experimental data by \cite{bayer2007dirty}. The 
description of the experimental setting can be found in the main text section \ref{subsec:calibration},
and the instructions and screenshots can be found in \cite{bayer2007dirty} Appendix A.

Following previous notations, I use $(t,x_{-i})$ to denote the situation where subject $i$
sees type $x_{-i}$ at period $t$. After excluding the data from the case where there is no 
public announcement, the raw data at each information set is reported in Table \ref{tab:round_1_data}.
Each entry in the table states the number of observations and 
the percentage of the choices that follow the equilibrium predictions.
For instance, at information set $(t,x_{-i})=(2,X)$, 
there are 170 choices and 62 percent of the choices are $C$,
which is the action predicted by the equilibrium.

\begin{table}[htbp!]
\renewcommand\thetable{A.1}
\centering
\begin{threeparttable}
\caption{Experimental Data from \cite{bayer2007dirty}}
\label{tab:round_1_data}
\begin{tabular}{ccccccc}
\hline
 & \multicolumn{6}{c}{Number of Players} \\ \cline{2-7} 
 & \multicolumn{2}{c}{2} &  & \multicolumn{3}{c}{3} \\ \cline{2-3} \cline{5-7} 
$x_{-i}$ & $O$ & $X$ &  & $OO$ & $OX$ & $XX$ \\
EQ & $C$ & $WC$ &  & $C$ & $WC$ & $WWC$ \\ \hline
Period & \multicolumn{6}{c}{Number of Obs (EQ \%)} \\
1 & 123 (0.94) & 391 (0.79) &  & \phantom{1}48 (0.92) & 280 (0.61) & 320 (0.76) \\
2 & \phantom{12}6 (0.50) & 170 (0.62) &  & \phantom{14}2 (0.50) & \phantom{1}60 (0.58) & 145 (0.79) \\
3 & --- & --- &  & --- & \phantom{1}10 (0.20) & \phantom{1}56 (0.36) \\ \hline
\end{tabular}
\begin{tablenotes}
\footnotesize
\item Note: In Treatment 1, there are 21 groups of subjects (42 subjects in total), and in 
Treatment 2, there are 16 groups of subjects (48 subjects in total). Because each group
plays 14 rounds, the data set consists of $(21+16)\times14 = 518$ games.
\end{tablenotes}

\end{threeparttable}
\end{table}

From Table \ref{tab:round_1_data}, we can observe that the behavior aligns 
with the equilibrium prediction when players do not see any dirty face. 
In this situation ($x_{-i}=O$ or $OO$), players are aware that their face type is $X$ and choose $C$ in 
period 1. However, the behavior becomes less consistent with the equilibrium as the reasoning complexity 
increases. When players see only one dirty face ($x_{-i}=X$ or $OX$), they should realize that their face 
type is $X$ as the game progresses to period 2. 
However, the empirical data show that only 62\% and 58\% of players in Treatment 1 and 2, 
respectively, are able to do so. 
Furthermore, when players see two dirty faces ($x_{-i}=XX$), only 30\% 
of the players claim to have a dirty face in period 3.

These observations suggest that the equilibrium fails to explain a significant portion of the data. In the 
following analysis, I compare the fitness of the DCH model with that of the standard CH model and the agent 
quantal response equilibrium (AQRE) proposed by \cite{mckelvey1998quantal}. By comparing the DCH and the 
standard CH models, I can quantify the improvement achieved by incorporating learning from past actions 
into the CH framework. On the other hand, AQRE is an equilibrium model designed for extensive games, 
where players make stochastic choices and assume that other players do the same. The comparison between the 
DCH and AQRE demonstrates how hierarchical thinking models can generate statistically comparable 
predictions as equilibrium-based models.

\subsection{Likelihood Functions}
\label{subsecsec:likelihood_functions}

This section derives the likelihood functions. 
For the cognitive hierarchy theories, I follow \cite{camerer2004cognitive} to 
assume the prior distribution of levels follows Poisson distribution.
Therefore, for both of the Poisson-DCH and the standard Poisson-CH, there is one parameter to be 
estimated---the average number of levels $\tau$. 
For AQRE, I follow \cite{mckelvey1998quantal} to estimate the 
logit-AQRE which has a single parameter $\lambda$.

\subsubsection{Poisson-CH Models}

The Poisson-CH models assume each player's level is i.i.d. drawn from 
$\left(p_k\right)_{k=0}^{\infty}$ where
\begin{align*}
    p_k \equiv \frac{e^{-\tau}\tau^k}{k!}, \;\;\;\mbox{ for all } \; k=0,1,2,\ldots 
\end{align*}
and $\tau>0$. Because $\tau$ is the mean and variance of the Poisson distribution, 
the economic meaning of $\tau$ is the average level of sophistication among the population. 

I first construct the likelihood function for the Poisson-DCH model.  
For each subject $i$, let $\Pi_i$ denote 
the set of information sets that subject $i$ has encountered
in the game, and let $\mathcal{I}_i = (t, x_{-i})$ denote a generic information set.
At any information set $\mathcal{I}_i$, subject $i$ can choose $c_i \in \{C, W \}$. 
Let $P_k(c_i | \mathcal{I}_i, \tau )$
be the probability of level $k$ players
choosing $c_i$ at information set $\mathcal{I}_i$.
Moreover, let $f(k|\mathcal{I}_i, \tau)$ be the posterior distribution 
of levels at information set $\mathcal{I}_i$.
At period 1, $f(k|\mathcal{I}_i, \tau) = e^{-\tau}\tau^k / k!$. For later periods, 
$f(k|\mathcal{I}_i, \tau)$ given any $\tau$ can be analytically solved by
Proposition \ref{prop_extensive_dirty} (two-person games) 
and Proposition \ref{prop:three_p_ext} (three-person games). 
Finally, the predicted choice probability for $c_i$ at 
information set $\mathcal{I}_i$ is simply 
the aggregation of best responses from all levels 
weighted by the proportion $f(k|\mathcal{I}_i, \tau)$:
\begin{equation*}
\mathcal{D}(c_i | \mathcal{I}_i, \tau) \equiv \sum_{k=0}^{\infty} 
f(k|\mathcal{I}_i, \tau) P_k(c_i | \mathcal{I}_i, \tau ).    
\end{equation*}
Consequently, the log-likelihood function for the DCH model can be formed 
by aggregating over every subject $i$, actions $c_i$ 
and information set $\mathcal{I}_i$:
\begin{equation*}
\ln L^D (\tau) = \sum_{i}\sum_{\mathcal{I}_i \in \Pi_i} \sum_{c_i \in \{W,C \}}
\mathds{1}\{c_i,\mathcal{I}_i\}\ln \left[ \mathcal{D}(c_i | \mathcal{I}_i, \tau)\right],
\end{equation*}
where $\mathds{1}\{c_i,\mathcal{I}_i\}$ is the 
indicator function which is 1 when subject $i$ chooses 
$c_i$ at $\mathcal{I}_i$.

Second, the log-likelihood function for the standard Poisson-CH model can be 
constructed in the similar way. Given any $\tau$, the standard Poisson-CH model
predicts a probability distribution over $\{1,\ldots, T, T+1 \}$ 
(earliest period to choose $C$ or always $W$) for each level of players
conditional on the announcement and other players' faces.
Following previous notations, the probability of level $k$ subject $i$ 
choosing $t$ conditional on $x_{-i}$ is denoted by 
$\tilde{\sigma}_{i}^k(t|x_{-i})$, which can be analytically solved by 
Proposition \ref{prop_strategic_dirty} (two-person games) and
Proposition \ref{prop:three_p_strategic} (three-person games).
Therefore, subject $i$'s predicted choice probability for $t\in \{1,\ldots, T, T+1 \}$
conditional on $x_{-i}$ is the aggregation of choice frequencies of all 
levels weighted by Poisson($\tau$): 
\begin{equation*}
\tilde{\mathcal{S}}(t |x_{-i} , \tau) = \sum_{k=0}^{\infty} 
\frac{e^{-\tau}\tau^k}{k!}  \tilde{\sigma}_{i}^k(t|x_{-i}).    
\end{equation*}
Since $\tilde{\sigma}_{i}^0(t|x_{-i})= \frac{1}{T+1}$ for all $t$, 
$\tilde{\mathcal{S}}(t |x_{-i} , \tau) >0$ for all $t$.
Moreover, the conditional probability to choose $C$ or $W$ at 
information set $\mathcal{I}_i= (t,x_{-i})$ can be computed by: 
\begin{align*}
\mathcal{S}(C | \mathcal{I}_i, \tau) &= \frac{\tilde{\mathcal{S}}(t |x_{-i} , \tau)}{\sum_{t'\geq t} \tilde{\mathcal{S}}(t' | x_{-i} , \tau)} \;\;\mbox{ and }\;\;
\mathcal{S}(W | \mathcal{I}_i, \tau) = 1 - \mathcal{S}(C | \mathcal{I}_i, \tau).
\end{align*}
Finally, the log-likelihood function for the standard CH model can be constructed by 
aggregating over every subjects $i$, actions $c_i$, and 
information set $\mathcal{I}_i$:
\begin{equation*}
\ln L^S (\tau) = \sum_{i}\sum_{\mathcal{I}_i \in \Pi_i} \sum_{c_i \in \{W,C \}}
\mathds{1}\{c_i,\mathcal{I}_i\}\ln \left[ \mathcal{S}(c_i | \mathcal{I}_i, \tau)\right].
\end{equation*}

\subsubsection{Logit-AQRE Model}

For the purpose of illustrate, I only derive the likelihood function for two-person games. 
The likelihood function for three-person games can be derived by a similar calculation.

Let $Q(c_i | \mathcal{I}_i, \lambda)$ be the probability of subject $i$ choosing $c_i$ 
at information set $\mathcal{I}_i$ predicted by the logit-AQRE.
In the two-person two-period dirty-faces game, each player's strategy is defined by 
a four-tuple $(q_1, q_2, r_1, r_2)$ which 
corresponds to $Q(C | 1,O,\lambda)$, $Q(C | 2,O,\lambda)$, $Q(C | 1,X,\lambda)$,
and $Q(C | 2,X,\lambda)$, respectively.
At information set $(t,x_{-i}) = (1,O)$, players would estimate the payoff of $C$ and $W$ by
\begin{align*}
    U_{1,O}(C) &= \alpha + \epsilon_{1,O,C} \\
    U_{1,O}(W) &= \delta\alpha(1-r_1)q_2 + \epsilon_{1,O,U},
\end{align*}
where $\epsilon_{1,O,C}$ and $\epsilon_{1,O,W}$ are independent random variables with a Weibull distribution 
with the precision parameter $\lambda$. Then the logit formula suggests 
\begin{equation*}
q_1 = \frac{1}{1 + exp\left\{\lambda\left[\delta\alpha(1-r_1)q_2 - \alpha\right] \right\}}.    
\end{equation*}
Similarly, $q_2$ can be expressed by: 
$$q_2 = \frac{1}{1 + exp\left\{-\delta\alpha\lambda \right\}}.$$
On the other hand, when observing a dirty face and the game proceeds to period 2, 
players' posterior beliefs become:
$$\mu \equiv \Pr(X|2,X) = \frac{p(1-r_1)}{p(1-r_1) + (1-p)(1-q_1)}
= \frac{1}{1 + \left(\frac{1-p}{p}\right)\left(\frac{1-q_1}{1-r_1}\right)},$$
and hence the expected payoff to choose $C$ at information set $(2,X)$ is:
$$\delta\left[\alpha\mu - (1-\mu) \right] 
= \delta \left[(1+\alpha)\mu -1 \right].$$
As a result, $r_2$ satisfies that 
$$r_2 = \frac{1}{1 + exp\left\{\lambda\delta\left[1-(1+\alpha)\mu \right]\right\}}.$$
Finally, the expected payoff of choosing $C$ at information set $(1,X)$ is 
$\alpha p -(1-p)$, while the expected payoff of $W$ is
\begin{align*}
    \underbrace{[p(1-r_1) + (1-p)(1-q_1)]}_{\mbox{prob. to reach period 2}}
    r_2 \delta \left[(1+\alpha)\mu -1 \right] \equiv A,
\end{align*}
and therefore, $r_1$ can be expressed by:
\begin{align*}
r_1 = \frac{1}{1 + exp\left\{\lambda\left[ A + (1-p) -\alpha p\right]\right\}}.
\end{align*}
As plugging $p=2/3$, $\delta=4/5$ and $\alpha=2/3$ into the choice probabilities, 
we can obtain that
\begin{align*}
    r_1 &= \frac{1}{1 + exp\left\{\lambda\left[\frac{2}{15}(1-r_1)r_2 - \frac{4}{15}(1-q_1)r_2
    + \frac{1}{6}\right]\right\}}\\
    r_2 &= \frac{1}{1 + exp\left\{\lambda\left[\frac{4}{5} - 
    \frac{2-2r_1}{3-2r_1 - q_1} \right]\right\}}\\
    q_1 &= \frac{1}{1 + exp\left\{\lambda\left[\frac{1}{5}(1-r_1)q_2 - \frac{1}{4}\right] \right\}}\\
    q_2 &= \frac{1}{1 + exp\left\{-\frac{1}{5}\lambda \right\}}.
\end{align*}
Given each $\lambda$, the system of four equations with four unknowns can be solved uniquely.
Besides, for each $\mathcal{I}_i$, $Q(W|\mathcal{I}_i, \lambda)
= 1 - Q(C|\mathcal{I}_i, \lambda)$.
Thus, the log-likelihood function can be formed by aggregating over every subject $i$, action $c_i$,
and information set $\mathcal{I}_i$: 
\begin{equation*}
\ln L^Q (\lambda) = \sum_{i}\sum_{\mathcal{I}_i \in \Pi_i} \sum_{c_i \in \{W,C \}}
\mathds{1}\{c_i,\mathcal{I}_i\}\ln \left[ Q(c_i | \mathcal{I}_i, \lambda)\right].
\end{equation*}

\subsection{Estimation Results}

The Poisson-DCH, standard Poisson-CH and the AQRE models are estimated by maximum likelihood 
estimation. Table \ref{tab:estimation_result_overall} reports the estimation results on Treatment 1 
and Treatment 2 data, showing the estimated parameters and 
the fitness of each model. Comparing the fitness of these models, I find that the log-likelihood of 
DCH is significantly higher than standard CH (Vuong Test p-value $<0.001$ for both treatments), 
while it is not significantly different from AQRE (Treatment 1:
p-value $=0.144$; Treatment 2: p-value $=0.184$). This result suggests in both Treatment 1 and 2,  
DCH outperforms the standard CH in capturing the empirical patterns and generates predictions that 
are statistically comparable to other equilibrium-based behavioral solution concepts.

\begin{table}[htbp!]
\renewcommand\thetable{A.2}
\centering
\caption{Estimation Results for Treatment 1 and Treatment 2 Data}
\label{tab:estimation_result_overall}

\begin{threeparttable}

\begin{tabular}{ccccclccc}
\hline
 &  & \multicolumn{3}{c}{Two-Person Games} &  & \multicolumn{3}{c}{Three-Person Games} \\ \cline{3-5} \cline{7-9} 
 &  & DCH & \begin{tabular}[c]{@{}c@{}}Standard\\ CH\end{tabular} & AQRE &  & DCH & \begin{tabular}[c]{@{}c@{}}Standard\\ CH\end{tabular} & AQRE \\ \hline
Parameters & $\tau$ & 1.269 & 1.161 & --- &  & 0.370 & 0.140 & --- \\
 & S.E. & (0.090) & (0.095) & --- &  & (0.043) & (0.039) & --- \\
 & $\lambda$ & --- & --- & 7.663 &  & --- & --- & 5.278 \\
 & S.E. & --- & --- & (0.493) &  & --- & --- & (0.404) \\ \hline
Fitness & LL & -360.75\phantom{-} & -381.46\phantom{-} & -368.38\phantom{-} &  & -575.30 & -608.45 & -565.05 \\
 & AIC & 723.50 & 764.91 & 738.76 &  & 1152.61 & 1218.89 & 1132.11 \\
 & BIC & 728.04 & 769.45 & 743.29 &  & 1157.43 & 1223.72 & 1136.93 \\ \hline
Vuong Test &  &  & 6.517 & 1.463 &  &  & 3.535 & -1.330\phantom{-} \\
p-value &  &  & $<0.001\phantom{<}$ & 0.144 &  &  & $<0.001\phantom{<}$ & 0.184 \\ \hline
\end{tabular}
\begin{tablenotes}
\footnotesize
\item Note: There are 294 games (rounds $\times$ groups) in 
Treatment 1 and 224 games in Treatment 2.
\end{tablenotes}

\end{threeparttable}

\end{table}

Comparing the estimation results of Treatment 1 and 2, I observe that there is more randomness in 
three-person games compared to two-person games. 
In two-person games, the DCH estimates indicate that players can think 1.269 steps 
(95\% C.I. $=[1.093, 1.445]$) on average, while in three-person games, players can only think 
an average of 0.370 steps 
(95\% C.I. $=[0.286, 0.454]$). Additionally, the estimation result of AQRE suggests 
that as the game changes from two-person games to three-person games, 
the precision of decision-making decreases significantly (from 7.663 to 5.278). 
This implies that players are less likely to make best responses in three-person games.

To analyze the differences between the models in detail, 
I compare the choice probabilities predicted by each model. 
Figure \ref{fig:two_person_choice_prob} illustrates the choice probabilities in two-person games, while 
Figure \ref{fig:three_person_choice_prob} displays the choice probabilities in three-person games.
Comparing the DCH and the standard CH models, I observe that the standard CH model 
generally underestimates the probability of choosing $C$ in period 1. 
In two-person games, the empirical frequencies of choosing $C$ at information sets $(1,O)$ and $(1,X)$ 
are 0.943 and 0.210, respectively. 
Yet, the predictions of the standard CH model are 0.791 and 0.104 for the same information sets. A 
similar pattern of underestimation is also evident in three-person games.

\begin{figure}[htbp!]
\renewcommand\thefigure{A.2}
    \centering
    \includegraphics[width=\textwidth]{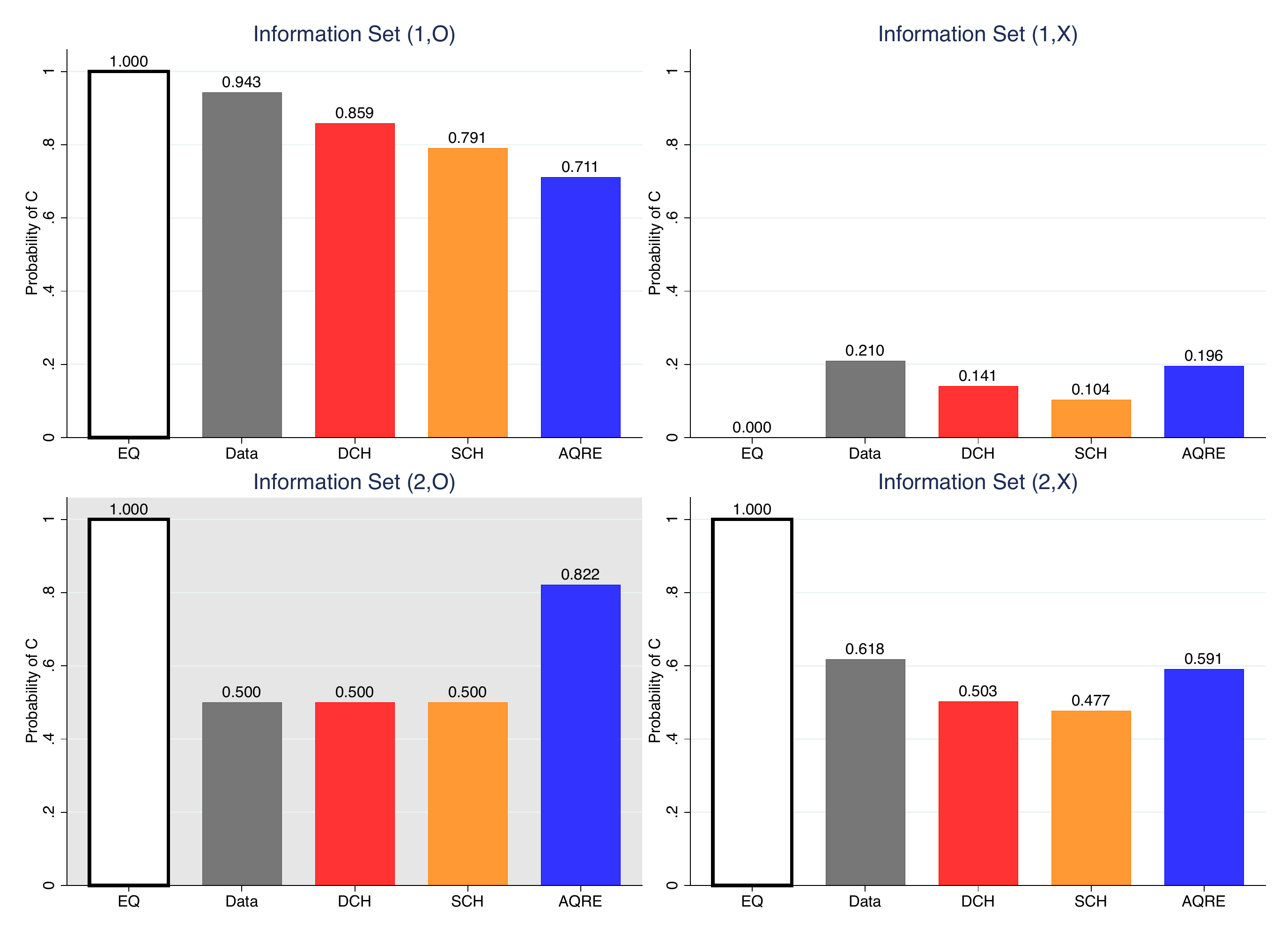}
    \caption{The choice probabilities in two-person games at different information sets. Each panel plots the empirical 
    choice frequencies and the predictions of different models at 
    one information set. The gray panel represents the 
    off-equilibrium path information set.}
    \label{fig:two_person_choice_prob}
\end{figure}

The underestimation is primarily caused by the difference in the specifications of level 0 players' 
behavior. In two-person games, the standard CH model assumes that level 0 players 
uniformly randomize across the set $\{1, 2, 3\}$. Consequently, 
the probability of level 0 players choosing $C$ in period 1 
according to the standard CH model is $1/3$. In contrast, in the DCH model, 
level 0 players uniformly randomize at every information set, 
resulting in a probability of $1/2$ for them to choose $C$.
Similarly, in three-person games, the standard CH model assumes that level 0 players 
uniformly randomize across the set $\{1, 2, 3, 4\}$, 
leading to a probability of $1/4$ for them to choose $C$ in period 1. 
In contrast, in the DCH model, level 0 players' behavior 
remains the same across both two-person and three-person games.
These differences in level 0 players' behavior contribute to the underestimation of the probability of 
choosing $C$ in the standard CH model compared to DCH.

Moreover, the key difference between the CH approach and AQRE is highlighted in the off-equilibrium-path 
information sets.\footnote{When $x_{-i}=OO$, 
the equilibrium predicts that players will choose $C$ in period 1, 
resulting in the game not proceeding beyond period 2. 
Similarly, when $x_{-i}=OX$, the equilibrium suggests that players should choose $C$ in period 2, 
preventing the game from progressing to period 3.}  
Conceptually, the reason why the game could proceed to the off-equilibrium-path 
information sets differs between the CH approach and AQRE. 
From the perspective of AQRE, the off-equilibrium-path information sets are reached due to mistakes. 
As a result, AQRE predicts a high probability of choosing $C$ at these off-equilibrium-path 
information sets because the expected payoff of choosing $C$ is much higher than $W$ 
at these information sets. By contrary, in the CH approach, 
the off-equilibrium-path information sets are reached because the players are not sophisticated enough. 
For instance, when observing no dirty face, players should immediately choose $C$ 
since it is a dominant strategy. If someone doesn't choose $C$, they are definitely a level 0 player.

\begin{figure}[htbp!]
\renewcommand\thefigure{A.3}
    \centering
    \includegraphics[width=\textwidth]{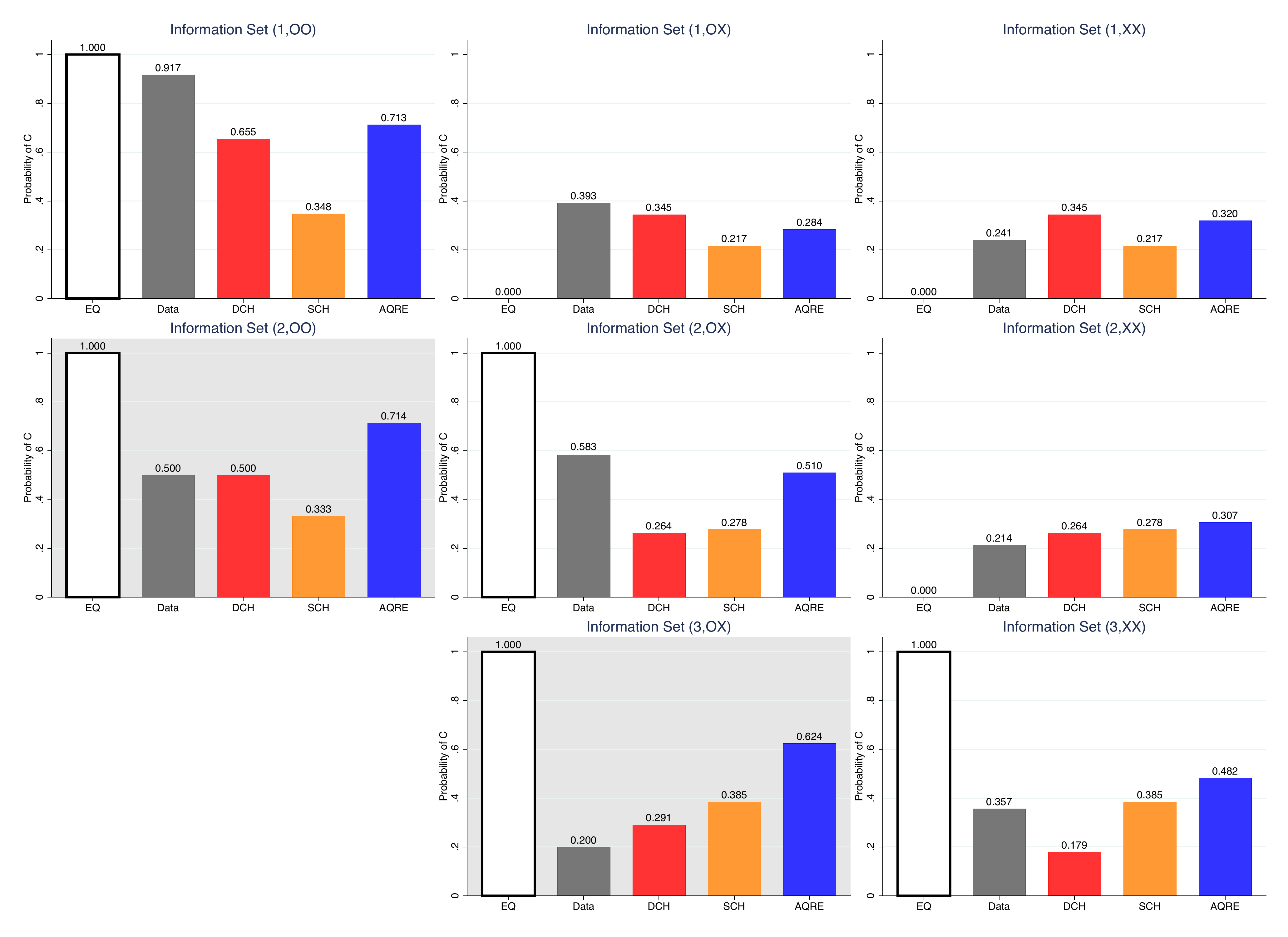}
    \caption{The choice probabilities in three-person games at different information sets. Each panel plots the empirical 
    choice frequencies and the predictions of different models at 
    one information set. The gray panels represent the 
    off-equilibrium-path information sets.}
    \label{fig:three_person_choice_prob}
\end{figure}

From the choice probabilities, it can be observed that DCH provides the most accurate predictions 
at off-path information sets, regardless of whether it is in two-person or three-person games. 
At information sets $(2,O)$ and $(2,OO)$, the empirical choice probabilities of $C$ are 0.5, 
which are correctly predicted by DCH. Furthermore, at the information set $(3,OX)$, 
the empirical choice probability of $C$ is 0.2, while the predictions of DCH, standard CH, and 
AQRE are 0.291, 0.385, and 0.624, respectively.

\begin{table}[htbp!]
\renewcommand\thetable{A.3}
\centering
\caption{Estimation Results for Pooled Data}
\label{tab:pooled_estimation}
\begin{threeparttable}

\begin{tabular}{ccccc}
\hline
 &  & DCH & \begin{tabular}[c]{@{}c@{}}Standard\\ CH\end{tabular} & AQRE \\ \hline
Parameters & $\tau$ & 1.030 & 0.241 & --- \\
 & S.E. & (0.060) & (0.033) & --- \\
 & $\lambda$ & --- & --- & 6.235 \\
 & S.E. & --- & --- & (0.302) \\ \hline
Fitness & LL & -956.92 & -1047.12\phantom{-} & -940.65 \\
 & AIC & 1915.84 & 2096.23 & 1883.30 \\
 & BIC & 1921.22 & 2101.62 & 1888.69 \\ \hline
Vuong Test &  &  & 7.513 & -1.363\phantom{-} \\
p-value &  &  & $<0.001\phantom{<}$ & 0.173 \\
LR Test $\chi^2_{(1)}$ &  & 41.74 & 114.42 & 14.44 \\
p-value &  & $<0.001\phantom{<}$ & $<0.001\phantom{<}$ & $<0.001\phantom{<}$ \\ \hline
\end{tabular}
\begin{tablenotes}
\footnotesize
\item Note: The likelihood ratio test is testing if the log-likelihood of 
two-parameter models (Treatment 1 and 2) is significantly higher than the log-likelihood 
of one-parameter models.
\end{tablenotes}

\end{threeparttable}
\end{table}

In addition, I estimate the three models using the pooled data, and the results are reported 
in Table \ref{tab:pooled_estimation}. Consistent with the results from the two-person games 
and three-person games, it can be observed that DCH provides a significantly better 
fit to the data compared to the standard CH model (Vuong test: p-value $<0.001$). 
However, there is no statistically significant difference between DCH and AQRE 
(Vuong test p-value $=0.173$). Furthermore, I conduct a likelihood ratio test on 
all three models to assess whether allowing different parameters for two-person and three-person 
games can significantly improve the model fit. 
The results indicate that the heterogeneous models are significantly better than the homogeneous models. 
Taken together, these findings lead to the conclusion that both the level of sophistication and the 
precision vary with the complexity of the games.



To summarize, it is not surprising that DCH can provide a better explanation for 
the data compared to the misspecified standard CH model in dynamic games. 
However, what is surprising is that when the CH model is correctly specified, 
the estimated average level of sophistication is $1.03$, 
which falls within the expected range of a ``regular'' $\tau$ value between 1 and 2, 
as predicted by \cite{camerer2004cognitive}.

\newpage
\section{Supplementary Analysis for Experimental Data (For Online Publication)}
\label{sec:exp_analysis_appendix}

\subsection{Supplementary Tables}
\label{sec:robustness_appendix}

This section includes all the supplementary tables from the experiment. 
Table \ref{tab:appendix_summary} lists the empirical frequencies of choosing $C$ at each 
information set for both treatments. Table \ref{tab:appendix_summary_payoff} presents the empirical 
frequencies of choosing $C$ at information set $(2,X)$ for different payoff structures and 
treatments.

\begin{table}[htbp!]
\renewcommand\thetable{A.4}
\centering
\caption{The Empirical Frequencies of $C$ at Each Information Set}
\label{tab:appendix_summary}
\begin{adjustbox}{width=\columnwidth,center}
\begin{threeparttable}
\begin{tabular}{cccccccccccccccc}
\hline
 & \multicolumn{7}{c}{Sequential Treatment} &  & \multicolumn{7}{c}{Simultaneous Treatment} \\ \cline{2-8} \cline{10-16} 
$x_{-i}$ & \multicolumn{3}{c}{$O$} &  & \multicolumn{3}{c}{$X$} &  & \multicolumn{3}{c}{$O$} &  & \multicolumn{3}{c}{$X$} \\ \cline{2-4} \cline{6-8} \cline{10-12} \cline{14-16} 
 & Obs & Claim \% & s.d. &  & Obs & Claim \% & s.d. &  & Obs & Claim \% & s.d. &  & Obs & Claim \% & s.d. \\ \hline
Periods &  &  &  &  &  &  &  &  &  &  &  &  &  &  &  \\
1 & 148 & 0.845 & 0.364 &  & 572 & 0.313 & 0.464 &  & 148 & 0.811 & 0.393 &  & 548 & 0.263 & 0.441 \\
2 & 16 & 0.438 & 0.512 &  & 210 & 0.600 & 0.491 &  & 28 & 0.250 & 0.441 &  & 404 & 0.223 & 0.417 \\
3 & 4 & 0.000 & 0.000 &  & 34 & 0.206 & 0.410 &  & 21 & 0.190 & 0.402 &  & 314 & 0.172 & 0.378 \\
4 & 3 & 0.000 & 0.000 &  & 21 & 0.190 & 0.402 &  & 16 & 0.250 & 0.447 &  & 259 & 0.131 & 0.338 \\
5 & 2 & 0.500 & 0.707 &  & 14 & 0.214 & 0.426 &  & 14 & 0.143 & 0.363 &  & 227 & 0.172 & 0.378 \\ \hline
\end{tabular}
\begin{tablenotes}
\footnotesize
\item Note: For the simultaneous treatment, the choice data at the information set level are implied by the contingent
strategies. For instance, choosing the contingent strategy ``claim at period 4'' implies that the subject will
wait from period 1 to period 3 and claim in period 4.
\end{tablenotes}
\end{threeparttable}
\end{adjustbox}
\end{table}

\begin{table}[htbp!]
\renewcommand\thetable{A.5}
\centering
\caption{The Empirical Frequencies of $C$ at Information Set $(2,X)$ 
for Different Games}
\label{tab:appendix_summary_payoff}
\begin{tabular}{cccccccc}
\hline
 & \multicolumn{3}{c}{Sequential Treatment} &  & \multicolumn{3}{c}{Simultaneous Treatment} \\ \cline{2-4} \cline{6-8} 
$(\delta, \bar{\alpha})$ & Obs & Claim \% & s.d. &  & Obs & Claim \% & s.d. \\ \hline
\multicolumn{1}{l}{Diagnostic Games} &  &  &  &  &  &  &  \\
(0.60, 0.45) & 39 & 0.564 & 0.502 &  & 78 & 0.256 & 0.439 \\
(0.95, 0.80) & 36 & 0.667 & 0.479 &  & 59 & 0.237 & 0.429 \\
\multicolumn{1}{l}{Control Games} &  &  &  &  &  &  &  \\
(0.60, 0.80) & 38 & 0.789 & 0.413 &  & 61 & 0.361 & 0.484 \\
(0.80, 0.45) & 35 & 0.543 & 0.505 &  & 67 & 0.134 & 0.344 \\
(0.80, 0.80) & 24 & 0.542 & 0.509 &  & 63 & 0.190 & 0.396 \\
(0.95, 0.45) & 38 & 0.474 & 0.506 &  & 76 & 0.171 & 0.379 \\ \hline
\end{tabular}
\end{table}

To compute the measure of violation of invariance under strategic equivalence of each payoff structure 
$(\delta, \bar{\alpha})$, I run the following regression on the data of information set $(2,X)$:
\begin{align}\tag{A.7}
    \mathds{1}\{\mbox{claim}\}_i = \alpha_0 + \alpha_1 \mathds{1}\{\mbox{sequential}\}_i + \epsilon_i
\end{align}
where $\mathds{1}\{\mbox{claim}\}_i$ is the dummy variable for player $i$ choosing $C$ and 
$\mathds{1}\{\mbox{sequential}\}_i$ is the dummy variable for the sequential treatment.
Table \ref{tab:appendix_rep_effect} reports the results for all payoff structures. The standard errors 
are clustered at the session level.

\begin{table}[htbp!]
\renewcommand\thetable{A.6}
\centering
\caption{The Magnitude of the Treatment Effect for Different Games}
\label{tab:appendix_rep_effect}
\begin{adjustbox}{width=\columnwidth,center}
\begin{threeparttable}
\begin{tabular}{cccccccc}
\hline
Payoff Structure &  &  &  &  &  &  &  \\
$(\delta,\bar{\alpha})$ &  & (0.60, 0.45) & (0.60, 0.80) & (0.80, 0.45) & (0.80, 0.80) & (0.95, 0.45) & (0.95, 0.80) \\ \hline
Sequential Treatment &  & \phantom{***}0.308*** & \phantom{*}0.429* & \phantom{*}0.409* & \phantom{**}0.351** & \phantom{**}0.303** & \phantom{**}0.429** \\
 &  & (0.066) & (0.160) & (0.160) & (0.101) & (0.073) & (0.110) \\
Constant &  & \phantom{***}0.256*** & \phantom{***}0.361*** & \phantom{***}0.134*** & \phantom{***}0.190*** & \phantom{**}0.171** & \phantom{**}0.237** \\
 &  & (0.057) & (0.070) & (0.018) & (0.025) & (0.049) & (0.062) \\ \hline
N &  & 117 & 99 & 102 & 87 & 114 & 95 \\
R-squared &  & 0.0914 & 0.1744 & 0.1889 & 0.1203 & 0.1028 & 0.1808 \\ \hline
\end{tabular}
\begin{tablenotes}
\footnotesize
\item Note: The standard errors are clustered at the session level. * $p<0.05$, ** $p<0.01$, *** $p<0.001$.
\end{tablenotes}
\end{threeparttable}
\end{adjustbox}
\end{table}

Finally, Table \ref{tab:appednix_rt_sequential} and \ref{tab:appednix_rt_simulaneous}
report the distributions of reaction times when players 
see a dirty face in the sequential and simultaneous treatments, respectively.

\begin{table}[htbp!]
\renewcommand\thetable{A.7}
\centering
\caption{Reaction Times (seconds) when Seeing $X$ in the Sequential Treatment}
\label{tab:appednix_rt_sequential}
\begin{tabular}{lcccccc}
\hline
Periods & Obs & Mean & s.d. & Q1 & Median & Q3 \\ \hline
1 & 572 & 11.29 & 9.845 & 5.359 & 7.574 & 13.72 \\
2 & 210 & 8.172 & 6.531 & 4.495 & 6.106 & 9.745 \\
3 & 34 & 7.530 & 3.816 & 4.785 & 6.277 & 11.44 \\
4 & 21 & 6.901 & 4.135 & 3.150 & 6.299 & 9.767 \\
5 & 14 & 7.663 & 5.597 & 3.231 & 6.677 & 8.856 \\ \hline
All & 851 & 10.20 & 8.917 & 5.000 & 7.152 & 12.08 \\ \hline
\end{tabular}
\end{table}

\begin{table}[htbp!]
\renewcommand\thetable{A.8}
\centering
\caption{Reaction Times (seconds) when Seeing $X$ in the Simultaneous Treatment}
\label{tab:appednix_rt_simulaneous}
\begin{tabular}{lcccccc}
\hline
Stopping Strategies & Obs & Mean & s.d. & Q1 & Median & Q3 \\ \hline
Claim at 1 & 144 & 11.93 & 8.804 & 5.811 & 8.958 & 15.91 \\
Claim at 2 & 90 & 13.85 & 10.18 & 7.191 & 11.32 & 16.62 \\
Claim at 3 & 54 & 17.64 & 12.81 & 8.095 & 14.30 & 24.61 \\
Claim at 4 & 34 & 21.15 & 15.16 & 12.32 & 16.46 & 22.96 \\
Claim at 5 & 39 & 23.34 & 14.43 & 13.84 & 19.76 & 27.83 \\
Always Wait & 187 & 14.24 & 11.09 & 7.342 & 10.48 & 20.24 \\ \hline
All & 548 & 15.04 & 11.58 & 7.413 & 11.32 & 19.48 \\ \hline
\end{tabular}
\end{table}

\subsection{Likelihood Functions}
\label{subsec:likelihood_qcse_qdch}

\subsubsection*{Quantal Cursed Sequential Equilibrium}

The ``Quantal Cursed Sequential Equilibrium (QCSE)'' is a model applicable to 
multi-stage games with observed actions. This model relaxes both the requirements of best
responses and Bayesian inference. Specifically, QCSE is a hybrid model, combining the Agent Quantal Response 
Equilibrium (AQRE) proposed by \cite{mckelvey1998quantal} and the 
Cursed Sequential Equilibrium introduced by \cite{fong2023cursed}.

Consider an assessment $(\mu ,\sigma )$. For any player $i$ and any history
$h^{t-1}$, the \textit{average behavioral strategy profile of $-i$} is defined as:
\begin{equation*}
\bar{\sigma}_{-i}(a_{-i}^{t}|\theta _{i}, h^{t-1})=\sum_{\theta _{-i}\in
\Theta _{-i}}\mu _{i}(\theta _{-i}|\theta _{i}, h^{t-1})\sigma
_{-i}(a_{-i}^{t}| \theta _{-i}, h^{t-1}).
\end{equation*}
In QCSE, players have incorrect perceptions about the behavioral strategies of 
other players. Instead of thinking they are using $\sigma_{-i}$, a $\chi $-cursed 
type $\theta_i$ player $i$ would
believe the other players are using a $\chi $-weighted average of the
average behavioral strategy and the true behavioral strategy:
\begin{equation*}
\sigma _{-i}^{\chi }(a_{-i}^{t}|\theta _{-i},\theta _{i}, h^{t-1})=\chi \bar{%
\sigma}_{-i}(a_{-i}^{t}|\theta _{i}, h^{t-1})+(1-\chi )\sigma
_{-i}(a_{-i}^{t}|\theta _{-i}, h^{t-1}).
\end{equation*}

The beliefs of player $i$ about $\theta _{-i}$ in QCSE are updated 
via Bayes' rule, whenever possible, assuming other players are using
the $\chi $-cursed behavioral strategy rather than the true behavioral
strategy. This updating rule is called the $\chi $\emph{-cursed Bayes' rule}.
Specifically, an assessment satisfies the $\chi $-cursed Bayes' rule if the
belief system is derived from the Bayes' rule while perceiving others are
using $\sigma _{-i}^{\chi }$ rather than $\sigma _{-i}$.

Consider any totally mixed strategy profile $\sigma\in \Sigma^0$. As shown by 
\cite{fong2023cursed}, if the belief system $\mu$ is derived from the 
$\chi $\emph{-cursed Bayes' rule}, then player $i$'s cursed belief is simply a
linear combination of player $i$'s cursed belief at the beginning of that
stage (with $\chi $ weight) and the Bayesian posterior belief (with $1-\chi $
weight). That is, for any $h^t = (h^{t-1}, a^t)$,
\begin{align*}
\mu_i(\theta_{-i}| \theta_i, h^{t}) = \chi \mu_i(\theta_{-i}|\theta_i,  h^{t-1}) 
+ (1-\chi)\left[\frac{\mu_i(\theta_{-i}|\theta_i, h^{t-1})%
\sigma_{-i}(a_{-i}^t|\theta_{-i}, h^{t-1})}{\sum_{\theta^{\prime}_{-i}}\mu_i(%
\theta^{\prime}_{-i}|\theta_i, h^{t-1})\sigma_{-i}(a_{-i}^t|\theta^{%
\prime}_{-i}, h^{t-1})} \right].
\end{align*}

For any player $i$, any $\chi\in[0,1]$, $\sigma \in \Sigma^0$, and type profile $\theta\in \Theta$, 
let $\rho_i^\chi(h^T|\theta, h^t, \sigma_{-i}^\chi,\sigma_i)$ be $i$'s
perceived conditional realization probability of terminal history $h^T \in
\mathcal{H}^T$ at history $h^t \in \mathcal{H}\backslash\mathcal{H}^T$ if
the type profile is $\theta$ and $i$ uses the behavioral strategy $%
\sigma_i$ whereas perceives other players' using the cursed behavioral
strategy $\sigma_{-i}^\chi$. At every non-terminal history $h^t$, a $\chi$%
-cursed player in QCSE will use $\chi$-cursed Bayes' rule to derive the 
posterior belief about the other
players' types. Accordingly, a type $\theta_i$ player $i$'s conditional
expected payoff at history $h^t$ is:
\begin{equation*}
\bar{u}_i(\sigma | \theta_i, h^t) \equiv \sum_{\theta_{-i} \in \Theta_{-i}}
\sum_{h^T \in \mathcal{H}^T} \mu_i(\theta_{-i}|\theta_i, h^t)
\rho_i^\chi(h^T|\theta, h^t, \sigma_{-i}^\chi,\sigma_i) u_i(h^T, \theta_i,
\theta_{-i}).
\end{equation*}
Moreover, let $\bar{u}_i(a,\sigma | \theta_i, h^t)$ be the conditional expected 
payoff of player $i$ of using $a\in A_i(h^{t})$ with probability one, and using $\sigma_i$
elsewhere.

In QCSE, there is a parameter $\lambda\in [0, \infty)$ that governs the precision of choices.
Given an assessment $(\mu, \sigma)$ where $\sigma\in \Sigma^0$ and $\mu$ is derived from the 
$\chi$-cursed Bayes' rule, if $(\mu, \sigma)$ is a QCSE for any player $i$, history $h^t$,
and type $\theta_i$, then type $\theta_i$ player $i$ will have choice probabilities at $h^t$ that 
follow a multinomial logit distribution. In particular, the probability of player $i$ choosing $a\in A_i(h^t)$ is
\begin{align*}
    \frac{e^{\lambda \bar{u}_i(a, \sigma | \theta_i, h^t)}}{\sum_{a' \in A_i(h^t)} e^{\lambda \bar{u}_i(a', \sigma | \theta_i, h^t)}}.
\end{align*}
In summary, for each $\lambda \in [0, \infty)$ and $\chi\in [0,1]$, an assessment $(\mu, \sigma)$ is a QCSE if 
\begin{itemize}
    \item[1.] The belief system is derived from the $\chi$-cursed Bayes' rule, and 
    \item[2.] For any player $i$, type $\theta_i$, history $h^t$ and $a\in A_i(h^t)$, 
    \begin{align*}
    \sigma_i(a|\theta_i, h^t) = \frac{e^{\lambda \bar{u}_i(a, \sigma |\theta_i,  h^t)}}{\sum_{a' \in A_i(h^t)} e^{\lambda \bar{u}_i(a', \sigma | \theta_i, h^t)}}.
\end{align*}
\end{itemize}

When estimating QCSE, constructing the likelihood function follows 
a similar process as described in Appendix \ref{subsecsec:likelihood_functions}.
For each information set $\mathcal{I}_i$, QCSE uniquely predicts the choice probability 
of each $a_i$, denoted as $\bar{Q}(a_i|\mathcal{I}_i, \lambda, \chi)$, given $\lambda$ and $\chi$.
The log-likelihood function can be formed by aggregating over every subject $i$, action $a_i$,
and information set $\mathcal{I}_i$: 
\begin{equation*}
\ln L^{\bar{Q}} (\lambda, \chi) = \sum_{i}\sum_{\mathcal{I}_i \in \Pi_i} \sum_{a_i\in A_i(\mathcal{I}_i)}
\mathds{1}\{a_i,\mathcal{I}_i\}\ln \left[ \bar{Q}(a_i | \mathcal{I}_i, \lambda, \chi)\right].
\end{equation*}

\subsubsection*{Quantal Dynamic Cognitive Hierarchy Solution}

The ``Quantal Dynamic Cognitive Hierarchy Solution (QDCH)'' is a natural extension of DCH, 
where all strategic levels of players make quantal responses instead of best responses. 
In particular, following previous notations, for any $i\in N$, $\tau_i\geq 1$, $\theta\in\Theta$, 
$\sigma$, and $\tau_{-i}$ such that
$\tau_j<\tau_i$ for any $j\neq i$, 
let $P_i^{\tau_i}(h^T | \theta, h^{t-1}, \tau_{-i},
\sigma_{-i}^{-\tau_i},\sigma_i^{\tau_i})$ be level $\tau_i$ player $i$'s belief about the conditional realization probability of $h^T\in\mathcal{H}^T$
at history $h^{t-1}\in \mathcal{H}\backslash
\mathcal{H}^T$ if the type profile is $\theta$,
the level profile is $\tau$, and player $i$ uses 
$\sigma_i^{\tau_i}$. In this case, level $\tau_i$ player $i$'s expected
payoff at any $h^t\in \mathcal{H}\backslash\mathcal{H}^T$ is:
\begin{align*}
&\bar{u}_i^{\tau_i}(\sigma| \theta_i, h^t) \equiv   \\
&\;\;\; \sum_{h^T\in\mathcal{H}^T} \sum_{\theta_{-i}\in \Theta_{-i}}
\sum_{\{\tau_{-i}: \tau_j <k \; \forall j\neq i\}}\mu_i^{\tau_i}(\theta_{-i},\tau_{-i} \mid
\theta_i, h^t) P_i^{\tau_i}(h^T|\theta, h^{t}, \tau_{-i},
\sigma_{-i}^{-\tau_i},\sigma_i^{\tau_i}) 
u_i(h^T,\theta_i,\theta_{-i}).
\end{align*}

Similar to QCSE, in QDCH, there is a parameter $\lambda\in [0, \infty)$ that governs the precision of choices.
Let $\bar{u}_i^{\tau_i}(a, \sigma| \theta_i, h^t) $ be the conditional expected payoff of level $\tau_i$ player $i$ 
of using $a\in A_i(h^t)$ with probability one, and using $\sigma_i^{\tau_i}$ elsewhere. In QDCH, players' choice probabilities 
follow multinomial logit distributions. That is, in QDCH, the probability of level $\tau_i$ player $i$ choosing 
$a\in A_i(h^t)$ at history $h^t$ is 
\begin{align*}
    \sigma_i^{\tau_i}(a|\theta_i, h^t) = \frac{e^{\lambda \bar{u}_i^{\tau_i}(a, \sigma |\theta_i,  h^t)}}{\sum_{a' \in A_i(h^t)} e^{\lambda \bar{u}_i^{\tau_i}(a', \sigma | \theta_i, h^t)}}.
\end{align*}

When estimating QDCH, I assume the prior distribution of levels follows Poisson($\tau$). 
At any information set $\mathcal{I}_i$, let $f(k|\mathcal{I}_i, \lambda, \tau)$ be the posterior 
distribution of levels at information set $\mathcal{I}_i$ given $\lambda$ and $\tau$. In this case, 
the predicted choice probability for $a_i$ at $\mathcal{I}_i$ is the aggregation of 
quantal responses from all levels weighted by the proportion $f(k|\mathcal{I}_i, \lambda, \tau)$:
\begin{equation*}
\bar{\mathcal{D}}(a_i | \mathcal{I}_i, \lambda, \tau) \equiv \sum_{k=0}^{\infty} 
f(k|\mathcal{I}_i, \lambda, \tau) P_k(a_i | \mathcal{I}_i, \lambda, \tau ),
\end{equation*}
where $P_k(a_i | \mathcal{I}_i, \lambda, \tau )$ is the probability of level $k$ players choosing $a_i$ at $\mathcal{I}_i$.
Consequently, the log-likelihood function for QDCH can be formed 
by aggregating over every subject $i$, actions $a_i$ 
and information set $\mathcal{I}_i$:
\begin{equation*}
\ln L^{\bar{D}} (\lambda, \tau) = \sum_{i}\sum_{\mathcal{I}_i \in \Pi_i} \sum_{a_i \in A_i(\mathcal{I}_i)}
\mathds{1}\{a_i,\mathcal{I}_i\}\ln \left[ \bar{\mathcal{D}}(a_i | \mathcal{I}_i, \lambda, \tau)\right],
\end{equation*}
where $\mathds{1}\{a_i,\mathcal{I}_i\}$ is the 
indicator function which is 1 when subject $i$ chooses 
$a_i$ at $\mathcal{I}_i$.

\newpage
\section{Experimental Instructions (For Online Publication)}
\label{sec:instruction_appendix}

\subsection{Sequential Treatment}

\textbf{General Instructions}

\bigskip

\noindent Thank you for participating in the experiment. You are about to take part in a decision-making experiment, in which your earnings will depend partly on your decisions, partly on the decision of others, and partly on chance. 

\bigskip

\noindent The entire session will take place through computer terminals, and all interactions between participants will be conducted through the computers. Please do not talk or in any way try to communicate with other participants during the session. 

\bigskip

\noindent The main task of the experiment consists of 12 matches. Before the main task, you will be asked to complete some comprehension questions. If you have any questions, please raise your hand and the question will be answered so that everyone can hear. 

\bigskip

\noindent In this experiment, you will earn ``points'' in each match. Your earnings will be determined by the total points you earn in the 12 matches. Each point has a value of \$0.02. That is, every 100 points generates \$2 in earnings for you. In addition to your earnings from decisions, you will receive a show-up fee of \$10. At the end of the experiment, your earnings will be rounded up to the nearest dollar amount. All your earnings will be paid in cash privately at the end of the experiment.

\bigskip

\noindent\textbf{Main Task}

\bigskip

\noindent 1. In this experiment, you will be asked to make decisions in 12 matches. You will be randomly matched with another participant into a group for every separate match. This random pairing changes in every match.

\bigskip

\noindent 2. Each match in this experiment corresponds to a game with the following rules. 
\begin{itemize}
    \item At the beginning of each match, each of you and the other participant will be randomly assigned a ``color'' (either Red or White). After the colors are assigned, you will be able to see the color of the other participant who is paired with you. However, you cannot see your own color! 

    \item There are 3 possible situations, and the probabilities of these situations are summarized in the following table. 
    \begin{table}[H]
        \centering
        \begin{tabular}{|c|c|}
        \hline
            Situations & Probabilities \\ \hline
            You are Red and the other participant is White. & $p$ \\ \hline
            You are White and the other participant is Red. & $p$ \\ \hline
            Both of you are Red. & $1-2p$ \\ \hline
        \end{tabular}
        \end{table}
    In other words, there is always at least one Red participant among each group.

    \item Each match is played in rounds. In each match, there are at most 5 rounds. Your color and the other participant’s color are fixed in the match. In each round, you and the other participant will simultaneously choose either ``I’m Red'' or ``wait.'' If both participants choose ``wait,'' then the match will continue to the next round. The match will end:
    \begin{enumerate}
        \item[(1)] after round 5; or
        \item[(2)] after some round where there is at least one participant choosing ``I’m Red.''
    \end{enumerate}
    This round is called the ``terminal round.'' Your payoff for this match depends on which round the terminal round is, your action in the terminal round, and your color. Important: your payoff does not depend on the other participant’s color.

    \item Payoffs:
    \begin{enumerate}
        \item[(1)] If you choose ``wait'' in the terminal round, you will get 0 points for this match regardless of your color.
        \item[(2)] If you choose ``I’m Red'' in the terminal round, your payoff for this match depends on which round the terminal round is and your own color. The payoffs are summarized in the following table. Notice that in each match, you and the other participant will face the same payoff table.
        \begin{table}[H]
            \centering
            \begin{tabular}{|c|c|c|c|c|c|}
            \hline
            Terminal Round & \;\;\;\;\;1\;\;\;\;\; & \;\;\;\;\;2\;\;\;\;\; & 
            \;\;\;\;\;3\;\;\;\;\; & \;\;\;\;\;4\;\;\;\;\; & 
            \;\;\;\;\;5\;\;\;\;\; \\ \hline
            Your payoff if your color is Red & $p_1$ & $p_2$ & $p_3$ & $p_4$ & $p_5$ \\ \hline
            Your payoff if your color is White & $-w_1$ & $-w_2$ & $-w_3$ & $-w_4$ & $-w_5$ \\ \hline
            \end{tabular}
        \end{table}
        \item[(3)] Examples:
        \begin{itemize}
            \item[a.] If you choose ``I’m Red'' in round 3, you will earn $p_3$ points if your color is Red and $-w_3$ if your color is White. 	
            \item[b.]  If you choose ``wait'' in round 4 and the other participant chooses ``I’m Red'' in the same round, you will get 0 points regardless of your color.
        \end{itemize}
    \end{enumerate}
\end{itemize}

\bigskip

\noindent 3. Decisions:
\begin{itemize}
    \item After observing the other participant’s color, you and the other participant matched with you will play the game according to the rules described above. 

    \item Therefore, your payoffs are summarized as below. 
    \begin{table}[H]
        \centering
        \begin{tabular}{|c|c|cc|c|}
        \hline
        &  & \multicolumn{2}{c|}{\begin{tabular}[c]{@{}c@{}}You choose ``I'm Red'' in the \\ terminal round\end{tabular}} & \begin{tabular}[c]{@{}c@{}}You choose ``wait'' in the \\ terminal round\end{tabular} \\ \hline
        Terminal Round & Your color & \multicolumn{1}{c|}{\;\;\;\;\;\; Red \;\;\;\;\;\;} 
        & White & Red or White \\ \hline
        1 &  & \multicolumn{1}{c|}{$p_1$} & $-w_1$ & 0 \\ \hline
        2 &  & \multicolumn{1}{c|}{$p_2$} & $-w_2$ & 0 \\ \hline
        3 &  & \multicolumn{1}{c|}{$p_3$} & $-w_3$ & 0 \\ \hline
        4 &  & \multicolumn{1}{c|}{$p_4$} & $-w_4$ & 0 \\ \hline
        5 &  & \multicolumn{1}{c|}{$p_5$} & $-w_5$ & 0 \\ \hline
        \end{tabular}
    \end{table}

    \item Each match starts from Round 1. You will make your decision in the following screen. 
    \begin{figure}[H]
        \centering
        \includegraphics[width=\textwidth]{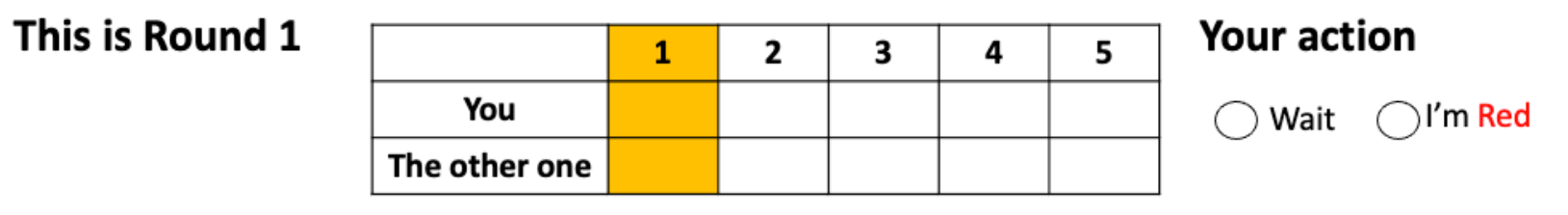}
    \end{figure}
    
    After you make your decision, the following would happen: If either you or the other participant chooses ``I'm Red,'' then this round is the terminal round, and your payoff is determined by your action in this round. However, if both you and the other participant choose ``wait,'' the match continues to the next round, and you will make your decision in the following screen. 
    
    \begin{figure}[H]
        \centering
        \includegraphics[width=\textwidth]{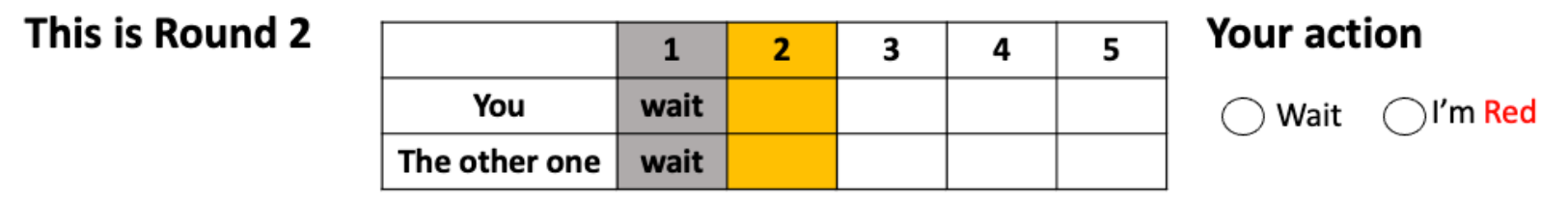}
    \end{figure}
    
    Like the previous round, if either you or the other participant chooses ``I’m Red,'' the match will end after this round. Yet if both you and the other participant choose ``wait,'' the match will proceed to the next round. 

    \item If the game proceeds to round 5, then the match will end after this round and your payoff is determined by your action (and color) in round 5. 
\end{itemize}

\bigskip

\noindent 4. At the end of each match, there will be a summary of the match which includes both of your colors, actions in each round (whenever applicable) and your own payoff for this match. 

\bigskip

\noindent 5. At the beginning of the experiment, you will start from 900 points. You will get paid in cash based on your total points earned from the 12 matches. If your total point is negative, you will only receive the show-up fee.

\bigskip

\noindent 6. Important:
\begin{itemize}
    \item[a.] After each match, you will be randomly paired with another participant in the next match.
    \item[b.] Your color and the other participant’s color will also be randomly re-drawn in each match. The colors in each match are independent of the colors in other matches.
    \item[c.] The probability distribution of colors and the payoff table will change in each match. 
\end{itemize}

\bigskip

\noindent Please raise your hand if you have any questions. The question will be answered so that everyone can hear.

\subsection{Simultaneous Treatment}

\textbf{General Instructions}

\bigskip

\noindent Thank you for participating in the experiment. You are about to take part in a decision-making experiment, in which your earnings will depend partly on your decisions, partly on the decision of others, and partly on chance. 

\bigskip

\noindent The entire session will take place through computer terminals, and all interactions between participants will be conducted through the computers. Please do not talk or in any way try to communicate with other participants during the session. 

\bigskip

\noindent The main task of the experiment consists of 12 matches. Before the main task, you will be asked to complete some comprehension questions. If you have any questions, please raise your hand and the question will be answered so that everyone can hear. 

\bigskip

\noindent In this experiment, you will earn ``points'' in each match. Your earnings will be determined by the total points you earn in the 12 matches. Each point has a value of \$0.02. That is, every 100 points generates \$2 in earnings for you. In addition to your earnings from decisions, you will receive a show-up fee of \$10. At the end of the experiment, your earnings will be rounded up to the nearest dollar amount. All your earnings will be paid in cash privately at the end of the experiment.

\bigskip

\noindent\textbf{Main Task}

\bigskip

\noindent 1. In this experiment, you will be asked to make decisions in 12 matches. You will be randomly matched with another participant into a group for every separate match. This random pairing changes in every match.

\bigskip

\noindent 2. Each match in this experiment corresponds to a game with the following rules. 
\begin{itemize}
    \item At the beginning of each match, each of you and the other participant will be randomly assigned a ``color'' (either Red or White). After the colors are assigned, you will be able to see the color of the other participant who is paired with you. However, you cannot see your own color! 

    \item There are 3 possible situations, and the probabilities of these situations are summarized in the following table. 
    \begin{table}[H]
        \centering
        \begin{tabular}{|c|c|}
        \hline
            Situations & Probabilities \\ \hline
            You are Red and the other participant is White. & $p$ \\ \hline
            You are White and the other participant is Red. & $p$ \\ \hline
            Both of you are Red. & $1-2p$ \\ \hline
        \end{tabular}
        \end{table}
    In other words, there is always at least one Red participant among each group.

    \item Each match is played in rounds. In each match, there are at most 5 rounds. Your color and the other participant’s color are fixed in the match. In each round, you and the other participant will simultaneously choose either ``I’m Red'' or ``wait.'' If both participants choose ``wait,'' then the match will continue to the next round. The match will end:
    \begin{enumerate}
        \item[(1)] after round 5; or
        \item[(2)] after some round where there is at least one participant choosing ``I’m Red.''
    \end{enumerate}
    This round is called the ``terminal round.'' Your payoff for this match depends on which round the terminal round is, your action in the terminal round, and your color. Important: your payoff does not depend on the other participant’s color.

    \item Payoffs:
    \begin{enumerate}
        \item[(1)] If you choose ``wait'' in the terminal round, you will get 0 points for this match regardless of your color.
        \item[(2)] If you choose ``I’m Red'' in the terminal round, your payoff for this match depends on which round the terminal round is and your own color. The payoffs are summarized in the following table. Notice that in each match, you and the other participant will face the same payoff table.
        \begin{table}[H]
            \centering
            \begin{tabular}{|c|c|c|c|c|c|}
            \hline
            Terminal Round & \;\;\;\;\;1\;\;\;\;\; & \;\;\;\;\;2\;\;\;\;\; & 
            \;\;\;\;\;3\;\;\;\;\; & \;\;\;\;\;4\;\;\;\;\; & 
            \;\;\;\;\;5\;\;\;\;\; \\ \hline
            Your payoff if your color is Red & $p_1$ & $p_2$ & $p_3$ & $p_4$ & $p_5$ \\ \hline
            Your payoff if your color is White & $-w_1$ & $-w_2$ & $-w_3$ & $-w_4$ & $-w_5$ \\ \hline
            \end{tabular}
        \end{table}
        \item[(3)] Examples:
        \begin{itemize}
            \item[a.] If you choose ``I’m Red'' in round 3, you will earn $p_3$ points if your color is Red and $-w_3$ if your color is White. 	
            \item[b.]  If you choose ``wait'' in round 4 and the other participant chooses ``I’m Red'' in the same round, you will get 0 points regardless of your color.
        \end{itemize}
    \end{enumerate}
\end{itemize}

\bigskip

\noindent 3. Decisions:
\begin{itemize}
    \item Instead of playing the game round by round, after observing the other participant’s color, you and the other participant are asked to simultaneously choose a ``plan'' which describes how you would commit to play the game if the game were played round by round. After you and the other participant both submit the plans, the computer will implement the plans and your payoff is determined accordingly. 

    \item Since the game ends after some participant chooses ``I’m Red,'' there are six possible plans corresponding to the earliest round you intend to choose ``I’m Red'' or ``always wait.'' Specifically, the six plans are listed below.
    \begin{enumerate}
        \item[$\circ$] \textbf{``I’m Red in Round 1''} means you plan to choose ``I’m Red'' in Round 1.
        \item[$\circ$] \textbf{``I’m Red in Round 2''} means you plan to choose ``wait'' in Round 1 and choose ``I’m Red'' in Round 2. 
        \item[$\circ$] \textbf{``I’m Red in Round 3''} means you plan to choose ``wait'' in Round 1 and Round 2 and choose ``I’m Red'' in Round 3.
        \item[$\circ$] \textbf{``I’m Red in Round 4''} means you plan to choose ``wait'' in Round 1 to Round 3 and choose ``I’m Red'' in Round 4.
        \item[$\circ$] \textbf{``I’m Red in Round 5''} means you plan to choose ``wait'' in Round 1 to Round 4 and choose ``I’m Red'' in Round 5.
        \item[$\circ$] \textbf{``Always wait''} means you plan to choose ``wait'' in Round 1 to Round 5.
    \end{enumerate}

    \item In each match, you will be asked to choose your plan in the following screen.
    \begin{figure}[H]
        \centering
        \includegraphics[width=\textwidth]{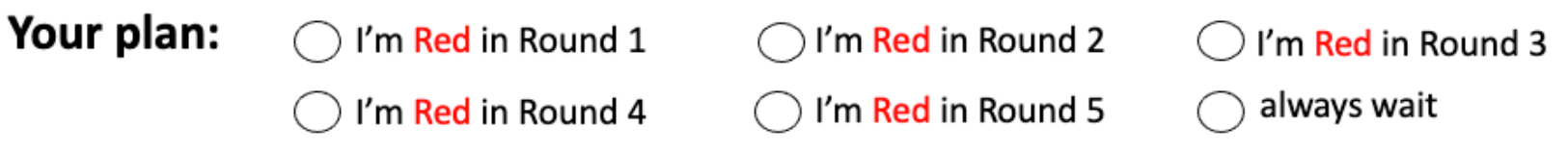}
    \end{figure}

    \item Therefore, your payoffs are summarized as below. 
    \begin{table}[H]
        \centering
        \begin{adjustbox}{max width=\textwidth}
        \begin{tabular}{|c|c|cc|c|}
        \hline
        &  & \multicolumn{2}{c|}{\begin{tabular}[c]{@{}c@{}}You choose ``I'm Red'' \textbf{no later} \\ 
        \textbf{than} the other participant\end{tabular}} & \begin{tabular}[c]{@{}c@{}}You choose ``I'm Red'' \textbf{later} \\ or choose \textbf{``always wait''}\end{tabular} \\ \hline
        Terminal Round & Your color & \multicolumn{1}{c|}{\;\;\;\;\;\;\; Red \;\;\;\;\;\;\;} 
        & White & Red or White \\ \hline
        1 &  & \multicolumn{1}{c|}{$p_1$} & $-w_1$ & 0 \\ \hline
        2 &  & \multicolumn{1}{c|}{$p_2$} & $-w_2$ & 0 \\ \hline
        3 &  & \multicolumn{1}{c|}{$p_3$} & $-w_3$ & 0 \\ \hline
        4 &  & \multicolumn{1}{c|}{$p_4$} & $-w_4$ & 0 \\ \hline
        5 &  & \multicolumn{1}{c|}{$p_5$} & $-w_5$ & 0 \\ \hline
        \end{tabular}
        \end{adjustbox}
    \end{table}
\end{itemize}

\noindent 4. At the end of each match, there will be a summary of the match which includes both of your colors, the terminal round, your action, and your own payoff for this match. If you choose ``I’m Red'' later or at the same round as the other participant, you will be informed the other participant’s exact plan. Otherwise, you will be told that the other participant is ``later than you.''

\bigskip

\noindent 5. At the beginning of the experiment, you will start from 900 points. You will get paid in cash based on your total points earned from the 12 matches. If your total point is negative, you will only receive the show-up fee.

\bigskip

\noindent 6. Important:
\begin{itemize}
    \item[a.] After each match, you will be randomly paired with another participant in the next match.
    \item[b.] Your color and the other participant’s color will also be randomly re-drawn in each match. The colors in each match are independent of the colors in other matches.
    \item[c.] The probability distribution of colors and the payoff table will change in each match. 
\end{itemize}

\bigskip

\noindent Please raise your hand if you have any questions. The question will be answered so that everyone can hear.

\subsection{Screenshots}

Figures \ref{fig:dynamic_decision_screen} and \ref{fig:dynamic_feedback_screen} show the actual screenshots 
of the sequential treatment, and Figures \ref{fig:static_decision_screen} to \ref{fig:static_feedback_screen_2} display the 
actual screenshots of the simultaneous treatment.
Notice that Figure \ref{fig:static_feedback_screen_1} represents the feedback screen of a player who selects ``I'm Red'' earlier 
than the opponent, and Figure \ref{fig:static_feedback_screen_2} provides the perspective from the other player.

\begin{figure}[htbp!]
\renewcommand\thefigure{A.4}
    \centering
    \includegraphics[width=\textwidth]{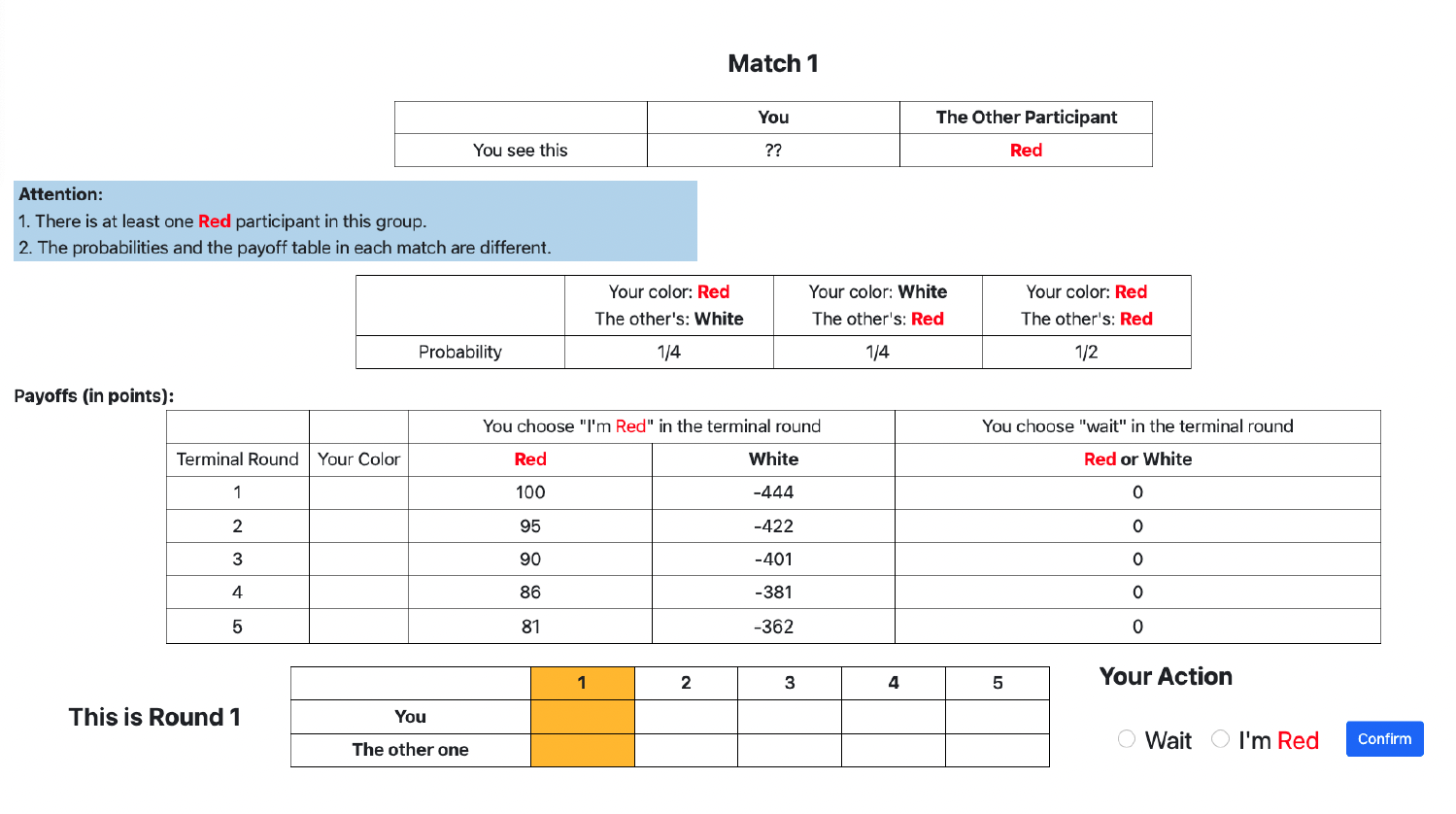}
    \caption{The decision stage of the sequential treatment}
    \label{fig:dynamic_decision_screen}
\end{figure}

\begin{figure}[htbp!]
\renewcommand\thefigure{A.5}
    \centering
    \includegraphics[width=\textwidth]{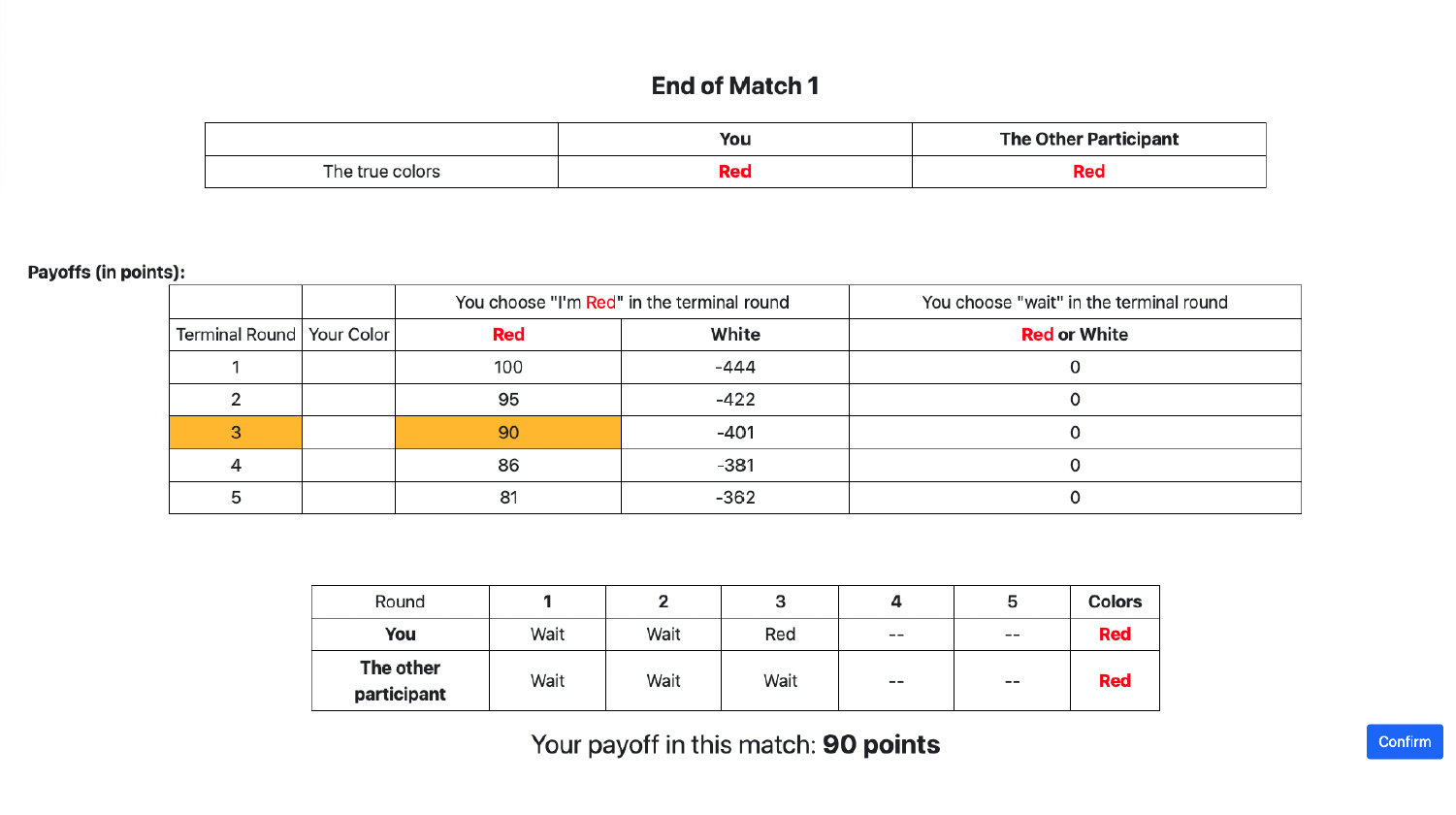}
    \caption{The feedback stage of the sequential treatment}
    \label{fig:dynamic_feedback_screen}
\end{figure}

\begin{figure}[htbp!]
\renewcommand\thefigure{A.6}
    \centering
    \includegraphics[width=\textwidth]{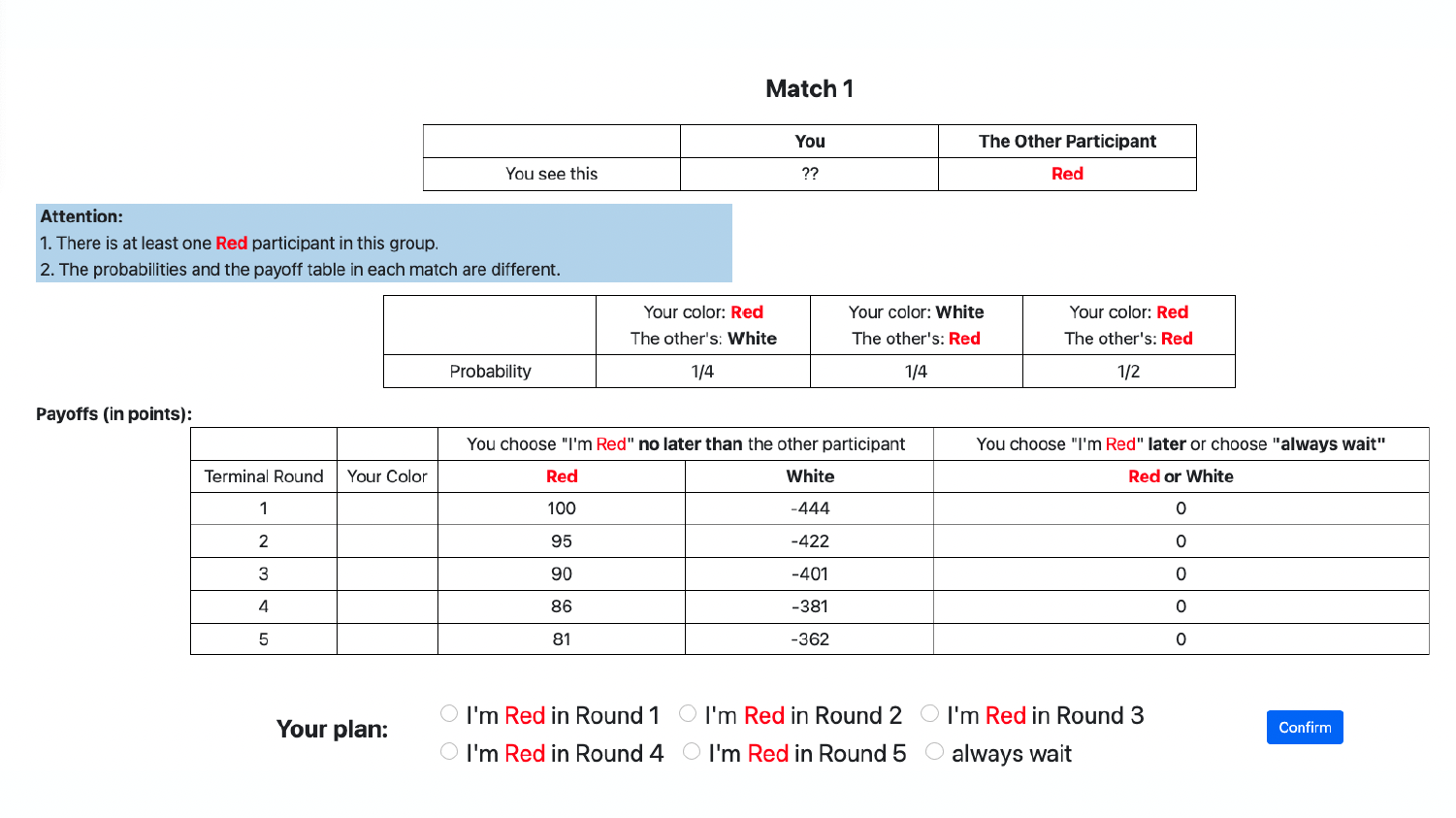}
    \caption{The decision stage of the simultaneous treatment}
    \label{fig:static_decision_screen}
\end{figure}

\begin{figure}[htbp!]
\renewcommand\thefigure{A.7}
    \centering
    \includegraphics[width=\textwidth]{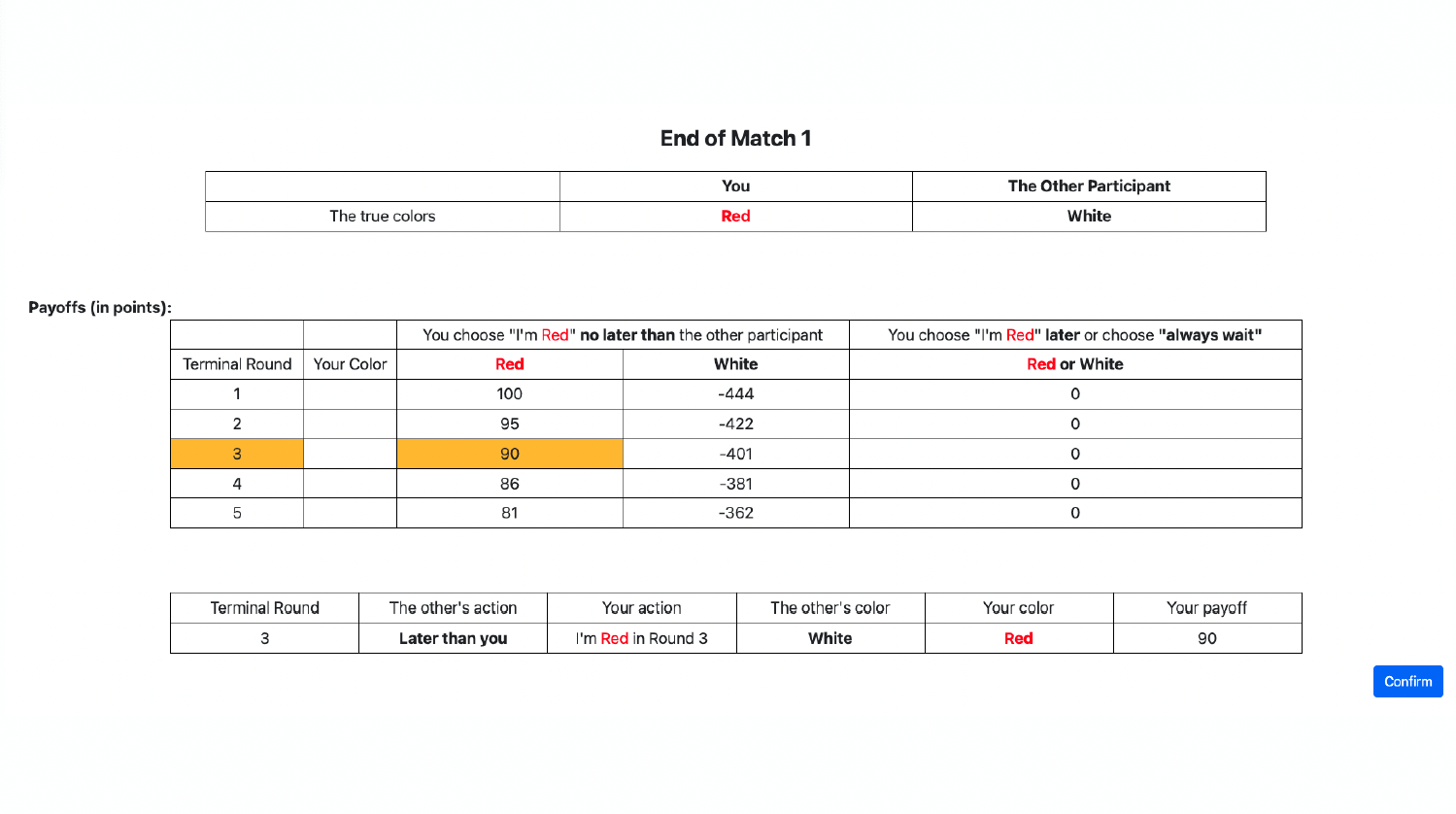}
    \caption{The feedback stage of the simultaneous treatment}
    \label{fig:static_feedback_screen_1}
\end{figure}

\begin{figure}[htbp!]
\renewcommand\thefigure{A.8}
    \centering
    \includegraphics[width=\textwidth]{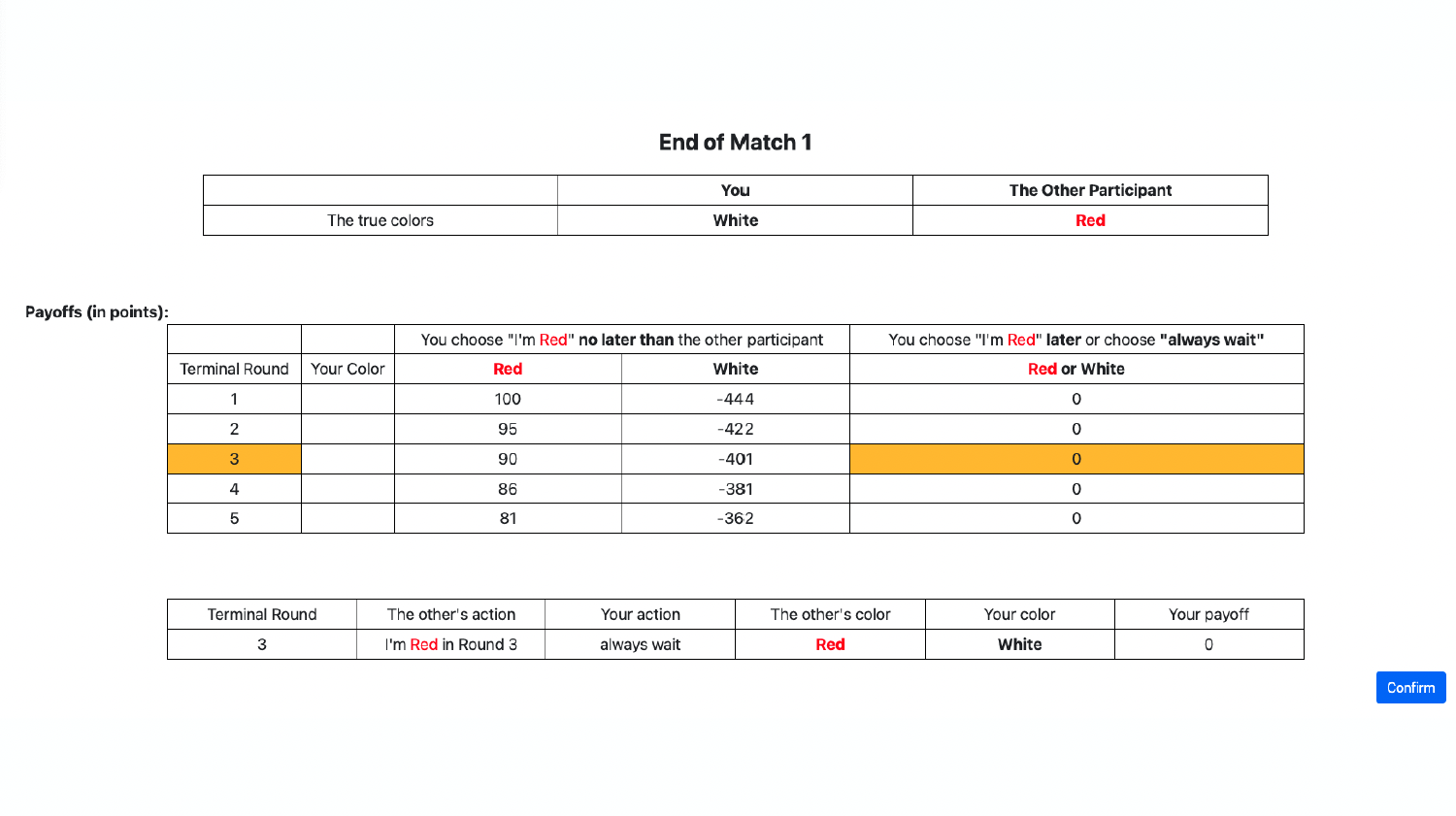}
    \caption{The feedback stage of the simultaneous treatment}
    \label{fig:static_feedback_screen_2}
\end{figure}

\end{document}